\documentclass[11pt,british]{article}
\usepackage[utf8]{inputenc}
\usepackage{mathtools}
\usepackage{amsmath}
\usepackage{amsthm}
\usepackage{amssymb}
\usepackage{setspace}
\usepackage{microtype}
\usepackage{DaveThesis}
\usepackage{babel}
\usepackage{stmaryrd}
\usepackage{slashed}
\usepackage{cancel}
\usepackage{bbm}
\usepackage{booktabs}
\makeatletter
\usepackage{tikz}
\usepackage{tikz-cd}
\usetikzlibrary{backgrounds}
\usepackage[framemethod=tikz]{mdframed}
\AtBeginEnvironment{mdframed}{%
\tikzset{every picture/.style={}}%
}
\mdfsetup{roundcorner=.5ex}

\theoremstyle{definition}
\newtheorem*{defn*}{Definition}

\usepackage{jheppub}                   % jhep style
\usepackage[linecolor=blue,backgroundcolor=blue!25,bordercolor=blue,textsize=scriptsize]{todonotes}
\makeatletter
\gdef\@fpheader{\ }                    % hack the jhep header
\makeatother

\usepackage{setspace}
\usepackage{slashed}	 	% Correct formatting of Feynman slash
\usepackage{amsfonts} 		% AMS fonts
\SetTracking{encoding={*}, shape=sc}{0} 	% Correct spacing of smallcaps
\usepackage{color} 		% Allow coloured links
\definecolor{darkblue}{rgb}{0.0,0.0,0.3} 	% Dark-blue links
\allowdisplaybreaks		% Allow pagebreaks in multiline eqs
\date{\today} 		% Show today's date
\numberwithin{equation}{section}	% Number by section
\makeatletter
\g@addto@macro\bfseries{\boldmath}
\makeatother

\let\originalleft\left
\let\originalright\right
\renewcommand{\left}{\mathopen{}\mathclose\bgroup\originalleft}
\renewcommand{\right}{\aftergroup\egroup\originalright}

\title{Exceptional complex structures and the hypermultiplet moduli of 5d Minkowski compactifications of M-theory}

\author[a]{David Tennyson}
\emailAdd{d.tennyson16@imperial.ac.uk}
\author[a]{and Daniel Waldram}
\emailAdd{d.waldram@imperial.ac.uk}

\affiliation[a]{Department of Physics,
	Imperial College London, \\
	Prince Consort Road, London, SW7 2AZ, UK}
	
\subheader{\hfill\textrm{Imperial/TP/21/DW/2}}
	
\abstract{We present a detailed study of a new mathematical object in $\ER{6}$ generalised geometry called an `exceptional complex structure' (ECS). It is the extension of a conventional complex structure to one that includes all the degrees of freedom of M-theory or type IIB supergravity in six or five dimensions, and as such characterises, in part, the geometry of generic supersymmetric compactifications to five-dimensional Minkowkski space. We define an ECS as an integrable $\Us{6}\times \bbR^{+}$ structure and show it is equivalent to a particular form of involutive subbundle of the complexified generalised tangent bundle $L_{1} \subset E_{\bbC}$. We also define a refinement, an $\SUs{6}$ structure, and show that its integrability requires in addition a vanishing moment map on the space of structures. We are able to classify all possible ECSs, showing that they are characterised by two numbers denoted `type' and `class'. We then use the deformation theory of ECS to find the moduli of any $\SUs{6}$ structure. We relate these structures to the geometry of generic minimally supersymmetric flux backgrounds of M-theory of the form $\bbR^{4,1}\times M$, where the $\SUs{6}$ moduli correspond to the hypermultiplet moduli in the lower-dimensional theory. Such geometries are of class zero or one. The former are equivalent to a choice of (non-metric-compatible) conventional $\SL{3,\bbC}$ structure and strikingly have the same space of hypermultiplet moduli as the fluxless Calabi--Yau case.}

\begin{document}

\maketitle

% \pagebreak

\section{Introduction}

Supersymmetric flux backgrounds of string/M-theory remain of great interest for both phenomenology and the AdS/CFT correspondence. They also provide extensions of conventional geometrical structures that are of mathematical interest in their own right \cite{Hitchin02,Gualtieri04,Grana:2004bg,Pacheco:2008ps,Grana:2011nb,Coimbra:2014uxa,Ashmore:2015joa,Coimbra:2016ydd,Grana:2016dyl,Ashmore:2016qvs}. In the absence of flux, compactifications on manifolds with special holonomy are required in order to preserve supersymmetry. The integrable $G$-structures that describe such backgrounds can be used to understand key properties of the background, such as the massless moduli. Following \cite{Gauntlett:2002sc,Gauntlett:2002fz,Gauntlett:2002nw}, turning the fluxes on breaks the integrability of the conventional $G$-structures, and understanding properties of the geometry becomes much more difficult. 

In this paper we will focus on generic minimal supersymmetric flux compactifications to five-dimensional Minkowski space. Although the geometrical structures we describe will be equally applicable to type II theories, we will focus on M theory backgrounds. They are the natural M-theory extensions of a Calabi--Yau threefold, a geometry that enters many classic problems such as braneworld models \cite{Lukas:1998yy,Lukas:1998tt}, geometrical engineering of five-dimensional supersymmetric field theories \cite{Morrison:1996xf,Douglas:1996xp,Intriligator:1997pq}, or the M-theory interpretation of topological string amplitudes\cite{Gopakumar:1998ii,Gopakumar:1998jq}. When analysed using conventional $G$-structure techniques one finds a local $\SU2$ structure \cite{Gauntlett:2004zh} and the solutions naturally fall into two classes, one that includes the fluxless Calabi--Yau solution and the other corresponding to a back-reacted M5-brane geometry. Our focus will be in particular on the geometry that captures the moduli that fall into hypermultiplet representations in five-dimensions. For the Calabi--Yau case these come from the complex structure moduli and the three-form gauge field potentials, so our structure can be viewed as some extension of conventional complex geometry. 

We will use the formalism of exceptional generalised geometry \cite{Hull:2007zu,Pacheco:2008ps,Coimbra:2011ky} which describes arbitrary type II or M-theory flux backgrounds in terms of geometry on an extended tangent space on which there is a natural action of the exceptional group $\ER{d}$. We will define a new $G$-structure in $\ER{6}$ geometry, referred to as an \emph{exceptional complex structure} (ECS), that is a natural extension of a conventional complex structure or the generalised complex structures of Hitchin and Gualtieri. Its moduli space is quaternionic-K\"ahler and parameterises the hypermultiplet degrees of freedom in the $D=5$ theory. The analogous objects were first introduced in $\ER{7}$ geometry in the context of $D=4$ backgrounds of string and M-theory in \cite{Ashmore:2019qii}, and analogous structures were later found in the $\mathrm{O}(6,6+n)$ geometry relevant to heterotic strings \cite{Ashmore:2019rkx}. This work can be viewed as an extension of these ideas to $D=5$ Minkowski backgrounds. As we will see the ECS is a particularly rich structure that can be viewed as having aspects of both conventional complex and hyperk\"ahler geometries.

Recall that in the language of exceptional generalised geometry, a generic supergravity background is described through a generalised metric, that is, a reduction of the structure group to the maximally compact subgroup $H_{d}\subseteq \ER{d}$ \cite{Hull:2007zu,Coimbra:2011ky,Coimbra:2012af}. If the background also preserves supersymmetry then there is a further reduction of the structure group such that the Killing spinor equations are equivalent to the existence of a torsion-free generalised connection compatible with the reduced structure group \cite{Coimbra:2014uxa,Coimbra:2016ydd}. This rephrases the conditions of a  supersymmetric background with arbitrary flux as the existence of global integrable $G$-structure in generalised geometry in analogy with conventional $G$-structure in the fluxless case\footnote{This is in fact only true for Minkowski backgrounds. AdS backgrounds still have a well-defined global $G$-structure, just now they are weakly integrable meaning that the intrinsic torsion lies in a singlet representation of the $G$-structure \cite{Ashmore:2016qvs,Ashmore:2016oug,Coimbra:2017fqv} }. The $D=5$ backgrounds we discuss in this paper are described by an integrable \emph{global} $\USp{6}$ structure and were studied in \cite{Ashmore:2015joa} where they were dubbed exceptional Calabi--Yau (ECY) spaces. They are defined in terms two sub-structures that have to obey some compatibility conditions: a V-structure, that is invariant under $\mathrm{F}_{4(4)}$ structure, and an H-structure, invariant under $\SUs{6}$ structure \cite{Grana:2009im}. Here, we will find that it is convenient to define an ECS as a slightly weaker version of an H-structure, invariant under $\bbR^{+}\times \Us{6}$. It is equivalent to specifying a particular complex subbundle of the generalised tangent bundle, the analogue of holomorphic tangent bundle $T^{1,0}\subset T_\bbC$ in complex geometry. It is these $\SUs6$ and  $\bbR^{+}\times \Us{6}$ structures that we will study in some detail in this paper.

We find that we can describe a given ECS in terms of two labels which we have called \emph{type}, analogous to the type of generalised complex structures, and \emph{class}. The class will turn out to exactly match the two classes of solutions identified in \cite{Gauntlett:2004zh} using conventional $G$-structures. Supersymmetry means that the space of (non-integrable) $\SUs6$ structures and ECSs are themselves hyperk\"ahler and k\"ahler spaces respectively, and, furthermore, admit analogues of the Hitchin functionals \cite{Hitchin00,Hitchin01} that exist, for example, for conventional complex structures. As we will discuss, supersymmetry can  be regarded as extremisation of the appropriate functional. We derive the most general form of the ECS in each case and use this to find the exact moduli in terms of natural cohomology groups. 

It is well-known that there are general no-go theorems \cite{Maldacena:2000mw,Giddings:2001yu,Gauntlett:2003cy,Gauntlett:2002sc} that exclude compact solutions with non-zero flux (unless one allows sources) so that the only allowed compact background is, in our case, a Calabi--Yau manifold. Thus at first blush our moduli space calculations should be regarded as identifying the hypermultiplet moduli for non-compact backgrounds or alternatively for spaces with boundaries where the sources have been removed. Interestingly, however, it is also possible that the sources enter only in the V-structure equations, such that the $\SUs6$ H-structure remains well defined even at the source. In this case, our expressions would be valid without excising sources.

The paper is structured as follows. In section \ref{sec:D=5 geometry} we review the geometry of $D=5$, $\mathcal{N}=1$ backgrounds in terms of local $\SU{2}$ structures and in terms of exceptional Calabi--Yau structures. Following this, we define ECS as a substructure of the ECY and study the integrability conditions in section \ref{sec: ECS}. In section \ref{sec:classification}, we completely classify ECS and $\SUs{6}$ structures, introducing the notion of \emph{type} and \emph{class}. We further study the integrability conditions, and relate their structure to the local $\SU{2}$ structure of the background. In section \ref{sec:moduli of H-structures} we use the deformation theory of ECS to find the hypermultiplet moduli, following similar results for complex structures \cite{kodaira2006complex}. We first analyse the class 0 and class 1 cases separately and then introduce exceptional Dolbeault operators to analyse the moduli in full generality. Finally, in section \ref{sec: Minkowski Backgrounds}, we consider how these structures apply to Minkowski backgrounds. The appendices are left for conventions and proofs of results in the main text.

\section{Review of \texorpdfstring{$D=5$}{D=5} Minkowski Backgrounds}\label{sec:D=5 geometry}

We would like to analyse the geometry of $\mathcal{N}=1$, $D=5$ flux compactifications of M-theory with the aim of finding their moduli. First, we will briefly review what is already known about these backgrounds, both in terms of local $\SU{2}$ structures in conventional geometry, and in terms of exceptional Calabi--Yau structures in exceptional geometry. This provides the physical context for exceptional complex structures that we study in greater detail in the following sections. The reader that is only interested in the mathematical aspects of these structures can happily skip this section, referring back for notation.

\subsection{Local \texorpdfstring{$\SU{2}$}{SU(2)} Structures}

The generic geometry for the internal manifold for backgrounds of the form $\bbR^{4,1}\times M$ was found in \cite[App D]{Gauntlett:2004zh} where it was given in terms of a local $\SU{2}$ structure. Generically the the supersymmetry parameter can be written in terms of two orthonormal chiral spinors $\eta^{1},\eta^{2}$ as
\begin{equation}\label{eq:Killing Spinor Form}
    \tilde{\eta} = \sqrt{2}(\cos\alpha\, \eta^{1} + \sin \alpha\,(a\eta^{1} + \sqrt{1-|a|^{2}} \eta^{2})^{*}) .
\end{equation}
where $a$ is some complex function with $|a|\leq 1$. Note that there can be points on $M$ where the supersymmetry parameter is chiral. These correspond to having $\sin(2\alpha)=0$, so that one of the chiral terms vanishes and the $\SU2$ structure degenerates to an $\SU3$ structure. If this happens at every point, we have a true global $\SU3$ structure and one finds that supersymmetry is incompatible with the presence of flux \cite{Behrndt:2000zh,Gauntlett:2004zh}. Hence in this case $M$ is a Calabi--Yau manifold and the geometry is well understood. To match \cite{Gauntlett:2004zh}, we will make the following definitions\footnote{Note that our notation slightly differs from \cite{Gauntlett:2004zh} in that their $\xi$ is our $\eta$, their $\zeta$ is our $\theta$, and their $K_{i}$ are our $\zeta_{i}$. }
\begin{equation}
    \epsilon^{+} = \frac{1}{\sqrt{2}}\eta \qquad \epsilon^{-} = -\frac{1}{\sqrt{2}}\ii\gamma^{7}\eta \qquad \theta = \frac{\pi}{2} - 2\alpha
\end{equation}

The local $\SU{2}$ structure can equivalently be described in terms of 1-forms $\zeta_{1},\zeta_{2}$ and three real 2-forms $\omega_{1},\omega_{2},\omega_{3}$ such that the metric and volume form can be written as
\begin{equation}\label{eq:SU2 metric}
    \dd s^{2}(M) = \dd s^{2}_{\SU{2}} + \zeta_{1}^{2} + \zeta_{2}^{2} \qquad \vol = \frac{1}{2}\omega\wedge \omega\wedge \zeta_{1}\wedge \zeta_{2}
\end{equation}
where $\dd s^{2}_{\SU{2}}$ is the 4-dimensional $\SU{2}$ metric given by the $\omega_{i}$. In the volume form, we can take $\omega=\omega_{i}$ for any $i$. There is an alternative description of the local $\SU{2}$ structure through a complex 2-form $\hat{\Omega} = \omega_{2} + \ii \omega_{1}$, and a real 2-form $\hat{\omega} = \omega_{3}$ (following the conventions of \cite{Gauntlett:2004zh}). These objects are defined through bilinears in $\eta_{i}$ which are not globally well-defined objects on $M$. Instead, we should build tensors from bilinears in the global spinors $\epsilon^{\pm}$. A full list of these bilinears in terms of the local $\SU{2}$ structure is given in appendix \ref{app:bilinears}, but here we will simply define the following
\begin{equation}
    \arraycolsep = 1.4pt
    \begin{array}{rclcrclcrcl}
        1 &=& \bar{\epsilon}^{+}\epsilon^{+} = \bar{\epsilon}^{-}\epsilon^{-} & \qquad & \tilde{\Omega} &=& \epsilon^{-\,\mathrm{T}}\gamma_{(2)} \epsilon^{+} & \qquad & \LX &=& \epsilon^{+\,\mathrm{T}}\gamma_{(3)}\epsilon^{+} \\
        \sin\theta &=& \bar{\epsilon}^{+}\epsilon^{-} = \bar{\epsilon}^{-}\epsilon^{+} & & Y &=& -\ii\bar{\epsilon}^{+}\gamma_{(2)}\epsilon^{+} & & V &=& \bar{\epsilon}^{+}\gamma_{(3)}\epsilon^{-} \\
        f=a^{*}\cos\theta &=& \epsilon^{+\,\mathrm{T}}\epsilon^{+} = -\epsilon^{-\,\mathrm{T}}\epsilon^{-} & & f\vol &=& -\ii\epsilon^{+\,\mathrm{T}}\gamma_{(6)} \epsilon^{-} & & \sin\theta \vol &=& - \ii\bar{\epsilon}^{+}\gamma_{(6)} \epsilon^{+}
    \end{array}
\end{equation}
Here the $\gamma$ are the gamma matrices for $\Cliff{6}$ in an orthonormal frame for $M$. There are some other useful spinor bilinear identities we can define
\begin{equation}
    \arraycolsep = 1.4pt
    \begin{array}{rclcrcl}
        \tilde{\zeta}_{1} &=& \bar{\epsilon}^{+}\gamma_{(1)} \epsilon^{+} & \qquad & \ii* \Lambda &=& \epsilon^{-\,\mathrm{T}} \gamma_{(3)} \epsilon^{+} \\
        \tilde{\zeta}_{2} &=& \ii\bar{\epsilon}^{+}\gamma_{(1)}\epsilon^{-} & & *V &=& - \ii\bar{\epsilon}^{+}\gamma_{(3)} \epsilon^{+} \\
        Y' &=& \ii\bar{\epsilon}^{+}\gamma_{(2)} \epsilon^{-} & & Z= *\tilde{\zeta}_{1} &=& \ii \bar{\epsilon}^{+}\gamma_{(5)}\epsilon^{-}
    \end{array}
\end{equation}
Note that, from \eqref{eq:Killing Spinor Form}, the $\SU{2}$ structure degenerates to an $\SU{3}$ structure when $\sin\theta = \pm1$. When this occurs, the form of the metric as given in \eqref{eq:SU2 metric} breaks down. Although $\zeta_{1}$ and $\zeta_{2}$ are not well defined at those points the spinors $\epsilon^\pm$ are, and hence the bilinears defined above are well defined but may degenerate.

The Killing spinor equations put constraints on these tensors. The necessary and sufficient conditions on the tensors for a supersymmetric background were found in \cite{Gauntlett:2004zh}. It is helpful to divide these into two sets, the reason for which we will explain properly in the next section when we introduce $\SUs{6}$ structures. First one has
\begin{equation}\label{eq:Killing spinor eqn}
\arraycolsep = 1.4pt
\begin{array}{rclcrcl}
    \dd(\ee^{3\warp}\sin\theta) &=& 0 &\qquad& \dd(\ee^{3\warp}V) &=& \ee^{3\warp}\sin \theta F \\
    \dd(\ee^{3\warp} f) &=& 0 & & \dd(\ee^{3\warp}\LX) &=& -\ee^{3\warp}f F
\end{array}
\end{equation}
The remaining constraints from the Killing spinor equations are
\begin{equation}\label{eq:Killing spinor eqn b}
%     \dd(\ee^{5\warp}\tilde{\zeta}_{2} ) = *F \qquad 
    \dd(\ee^{\warp} Y') = -\xi\lrcorner F \qquad 
    \dd(\ee^{\warp}Z) = \ee^{\warp}Y'\wedge F
\end{equation}
where $\xi$ is a Killing vector field that preserves all the bilinears, given by
\begin{equation}\label{eq:Killing vector}
    \xi = \ee^{\warp}\tilde{\zeta}_{2}^{\#} 
\end{equation}
Note that generically these equations are together sufficient to imply supersymmetry.\footnote{\label{class1 footnote} When $\sin\theta=f=0$ one must append \eqref{eq:Killing spinor eqn} by conditions relating the flux to the exterior derivatives of $\tilde{\zeta}_1\wedge Y$ and $\tilde{\zeta}_1\wedge\tilde{\Omega}$ as given in \eqref{eq:class 1 inv}-\eqref{eq:class 1 flux 2}.} 

As was noted in \cite{Gauntlett:2004zh}, the full set of equations imply that $f=0$. Since the warp factor $\ee^{2\warp}$ cannot vanish, the first equation in \eqref{eq:Killing spinor eqn} then splits the solutions into two classes depending on whether globally $\sin\theta \neq 0$ or $\sin\theta = 0$. Physically, these two classes correspond to fluxed backgrounds that live in the same family as the Calabi--Yau solution and backgrounds corresponding to the back-reacted geometry around an M5 brane transverse to $M_{\text{HK}}\times \bbR$, where $M_{\text{HK}}$ is a four-dimensional hyperk\"ahler manifold, and where a flat direction of the M5 brane fibers over $M_{\text{HK}}\times \bbR$ base\footnote{We can also consider compactifying $\bbR$ to $S^{1}$.}. We denote these 
\begin{itemize}
    \item[(i)] class 0 ($\sin\theta\neq 0$): flux-deformed Calabi--Yau space
    \item[(ii)] class 1 ($\sin\theta=0$): back-reacted wrapped M5-brane geometry
\end{itemize}
These same two cases will arise naturally in our analysis of generalised $\SUs{6}$ structures. In that case, integrability of the structure is not sufficient to set $f=0$ and so we will define class 1 by $\sin\theta=f=0$ and class 0 as the complement. 

\subsection{Exceptional Calabi--Yau Structures}\label{sec: ECY review}

$\mathcal{N}=1$, $D=5$ flux backgrounds have been studied in context of exceptional generalised geometry for $\bbR^{4,1}$ backgrounds in \cite{Coimbra:2014uxa,Ashmore:2015joa,Coimbra:2016ydd} and for AdS backgrounds in \cite{Ashmore:2016qvs,Coimbra:2017fqv}. These were respectively called exceptional Calabi--Yau structures (ECY), and exceptional Sasaki-Einstein structures (ESE) because they generalised the respective manifolds to arbitrary flux.\footnote{Closely related structures also exist in other dimensions \cite{Ashmore:2015joa,Ashmore:2016qvs}.} Both ECY and ESE structures are described by a global $\USp{6}\subset \ER{6}$ structure, but differ in their integrability conditions. ECY are defined by an entirely integrable $\USp{6}$ structure \cite{Coimbra:2014uxa}, i.e. a global $\USp{6}$ structure with a torsion-free compatible connection, while ESE backgrounds are defined to have weak generalised holonomy \cite{Coimbra:2017fqv} meaning the intrinsic torsion lies in a singlet representation. In this paper, we will focus on Minkowski backgrounds and will review the definition and integrability of ECY here. Details of $\ER{6}$ geometry, including expressions for the generalised tangent bundle, the adjoint bundle, and the adjoint action, can be found in appendix \ref{app:conventions}.

As with conventional geometry, we would like to describe integrable $G$-structures in terms of generalised tensors stabilised by the group $G$, and differential conditions on those tensors. This would be the equivalent of the existence and closure conditions of the K\"ahler form $\omega$ and the holomorphic 3-form $\Omega$ for an $\SU{3}$ structure. In \cite{Ashmore:2015joa} it was shown that a $\USp{6}$ structure is defined by the combination of what was called an H-structure and a V-structure, satisfying some compatibility conditions and differential conditions. In the effective five-dimensional theory, the H-structure is be related to the scalars in the hypermultiplets, while the V-structure is related to the scalars in the vector multiplets, hence the nomenclature\footnote{In fact, in \cite{Ashmore:2015joa}, they introduce these structures for compactifications down to 4, 5, and 6 dimensional Minkowski space preserving 8 supercharges. The H and V-structures for compactifications down to 4 dimensions were first introduced in \cite{Grana:2009im}.}. 

For M theory compactifications, the H-structure generalises the notion of complex structure on a Calabi--Yau manifold. As we will see, in some ways it is also analogous to a hyperk\"ahler geometry, and so can more generally be viewed as an object that interpolates between the two. It is defined by a triplet of weighted adjoint valued tensors $J_{\alpha}\in \Gamma((\det T^{*})^{1/2}\otimes \ad \tilde{F})$, $\alpha = 1,2,3$. These have to form a highest root $\su{2}$ algebra of $\e{6(6)}$. In particular we require
\begin{equation}\label{eq:H condition 1}
    [J_{\alpha},J_{\beta}] = 2\kappa\epsilon_{\alpha\beta\gamma}J_{\gamma} \qquad \Tr(J_{\alpha}J_{\beta}) = -\kappa^{2}\delta_{\alpha\beta}
\end{equation}
Where $\kappa$ is some section of $(\det T^{*})^{1/2}$. Alone, these tensors define an $\SUs{6}$ structure, where $\SUs{6}$ is a particular non-compact real form\footnote{It can be identified as the following subgroup of $\SL{6,\bbC}$. If $J$ is an antisymmetric $6\times 6$ matrix such that $J^{2}=-\mathrm{Id}$, then $U\in \SUs{6}$ with $U^{*}$ its complex conjugate if and only if $UJ = JU^{*}$.} of $\SL{6,\bbC}$  \cite{GUNAYDIN1985309,gilmore2008lie}. This $G$-structure is integrable if and only if the following generalised one-forms vanish.
\begin{equation}\label{eq:H condition 2}
    \mu_{\alpha}(V) := -\frac{1}{2}\epsilon_{\alpha\beta\gamma}\int_{M} \Tr(J_{\beta} L_{V}J_{\gamma}) \overset{!}{\equiv} 0 
\end{equation}
The objects $\mu_{\alpha}$ can be viewed as moment maps for the action of generalised diffeomorphisms on the infinite-dimensional hyperk\"ahler space of H-structures $\ZSUs$. We will expand a little more in the next section on how this works. In fact, the space $\ZSUs$ is a hyperk\"ahler cone over a quarternionic-K\"ahler base where the $\bbH^{*}= \SU{2}\times \bbR^{+}$ of the cone direction is precisely parameterised by the $\{J_{\alpha}, \kappa\}$ defining the $\SUs{6}$ structure. This fact will be important when we analyse the moduli space in section \ref{sec:moduli of H-structures}.

For a generic background ($\sin\theta \neq 0$), the vanishing of the moment maps above is equivalent to the set of Killing spinor equations given in \eqref{eq:Killing spinor eqn}. The remaining constraints \eqref{eq:Killing spinor eqn b}, \eqref{eq:Killing vector} come from the differential conditions on the V-structure and certain compatibility conditions that we will lay out below. In the non-generic case ($\sin\theta =0$), the picture is slightly more subtle and the vanishing of the moment maps above implies some extra conditions (see footnote \ref{class1 footnote}). 
The V-structure in M theory compactifications generalises the notion of symplectic structure on the Calabi--Yau. It is defined by a single generalised vector $K\in \Gamma(E)$ that satisfies
\begin{equation}\label{eq:V condition 1}
    c(K,K,K) > 0 %  \kappa^{2}
\end{equation}
where $c:S^{3}E\rightarrow \det T^{*}$ is the cubic invariant of $\E{6}$ and we have fixed an orientation on $M$ to define the inequality. This describes an $F_{4(4)}$ structure which is integrable if
\begin{equation}\label{eq:V condition 2}
    L_{K}K = 0
\end{equation}
% Again, this can be interpreted as the vanishing of a different moment map for the action of generalised diffeomorphisms on the K\"ahler moduli space of V-structures $\mathcal{M}_{V}$.

Finally, the $\USp{6}$ structure is defined by an H-structure and a V-structure obeying an additional compatibility and integrability conditions. These are
\begin{equation}\label{eq:ECY condition 1}
    J_{\alpha}\cdot K = 0 \qquad c(K,K,K) = \kappa^2 \qquad
    L_{K}J_{\alpha} = 0
\end{equation}
The first two compatibility conditions ensure that the stabiliser group of the two structures is $F_{4(4)}\cap \SUs{6} = \USp{6}$, and the extra differential condition is required to ensure the intrinsic torsion of the $\USp{6}$ structure completely vanishes.

As mentioned, these structures describe the geometry of arbitrary $\mathcal{N}=1$ backgrounds with $\bbR^{4,1}$ external space. Hence, we should be able to embed the results of the previous section into this language. This was done in \cite{Ashmore:2016qvs} for AdS backgrounds, and can be extended to Minkowski with the following identifications.
\begin{align}
    J_{3} &= \tfrac{1}{2}\kappa \left( -Y_{R} +(*V - *V^{\#}) + \sin\theta(\vol + \vol^{\#}) \right) \label{eq:general H 1} \\
    J_{+} &= \tfrac{1}{2}\kappa\left(\tilde{\Omega}_{R} - (\ii*\LX - \ii *\LX^{\#}) - \ii f(\vol + \vol^{\#})\right) \label{eq:general H 2} \\
    K &= \xi - \ee^{\warp}Y' +\ee^{\warp}Z \label{eq:general V}
\end{align}
where $J_{+}=J_{1}+\ii J_{2}$, $\kappa^{2} = \ee^{3\warp}\vol$ and the subscript $R$ denotes raising one index with the metric to creating a $\GL{6}$ adjoint element. We can see from these expressions that the triplet $J_{\alpha}$ are defined from a triplet of scalars, a triplet of 2-forms and a triplet of 3-forms. These are respectively given by $\{\sin\theta, \re f, \im f\}$, $\{Y, \re \tilde{\Omega}, \im \tilde{\Omega}\}$, and $\{V, \re \Lambda, \im \Lambda\}$.

In \eqref{eq:general H 1}-\eqref{eq:general V} we have not explicitly included the flux field gauge potentials of the geometry, which can be viewed as `twisting' the generalised tensors by an $\ER6$ element $\ee^{A+\tilde{A}}$, or, as we will mostly use in this paper, by modifying the Dorfman derivative by flux-dependent terms. (Note that $A$ is the three-form potential for the four-form flux $F=\dd A$ on $M$, while $\tilde{A}$ is the six-form potential giving a dual description of a form-form flux on the non-compact $\bbR^{4,1}$. Lorentz invariance implies the later has a trivial field-strength.) The authors of \cite{Ashmore:2016qvs} also showed that \eqref{eq:H condition 1}--\eqref{eq:ECY condition 1} precisely reproduce the algebraic conditions required for an $\SU{2}$ structure, along with the differential conditions for supersymmetry \eqref{eq:Killing spinor eqn}, \eqref{eq:Killing spinor eqn b}. In several places in this paper, we will use the very concrete special case of a  Calabi--Yau manifold case as an example, as defined in section \ref{sec:CY example}. 
%We will only return to these full expressions in section \ref{sec: Minkowski Backgrounds} when we link our results to hypermultiplet moduli.

% An obvious point is that \eqref{eq:general H 1} and \eqref{eq:general H 2} are of course not the only way of writing the hypermultiplet structure, that is, of identifying $J_3$ and $J_+$. Any global $\SU2$ rotation on the $J_\alpha$ defines an equivalent object. This will be particularly relevant when we define exceptional caomplex structures in the next section. 

\section{Exceptional Complex Structures}\label{sec: ECS}

A central focus of this paper is to give a general analysis of the moduli space of integrable $\SUs6$ structures. In terms of the reduction on a supersymmetric space, this space encodes the massless hypermultiplet degrees of freedom of an on-shell background and as such should be finite-dimensional. It turns out that rather than working directly with the $\SUs6$ structure on can rephrase the problem in terms of a slightly weaker notion of an exceptional complex structure (ECS). These objects, their integrability conditions, classification and relation to $\SUs6$ structures is the subject of this section. 

\subsection{\texorpdfstring{$\SUs{6}$}{SU*(6)} and \texorpdfstring{$\bbRpl\times\Us{6}$}{RxU*(6)} Structures}\label{sec:SUs(6) structures}

Much as one can study hyperk\"ahler geometries by focusing on one particular complex structure, we can study the geometry of $\SUs{6}$ structures by restricting to a single $J_{\alpha}$, which we will denote as $J_{3}$. The $\SU2$ action on the triplet $J_\alpha$ means that, as in the hyperk\"ahler case, there is an $S^2\simeq \bbC \mathbb{P}^{1} \simeq \SU{2}\quotient \U{1}$ of such choices. In fixing one, we find a structure analogous to the ECS that was used to describe four-dimensional $\mathcal{N}=1$ M-theory and type II geometries in \cite{Ashmore:2019qii}, and the corresponding heterotic backgrounds in \cite{Ashmore:2019rkx} backgrounds. As in those cases, we show that integrability is naturally defined in terms of of an involutive subbundle of the generalised tangent bundle. The moduli space of ECSs is generically infinite-dimensional. However, if the ECS arises from an integrable $\SUs6$ structure, as we show in section \ref{sec:E6 generic moduli}, there is a natural way to interpret the moduli space of the $\SUs6$ structure in terms of the ECS. 

% In this section, we will give a definition of an ECS in $\ER{6}$ geometry and see how we can reinterpret the $\SUs{6}$ structures using such an objects. To do so, we will break the explicit internal $\SU{2}$ symmetry of the H-structure. We will return to this point in section \ref{sec: SU(2) symmetry}.

First note that to define an $\SUs{6}$ structure, it is sufficient to just define $J_{+}$. Indeed we can then obtain the full $J_{\alpha}$ via
\begin{equation*}
    J_{-} = \bar{J}_{+} \qquad J_{3} = (-8\Tr(J_{+}J_{-}))^{-1/2}\,\ii\,[J_{+},J_{-}]
\end{equation*}
Recall that the space of H-structures $\ZSUs$ is hyperk\"ahler. Parameterising point on $\ZSUs$ by a choice of $J_{-}$, picks out a particular complex structure on $\ZSUs$, such that $\Jpl := \kappa J_{-}\in \Gamma(\det T^{*} \otimes \ad \tilde{F})$ is a holomorphic coordinate\footnote{This was first noticed in collaboration with Edward Tasker.}. From these tensors we can define two reductions of the structure group
\begin{equation}
    \begin{aligned}
    \text{$\SUs{6}$ structure :} && \Jpl &:= \kappa J_{-} \in \Gamma(\det T^{*} \otimes \ad \tilde{F}) \\
    \text{$\bbRpl\times \Us{6}$ structure :} && \JUs &:= \kappa^{-1}J_{3} \in \Gamma(\ad \tilde{F})
    \end{aligned}
\end{equation}
so that $\JUs$ is the unweighted $J_{3}$. To see that $\JUs$ defines an $\bbRpl\times \Us{6}$ structure we note that, since it is unweighted, it is invariant under the $\bbRpl\subset \ER{6}$ action. The additional $\U{1}$ symmetry comes from the action generated by $\JUs$ itself. Alternatively, we can define the $\bbR^{+}\times \Us{6}$ structure more directly from the supersymmetry.
\begin{defn}
Let $\mathfrak{t}\in\su2\subset\e{6(6)}$ be the Lie algebra element that generates $\U1\subset\ER6$ where $\su2$ is a highest root subalgebra, and we normalise $\Tr\mathfrak{t}^2=-1$. All such elements lie in the same adjoint orbit $\mathcal{O}$. An $\bbR^{+}\times \Us{6}$ structure is a smooth section $\JUs\in \Gamma(\ad\tilde{F})$ such that $\JUs|_p$ lies in $\mathcal{O}$ for every point $p\in M$.  
\end{defn}
Given such a $\JUs$, one can use it to decompose the generalised tangent bundle into eigenbundles. We find that
\begin{align}
    \begin{split}
        E_{\bbC} &= L_{1}\oplus L_{-1} \oplus L_{0} \\
        \rep{27} &\rightarrow \rep{6}_{1} \oplus \rep{6}_{-1} \oplus \repb{15}_{0}
    \end{split} \label{eq:E6 GTB decomposition}
\end{align}
In the second line we have expressed this decomposition in terms of $\Us{6}$ representations. Here the subscripts denote the charge under the $\U{1}\subset \Us{6}$ generated by the $\JUs$. Much as for conventional almost complex structures and almost generalised complex structures, we can also define the $\bbR^{+}\times \Us{6}$ structures purely in terms of $L_{1}$:
\begin{defn}\label{def:E6 ECS}
An $\bbRpl\times\Us{6}$ structure is defined by a subbundle $L_{1}\subset E_{\bbC}$ such that
\begin{itemize}
    \item[i)] $\dim_{\bbC} L_{1} = 6$
    \item[ii)] $L_{1}\proj{N} L_{1} = 0$
    \item[iii)] $L_{1}\cap \bar{L}_{1} = \{0\}$, \quad
    $L_{1}\cap L_{0} = \{0\}$
    \item[iv)] The map $\zeta:L_{1}\times (L_{-1})^{*} : \rightarrow \bbR$ defined by
    \begin{equation}
        \zeta(V,Z) = \Tr\left((V\times_{\ad}Z)(\bar{V}\times_{\ad}\bar{Z})\right)
    \end{equation}
    is negative $\forall\, V\in L_{1}$, $Z\in (L_{-1})^{*}$
\end{itemize}
We call such structures \emph{almost ECS}. Any bundle obeying the first two conditions is called an \emph{almost exceptional Dirac structure} in analogy with \cite{Gualtieri04}.
\end{defn}
\noindent Here, $N \subset S^{2}E$ is a particular bundle transforming in the $\repb{27_{2}}$ of $\ER{6}$, and $\times_{N}$ represents the projection onto that bundle\footnote{The decomposition of $N$ into natural geometric bundles, as well as the projection map, are given in appendix \ref{app:conventions}.}. Note that we could equally well define the structure in terms of $L_{-1}$. While conditions (iii) and (iv) appear to depend on the full decomposition \eqref{eq:E6 GTB decomposition}, one can define $L_{0}$ from $L_{1}$ via the following. Let $A = \{Z\in E^{*}\,|\, \left<V,Z\right> = 0 \; \forall \,V\in L_{1}\oplus L_{-1}\}\subset E^{*}$, where $\left<\cdot,\cdot\right>$ is the natural pairing between $E$ and $E^{*}$. That is, $A$ is the null space of $L_{1}\oplus L_{-1}$. Then we define
\begin{equation}
    L_{0} = (L_{1}\times_{\ad}A)\cdot L_{-1}
\end{equation}
Once we have found $L_{0}$ such that (iii) holds, we have a well-defined splitting of the dual space $E^{*}$ into $(L_{\pm1})^{*}$ and $(L_{0})^{*}$. Hence (iv) is well-defined.

We can also decompose the weighted adjoint bundle into eigenbundles of $\JUs$. We find
\begin{align}
    \begin{split}
        \rep{78}\oplus \rep{1} \rightarrow \rep{1}_{+2}\oplus \rep{1}_{-2} \oplus \rep{20}_{+1}\oplus \rep{20}_{-1} \oplus \ad P_{\bbRpl\times \Us{6}}
    \end{split} \label{eq:adjoint decomp under SUs6}
\end{align}
The singlets imply that an $\bbRpl\times\Us{6}$ structure defines a line bundle 
\begin{equation}
\label{eq:UJbundle}
    \mathcal{U}_{\JUs}\subset (\det T^{*})\otimes \ad \tilde{F}_\bbC 
       \simeq \bbC \oplus \ext^3 T^*_\bbC 
          \oplus (T_\bbC^*\otimes\ext^5 T^*_\bbC)
          \oplus (\ext^3 T^*_\bbC \otimes \ext^6 T^*_\bbC)
          \oplus (\ext^6 T^*_\bbC)^{\otimes 2}
\end{equation}
One can show that it is fixed by requiring
\begin{equation}
    V \bullet \Jpl = 0 \quad \forall \, V \in \Gamma(L_{1}), 
    % \qquad \Tr(\Jpl\bar{\Jpl}) \neq 0
\end{equation}
where $\Jpl$ is a local section of $\mathcal{U}_{\JUs}$. The product $V\bullet \Jpl$ is defined by the projection $E\otimes (\det T^{*})\otimes \ad \tilde{F} \rightarrow C$ where $C$ is the generalised tensor bundle transforming in the $\rep{351}_{\rep{4}}$ of $\ER{6}$\footnote{Note that this is $R_{4}$ in the tensor hierarchy of $\E{6}$ \cite{Berman:2012vc}}. One can equally define a local section $\Jpl$ by the condition $\JUs\cdot \Jpl = [\JUs,\Jpl] = 2\ii\Jpl$. Furthermore, one can show that for any non-zero $\Jpl$, one has that $\Tr(\Jpl\bar{\Jpl})$ is negative, where we recall that $\Tr(\Jpl\bar{\Jpl})$ is a section of $(\det T^{*})^{2}$ which has a canonical orientation and hence a well defined notion of a negative section. One can use the converse to reformulate of the last condition in definition~\ref{def:E6 ECS} as
\begin{itemize}
    \item[iv')] Given any non-zero local section $\Jpl$ of the line bundle $\mathcal{U}_\JUs\subset (\det T^{*})\otimes \ad \tilde{F}$, defined by $L_1$, one has $\Tr(\Jpl\bar{\Jpl}) < 0$
\end{itemize}
As we will see later one can view this requirement as an generalisation of the notion of stability for three-forms defining an $\SL{3,\bbC}$ structure introduced by Hitchin \cite{Hitchin00,Hitchin01}. 

We can then give the definition
\begin{defn}
Given an almost ECS $\JUs$ with trivial line bundle $\mathcal{U}_{\JUs}$, an $\SUs{6}$ structure is a global non-vanishing section of $\mathcal{U}_{\JUs}$.
\end{defn}
\noindent 
All $\SUs{6}$ structures will arise in this way and any two will be related by some $\ER{6}$ transformation. Note that any two $\SUs{6}$ structures $\Jpl$ and $\Jpl'$ which define the same $\bbRpl\times \Us{6}$ structure will be related by some non-vanishing function $h$
\begin{equation}\label{eq:equiv SU*(6) for U*(6)}
    \Jpl' = h\Jpl
\end{equation}
Much like for complex structures and generalised complex structures, suitable values for $\Jpl$ do not fill out the whole of $\rep{78}\oplus\rep{1}$. Instead, they exist in some particular $\ER{6}$ orbit.

It is important to reiterate that since $\SUs{6}\not\subset \USp{8}$, a choice of $\Jpl$ does not define a generalised metric and hence does not fully define a supergravity background. To do so, one needs to also specify a compatible V-structure $K$. Nonetheless, the space of $\SUs6$ structures encodes the moduli space of hypermultiplets in the effective theory that arises on compactification on $M$, and so, even by itself encodes important information. 

\subsection{Involutivity, Moment Maps, and Integrability}

We will now look at the conditions imposed on $\Jpl$ by integrability. Recall from section \ref{sec: ECY review} that the integrability of the $\SUs{6}$ structure is a subset of the supersymmetry conditions. This was given as the vanishing of a triplet of moment maps for the action of generalised diffeomorphisms on the space of H-structures $\ZSUs$. We will see that we will be able to recast this as an involutivity condition on $L_{1}$, which gives the integrability of the ECS (roughly corresponding to the vanishing of the moment maps $\mu_+=\mu_1+\ii\mu_2=0$), along with the vanishing of just a single moment map ($\mu_3$). In using this description for integrable H-structures, one drops the explicit hyperk\"ahler structure on $\ZSUs$ that is guaranteed by supersymmetry, and instead views $\ZSUs$ as a K\"ahler space. For each $\SUs6$ structure there is an $S^2$ of integrable ECS and corresponding moment maps which will all give equivalent integrability conditions, a point we will return to in the following section. 

The intrinsic torsion for the generalised structures lies in a subbundle of the torsion bundle $W$ which transforms under $\ER6$ as $\repb{27}_{-1}\oplus\rep{351}_{-1}$. Decomposing into $\Us{6}$ representations we find that they transform as \cite{Ashmore:2015joa}
\begin{align}
    W_{\tint}^{\SUs{6}} : & \quad \rep{15}_{2}\oplus \rep{15}_{-2} \oplus \rep{15}_{0} \oplus \repb{6}_{1}\oplus \repb{6}_{-1} \label{eq:SUs6 intrinsic torsion} \\
    W_{\tint}^{\bbRpl\times \Us{6}} : & \quad \rep{15}_{2}\oplus \rep{15}_{-2} \label{eq:Us6 intrinsic torsion}
\end{align}
where again the subscripts denote the $\U{1}$ charge. We define the integrable structures as those with vanishing intrinsic torsion. Just as the integrability of complex structures (and generalised complex structures) is given by involutivity of an eigenbundle, we find
\begin{defn}\label{def:integrable E6 ECS}
An integrable $\bbRpl\times \Us{6}$ structure, or \emph{ECS}, is an almost ECS that is involutive under the Dorfman derivative. That is
\begin{equation}\label{eq:Us(6) integrability}
    L_{V}W \in \Gamma(L_{1}) \qquad \forall \, V,W\in \Gamma(L_{1})
\end{equation}
In analogy with generalised complex geometry, we refer to an involutive structure that does not satisfy $L_{1}\cap L_{-1} = \{0\}$ as an \emph{exceptional Dirac structure}.
\end{defn}
\noindent In general, $L_{V}W \neq \llbracket V,W \rrbracket$. However, the definition of the Dorfman derivative is such that
\begin{equation}
    L_{V} W + L_{W}V \sim \dd\times_{E} (V\times_{N} W)
\end{equation}
which clearly vanishes by property (ii) of definition \ref{def:E6 ECS}. Hence, it is equivalent to determine integrability with respect to the Courant bracket.

To see that this definition is correct, one can introduce a compatible connection $D$ that is not necessarily torsion free. Generalised torsion is defined such that
\begin{align}
    L_{V}W &= L_{V}^{D}W - T(V)\cdot W \label{eq:torsion}
\end{align}
where $L^{D}_{V}$ is the Dorfman derivative where each instance of $\del$ is replaced with $D$. Because of the compatibility of $D$, the first term must be a section of $L_{1}$. Moreover, since the left hand side does not depend on the choice compatible connection, the projection of $L_{V}W$ onto $L_{0}\oplus L_{-1}$ can only depend on the intrinsic torsion $T^{\tint}$. Given \eqref{eq:Us6 intrinsic torsion} we can see that
\begin{equation}
    T^{\tint}(V)\cdot W \in \Gamma(L_{0})
\end{equation}
Hence, for $L_{V}W$ to be a section of $L_{1}$ for all $V,W\in \Gamma(L_{1})$, we need that $T^{\tint}|_{\rep{15}_{-2}}=0$. The complex conjugate of this condition then sets the whole of $T^{\tint}=0$. We then see that definition \ref{def:integrable E6 ECS} is correct.

From \eqref{eq:SUs6 intrinsic torsion} we can see that the integrability conditions for $\bbRpl\times \Us{6}$ structures is a subset of the conditions for an integrable $\SUs{6}$ structure. This makes sense as the $\bbRpl\times \Us{6}$ structure is a strictly weaker structure. Given $\Jpl$ defining an $\bbRpl\times \Us{6}$ structure that is integrable, we want to know what additional conditions are required so that the $\SUs{6}$ structure is also integrable. Let us consider just the map $\mu_{3}$ from \eqref{eq:H condition 2}. In \cite{Ashmore:2015joa}, it is shown that this is equal to
\begin{equation}
    \mu_{3}(V) \propto \int_{M} \Tr(\kappa\, J_{3}T^{\tint}(V)) + \frac{1}{2}\int_{M} T^{\tint}(J_{3}\cdot V)\cdot \kappa^{2}
\end{equation}
From this we can see that $\mu_{3}\equiv 0$ if and only if the singlet part of $T^{\tint}(V)$ is 0 for all $V\in \Gamma(E)$. From the decomposition of $E$ given in \eqref{eq:E6 GTB decomposition}, we see that this is equivalent to the $\rep{15}_{0}\oplus \repb{6}_{1}\oplus \repb{6}_{-1}$ part of $T^{\tint}$ vanishing. This is precisely the remaining part of the intrinsic torsion of the $\SUs{6}$ structure given in \eqref{eq:SUs6 intrinsic torsion}. This motivates the following alternative definition of an integrable H-structure
\begin{defn}\label{def:integrable SUs6}
An integrable $\SUs{6}$ structure $\Jpl:=\kappa J_{+}$, is an $\SUs{6}$ structure with an integrable $\Us{6}$ structure, along with the vanishing of the moment map $\mu_{3}$.
\end{defn}

As we mentioned, $\mu_{3}$ is a moment map for the action of generalised diffeomorphisms on space of $\SUs6$ structures $\ZSUs$. Following \cite{Ashmore:2015joa}, we recall that $\ZSUs$ has a hyperk\"ahler cone structure with hyperk\"ahler potential given by\footnote{Note our conventions differ from that of \cite{Ashmore:2015joa} by an irrelevant overall factor of $\frac{1}{2}$. }
\begin{equation}\label{eq:U(6) Kahler Potential}
    \mathcal{K} = \int_{M} \kappa^{2} =  \int_{M}\left(-\Tr(\Jpl\bar{\Jpl})\right)^{1/2}
\end{equation}
where in the second equality we have expressed the potential with respect to the particular holomorphic structure picked out by $\JUs$. We can view the right hand side as a K\"ahler potential in the holomorphic coordinate $\Jpl$. To see how the moment map arises, splitting the functional derivative into holomorphic and antiholomorphic parts $\delta = \del'+\delb'$, and using $\varpi = \ii\del'\delb'\mathcal{K}$ we can first write the K\"ahler form on the space as
\begin{align}
    \begin{split}
        \imath_{\beta}\imath_{\alpha}\varpi = -\frac{\ii}{2}\int_{M} & \frac{1}{\left(-\Tr(\Jpl\bar{\Jpl})\right)^{1/2}} \left[ \Tr(\imath_{\alpha}\delta\Jpl\imath_{\beta}\delta\bar{\Jpl}) - \Tr(\imath_{\beta}\delta\Jpl\imath_{\alpha}\delta\bar{\Jpl}) \right. \\
        & \left. +\frac{1}{2(-\Tr(\Jpl\bar{\Jpl}))}\bigl( \Tr(\imath_{\alpha}\delta\Jpl \bar{\Jpl}) \Tr(\Jpl\imath_{\beta}\delta\bar{\Jpl}) - \Tr(\imath_{\beta}\delta\Jpl \bar{\Jpl}) \Tr(\Jpl\imath_{\alpha}\delta\bar{\Jpl}) \bigr) \right]
    \end{split}
\end{align}
Using the non-holomorphic coordinate $J_{+}$, this takes the much simpler form
\begin{align}
    \begin{split}
        \imath_{\beta}\imath_{\alpha}\varpi = -\frac{\ii}{2}\int_{M}\left[ \Tr(\imath_{\alpha}\delta J_{+} \imath_{\beta}\delta J_{-}) - \Tr(\imath_{\beta}\delta J_{+} \imath_{\alpha}\delta J_{-}) \right]
    \end{split}
\end{align}
Here $\alpha,\beta \in \Gamma(T\ZSUs)$. Taking $\beta=\rho_{V}$ to be the vector generated by an infinitesimal generalised diffeomorphism, i.e. $\imath_{\rho_{V}}\delta J_{+} = L_{V}J_{+}$, then we see that
\begin{align}
    \imath_{\rho_{V}}\imath_{\alpha}\varpi &= -\frac{\ii}{2}\int_{M}\left[ \Tr(\imath_{\alpha}\delta J_{+} \imath_{\rho_{V}}\delta J_{-}) - \Tr(\imath_{\rho_{V}}\delta J_{+} \imath_{\alpha}\delta J_{-}) \right] \\
    &= -\frac{\ii}{2}\int_{M}\left[ \Tr(\imath_{\alpha}\delta J_{+} L_{V}J_{-}) - \Tr(L_{V} J_{+} \imath_{\alpha}\delta J_{-}) \right] \\
    &= \frac{\ii}{2} \int_{M} \left[ \Tr(J_{-}L_{V}\imath_{\alpha}\delta J_{+}) + \Tr(\imath_{\alpha}\delta J_{-}L_{V}J_{+}) \right] \\
    &= \imath_{\alpha}\delta \left( \frac{\ii}{2}\int_{M}\Tr(J_{-}L_{V}J_{+})\right) \\
    &= \imath_{\alpha}\delta\mu_{3}(V)
\end{align}
Therefore the moment map for generalised diffeomorphisms is indeed the function $\mu_{3}$ defined above.

\subsection{Example: Calabi--Yau Manifolds}\label{sec:CY example}

Calabi--Yau three-folds provide a supersymmetric background when all the fluxes vanish. We should thus be able to embed some of the data of the Calabi--Yau into the formalism of ECS. From \cite{Ashmore:2015joa} we see that the $\SUs6$ structure for the Calabi--Yau is given entirely in terms the $\SL{3,\bbC}$ structure parameterised by the complex three-form $\Omega$. Explicitly one has\footnote{Note that these expressions can be derived by taking the $\sin\theta \rightarrow -1$ limit in \eqref{eq:general H 1} and\eqref{eq:general H 2}. Taking the $\sin\theta\rightarrow 1$ limit gives the opposite orientation $\vol = -\frac{\ii}{8}\Omega\wedge \bar{\Omega}$.}
\begin{align}
    J_3 &= \tfrac{1}{2}\kappa(I - \vol - \vol^{\#})  \\
    J_+ = &= \tfrac{1}{2}\kappa(-\Omega + \Omega^{\#}) 
\end{align}
where $I$ is the $\GL{3,\bbC}$ structure associated to $\Omega$, that is the conventional complex structure of the Calabi--Yau, and $\kappa^2=\vol = \frac{\ii}{8}\Omega\wedge \bar{\Omega}$ is the volume form determined by the $\SL{3,\bbC}$ structure. 

As we have noted there is an $S^2$ family of ECSs defined by each $\SUs6$ structure. Conventionally, we take the one associated to $J_3$ so that in this case 
\begin{align}
    \text{$\SUs{6}$ structure :} & \quad \;\; \begin{array}{rcl}
        \Jpl &=& \frac{1}{2}\kappa^{2}(-\Omega + \Omega^{\#})
    \end{array} \label{eq:X-CY} \\[5pt]
    \text{$\bbRpl\times\Us{6}$ structure :} & \left\{ \begin{array}{rcl}
        \JUs & = & \frac{1}{2}(I - \vol - \vol^{\#})  \\[2pt]
        L_{1} & = & \ee^{\ii \vol}\cdot (T^{1,0}\oplus \ext^{0,2}T^{*}) 
    \end{array} \right. \label{eq:J-CY}
\end{align}
It is a simple check using the formulae for the adjoint action in appendix \ref{app:conventions} that $L_{1}$ is the $+i$ eigenbundle of $\JUs$, and $\Jpl$ lives in the $-2\ii$ eigenbundle of $\JUs$. Note that $\JUs$ is determined by knowing only the complex structure $I$ and the volume form, that together define a $\U1\times\SL{3,\bbC}\subset\GL{3,\bbC}$ structure. Recall that a complex structure always determines an orientation, so the additional information is just choice of section of the (trivial) determinant bundle. We see that embedding the Calabi--Yau structure into the language of ECS takes the following pattern of inclusions.
\begin{equation}
    \begin{array}{ccccc}
        \USp{6} & \subset & \SUs{6} & \subset & \bbRpl\times \Us{6} \\
        \cup & & \cup & & \cup \\
        \SU{3} &\subset & \SL{3,\bbC} & \subset & \U1\times\SL{3,\bbC} %\subset\GL{3,\bbC}
    \end{array}
\end{equation}
%From \cite{Ashmore:2015joa} we find that the relevant structures are\footnote{Note that these expressions can be derived by taking the $\sin\theta \rightarrow 1$ limit in \eqref{eq:general H 1}, \eqref{eq:general H 2}.}

%Here, $\Omega$ is the $\SL{3,\bbC}$ structure and $I$ is the associated $\GL{3,\bbC}$ structure. That is, $I$ is the conventional complex structure of the Calabi--Yau. $\vol = \frac{\ii}{8}\Omega\wedge \bar{\Omega}$ is the volume form picked out by the $\SL{3,\bbC}$ structure, and the musical isomorphisms are defined such that $\vol^{\#}\lrcorner \vol = 1$ and $\Omega^{\#} = -\ii\vol^{\#}\lrcorner\Omega$. $\kappa^{2}$ is some section of the bundle $\det T^{*}$. It is a simple check using the formulae for the adjoint action in appendix \ref{app:conventions} that $L_{1}$ is the $+i$ eigenbundle of $\JUs$, and $\Jpl$ lives in the $-2\ii$ eigenbundle of $\JUs$.

We next consider what the integrability conditions given by definitions \ref{def:integrable E6 ECS} and \ref{def:integrable SUs6} imply for the underlying $\GL{3,\bbC}$ and $\SL{3,\bbC}$ structures. Here we will ignore the flux for simplicity and assume that we know it vanishes. More generally, integrability of the $\SUs{6}$ structure also puts constraints on the flux as we will see later. First, let's consider the integrability of the $\bbRpl\times \Us{6}$ structure. Taking $V=\ee^{\alpha+\ii\vol}(v+\omega),V' = \ee^{\alpha+\ii\vol}(v' + \omega')\in \Gamma(L_{1})$, we have
\begin{align}
    \begin{split}
        L_{V}V' &= L_{\ee^{\ii\vol}(v+\omega)}\ee^{\ii\vol}(v' + \omega') \\
        &= \ee^{\ii\vol}L_{v+\omega}(v' + \omega') \\
        &= \ee^{\ii\vol}\left[\mathcal{L}_{v}v' + (\mathcal{L}_{v}\omega' - v'\lrcorner \dd\omega) - \omega'\wedge \dd\omega \right]
    \end{split}
\end{align}
For this to be a section of $L_{1}$ also, we require that $\mathcal{L}_{v}v' \in \Gamma(T^{1,0})$ for any $v,v'\in \Gamma(T^{1,0})$. This is precisely the statement that the $\GL{3,\bbC}$ structure is integrable. If this is the case then the exterior derivative decomposes into the Dolbeault operators $\dd = \del + \delb$. With this, the 5-form terms vanish identically. Further, the term in the parenthesis becomes $v\lrcorner\del\omega' - v'\lrcorner\del\omega\in \Gamma(\ext^{0,2}T^{*})$. Hence
\begin{equation}
    \text{$\bbRpl\times \Us{6}$ structure integrable} \quad \Leftrightarrow \quad \text{$\GL{3,\bbC}$ structure integrable}
\end{equation}

Next we consider the vanishing of the moment map $\mu_{3}$. Using the algebra in appendix \ref{app:conventions}, one can show that
\begin{align}\label{eq:CY moment map}
\begin{split}
    \mu_{3}(V) &= \ii\int_{M}-L_{V}\kappa^{2} + \frac{1}{16}\kappa^{2}\left[ -\bar{\Omega}^{\#}\lrcorner\mathcal{L}_{v}\Omega + \Omega^{\#}\lrcorner \mathcal{L}_{v}\bar{\Omega} - (\Omega^{\#}\wedge\bar{\Omega}^{\#})\lrcorner \dd\sigma \right]
\end{split}
\end{align}
This first term vanishes since $L_{V}\kappa^{2}$ is a total derivative\footnote{Here we use the fact that $\int_{M}\dd(...) = 0$ which means we are taking $M$ to be compact without boundary, or that the fields die off sufficiently quickly at infinity.}. The final term gives
\begin{equation}
    -\frac{1}{2}\int_{M} \kappa^{2} \vol^{\#}\lrcorner\dd\sigma \propto \int_{M} \vol^{\#}\lrcorner\kappa^{2} \dd\sigma \propto \int_{M} \dd(\vol^{\#}\lrcorner\kappa^{2})\wedge \sigma \overset{!}{=} 0 \quad \forall\,\sigma\in \Gamma(\ext^{5}T^{*})
\end{equation}
This is true if and only if $\kappa^{2} = c\vol$ for some constant $c\in \bbR$, which we can set to 1 without loss of generality. This just says that the volume form picked out by the $\SUs{6}$ structure is the same as that picked out by the $\SL{3,\bbC}$ structure. The rest of \eqref{eq:CY moment map} is proportional to
\begin{align}
    \begin{split}
        \int_{M}\left(\mathcal{L}_{v}\Omega \wedge \bar{\Omega} - \Omega \wedge \mathcal{L}_{v}\bar{\Omega}\right) &= \int_{M} \left(v\lrcorner\dd\Omega \wedge \bar{\Omega} + \dd(v\lrcorner\Omega)\wedge \bar{\Omega} - \Omega\wedge v\lrcorner\dd\bar{\Omega} - \Omega\wedge \dd(v\lrcorner\bar{\Omega})\right) \\
        &= 2\int_{M}v\lrcorner(\bar{a}-a)\,\Omega\wedge\bar{\Omega}
    \end{split}
\end{align}
In moving to the second line we have used integration by parts to put the derivatives on the $\Omega,\bar{\Omega}$, and have used the integrability of $I$ to write $\dd\Omega = \bar{a}\wedge\Omega$ for some $\bar{a}\in \Gamma(T^{*0,1})$. This vanishes for all $v\in \Gamma(T)$ if and only if $\bar{a}=0$, or equivalently, if $\dd\Omega = 0$. Therefore, we have
\begin{equation}
    \text{$\SUs{6}$ structure integrable} \quad \Leftrightarrow \quad \text{$\SL{3,\bbC}$ structure integrable}
    \label{eq:X-CY-int}
\end{equation}

As we have stressed, for a given H-structure there is an $S^2$ family of ECS one can define. It is thus natural to ask how the analysis would have changed if we had chosen a different ECS. Suppose for example, we act by a global $\SU2$ rotation on the $J_\alpha$, so that we have a new $\SUs6$ structure with $J'_1=J_3$, $J'_2=J_1$ and $J'_3=J_2$. We then have
\begin{align}
    \text{$\SUs{6}$ structure :} & \quad \;\; \begin{array}{rcl}
        \Jpl' &=& \frac{1}{2}\kappa^{2}(I - \vol - \vol^{\#}-\ii\rho + \ii \rho^{\#})
    \end{array} \\[5pt]
    \text{$\bbRpl\times\Us{6}$ structure :} & \left\{ \begin{array}{rcl}
        \JUs' & = & \frac{1}{2}( - \hat{\rho} + \hat{\rho}^{\#}) \\[2pt]
        L'_{1} & = & \ee^{\ii\hat{\rho}}\cdot T_\bbC
    \end{array} \right.
\end{align}
where $\Omega=\rho+\ii\hat{\rho}$. Recall that $\hat{\rho}$ by itself defines an $\SL{3,\bbC}$ structure~\cite{Hitchin00}. Hence in this case, both the $\SUs6$ structure and the $\bbRpl\times \Us{6}$ structure are defined by a conventional $\SL{3,\bbC}$ structure. As we will see below, this is actually the generic case. Turning to the integrability conditions, it is easy to see that, in this case, involutivity gives
\begin{equation}
    \text{$\bbRpl\times \Us{6}$ structure integrable} \quad \Leftrightarrow \quad \dd\hat{\rho} = 0 
\end{equation}
which is a considerably weaker than integrability of $I$ that we got previously. Imposing the moment map condition implies that in addition $\dd\rho=0$ and hence the $\SL{3,\bbC}$ structure is integrable. As expected, since the corresponding $\SUs6$ structures are the same (up to a constant $\SU2$ transformation), we thus get the same integrability condition as in \eqref{eq:X-CY-int}. 

\section{Classification of ECS and \texorpdfstring{$\SUs6$}{SU*(6)} structures} \label{sec:classification}

In this section we will examine what definition \ref{def:E6 ECS} implies for the structure of $L_{1}$ and of the corresponding $\SUs6$ structures. We will find that the isotropy and reality conditions place strong restrictions on the possible local form of the ECS in terms of natural bundles, such that characterised by two numbers that we refer to as type and class.

\subsection{A Closer Analysis of ECS}\label{sec:closer analysis of ECS}

First, we define a notion of \emph{type} similar to that of generalised complex structures defined in \cite{Gualtieri04}, and also defined in exceptional geometry in \cite{Ashmore:2019qii}. We then go further and show that ECS are also characterised by a second important property that we call \emph{class}. It is important to note that, in general, this classification is only local, in that the type of the structure can change over the manifold, just as for generalised complex structures. However, while the type and class of an ECS can change smoothly, as we will see in the following sections, it is not possible to have a smooth class-changing integrable $\SUs{6}$-structure. We make these definition only for M-theory though a similar analysis could be made for type IIB. To ensure that all objects defined here are globally well-defined sections of natural geometric bundles, we will work with the flux twisted\footnote{This implies that the generalised tangent bundle is globally isomorphic to $E=T\oplus \wedge^{2}T^{*} \oplus \wedge^{5}T^{*}$.} Dorfman derivative $L^{F}_{V}$, where $F\in H^{4}_{\dd}(M,\bbR)$ is defined only up to its cohomology class.

\begin{defn}
The \emph{type} of an almost ECS $L_{1}\subset E_{\bbC}$ at a point $p\in M$ is the (complex) codimension of its image under the anchor map. That is, if $\anchor:E\rightarrow T$ is the anchor map, naturally extended to the complexified bundles, then
\begin{equation}
    \type L_{1}|_p = \codim_{\bbC} a(L_{1}|_p) = 6 - \dim_{\bbC}a(L_{1}|_p)
\end{equation}
\end{defn}

\noindent We will find that the only allowed types of an ECS are 0 and 3.

We would like to classify the possible forms of $L_{1}$ based on the criteria set out in definition \ref{def:E6 ECS}. We will only state the results here and leave the proofs for appendix \ref{app:ECS in E6}. Let us first focus on condition (ii), the isotropy condition that states
\begin{equation}
    L_{1}\times_{N} L_{1} = 0
\end{equation}
If we write $V_{i} = v_{i}+\omega_{i}+\sigma_{i}\in L_{1}$, then using the formula for the projection onto $N$ around \eqref{eq:N projection} we find that the elements of $L_{1}$ must satisfy
\begin{align}
    v_{1}\lrcorner\omega_{2} + v_{2}\lrcorner\omega_{1} &= 0 \\
    j\omega_{1}\wedge \omega_{2} + j\omega_{2}\wedge\omega_{1} &= 0 \\
    \omega_{1}\wedge\omega_{2} - v_{1}\lrcorner\sigma_{2} - v_{2}\lrcorner\sigma_{3} &= 0
\end{align}
Careful consideration of these equations shows that any exceptional Dirac structure must be of the following form.

\begin{propn}
Any isotropic subbundle $L\subset E_{\bbC}$ has the form
\begin{equation}
    \ee^{\alpha+\beta}\cdot(\Delta \oplus S_{2}\oplus S_{5})
\end{equation}
where $\alpha\in\Omega^{3}(M)_{\bbC}$ and $\beta\in \Omega^{6}(M)_{\bbC}$ are arbitrary but fixed, and where $\Delta\subset T_{\bbC}$, $S_{2}\subset\ext^{2}T^{*}_{\bbC}$, $S_{5}\subset \ext^{5}T^{*}_{\bbC}$ satisfy the following conditions. For all $v\in \Delta$, $\omega,\omega'\in S_{2}$ and $\sigma\in S_{5}$ we have
\begin{equation}
    \arraycolsep=1.4pt
    \begin{array}{rclcrcl}
    v\lrcorner \omega &= &0 & \qquad & v\lrcorner\sigma &=&0 \\
    \omega\wedge\omega' &=& 0 & & j\omega\wedge\sigma &=& 0
    \end{array}
\end{equation}
\end{propn}

We now turn our attention to conditions (i) and (iii) in definition \ref{def:E6 ECS}. One finds that imposing $\dim_{\bbC}L_{1} = 6$ restricts us to type 0, 3 and 6. Then imposing $L_{1}\cap L_{0} = \{0\}$ excludes the type 6 case. We can therefore summarise the general form of an ECS in proposition \ref{prop:form of ECS}. To do so, it is helpful to introduce some notation which will be useful both here and later when we discuss the moduli of these structures.

Let $\Delta\subset T_{\bbC}$ be some subbundle. We define by $\mathcal{F}^{k}_{p}(\Delta) \subset \ext^{k}T^{*}_{\bbC}$ to be the bundle of differential $k$-forms $\phi$ satisfying
\begin{equation}
    \phi(x_{1},...,x_{p},v_{1},...,v_{k-p}) = 0 \qquad \forall \, x_{1},...,x_{p} \in \Gamma(T_\bbC), \, v_{1},...,v_{k-p}\in \Gamma(\Delta)
\end{equation}
Note that this defines a filtration of the fibers of $\ext^{k}T^{*}_{\bbC}$ with
\begin{equation}
    0 = \mathcal{F}^{k}_{k}(\Delta) \subseteq \mathcal{F}^{k}_{k-1}(\Delta)\subseteq ... \subseteq \mathcal{F}^{k}_{0}(\Delta) \subseteq \mathcal{F}^{k}_{-1}(\Delta) := \ext^{k}T^{*}_{\bbC}
\end{equation}
where we have defined $\mathcal{F}^{k}_{-1}(\Delta) = \ext^{k}T^{*}_{\bbC}$ for convenience. We can now summarise the results stated above as follows.

\begin{propn}\label{prop:form of ECS}
An ECS at $p\in M$ can only be of type 0 or type 3, and their general form is given by
\begin{equation}
\label{eq:type03L}
    \begin{array}{rcl}
    \text{type 0:}& \quad & \ee^{\alpha+\beta}\cdot T_{\bbC}|_p \\
    \text{type 3:}& & \ee^{\alpha+\beta}\cdot (\Delta \oplus \mathcal{F}^{2}_{1}(\Delta))|_p
    \end{array}
\end{equation}
where $\Delta|_p\subset T_{\bbC}|_p$ is dimension 3 and $\alpha \in \Omega^{3}(M)_{\bbC}$, $\beta \in \Omega^{6}(M)_{\bbC}$.
\end{propn}
\noindent
In each case we can also write the form of the line bundle $\mathcal{U}_\JUs$ defined by the ECS
\begin{equation}
    \mathcal{U}_\JUs|_p = \begin{cases}
     \ee^{\alpha+\beta}\cdot \bbC & \text{type 0}  \\
     \ee^{\alpha+\beta}\cdot \mathcal{F}^3_2(\Delta)|_p & \text{type 3} 
     \end{cases}
\end{equation}
where comparing with \eqref{eq:UJbundle}, for type 0 the leading term is in $\bbC$, the space of functions at $p$, and for type 3 it is a three-form in  $\mathcal{F}^3_2(\Delta)\subset\ext^3T^*_\bbC$ at $p$. Note that a section $\xi\in\Gamma(\mathcal{F}^3_2)$ defines $\Delta$ via the condition that $v\lrcorner\xi=0$ for all $v\in\Gamma(\Delta)$. Note also that, for type 3 ECS, the complex 3-form twist $\alpha$ is only defined up to a section of $\mathcal{F}^{3}_{1}$. Indeed, if $\gamma \in \Gamma(\mathcal{F}^{3}_{1}(\Delta))$ then, viewed as an adjoint element, one can show that
\begin{equation}\label{eq:complex twist ambiguity}
    \ee^{\alpha+\beta}\cdot (\Delta \oplus \mathcal{F}^{2}_{1}(\Delta)) = \ee^{(\alpha+\gamma)+(\beta-\frac{1}{2}\alpha\wedge\gamma)}\cdot (\Delta \oplus \mathcal{F}^{2}_{1}(\Delta))
\end{equation}
We will use this freedom to make a particularly simple choice of complex twist in the following.

The objects $\alpha,\beta,\Delta$ are not generic and are constrained by condition (iv) of definition \ref{def:E6 ECS}. To find the non-linear conditions imposed by (iv), we will consider the type 0 and type 3 cases separately. Note first that we can rewrite the group action
\begin{equation}
\label{eq:decomp-twist}
    \ee^{\alpha+\beta}  = \ee^{c+\tilde{c}}\ee^{\ii a+\ii b}
\end{equation}
where we have decomposed $\alpha=c+\ii a$ and $\beta=\tilde{c}+\ii(b-\frac12 c\wedge a)$ into real and imaginary parts. Since the conditions in definition \ref{def:E6 ECS} are preserved by real $\ER{6}$ transformations we expect they will only constrain $a$ and $b$. Starting with the type 0 structure, we can hence assume, without loss of generality, that it has the form. 
\begin{equation}
    L_{1}|_p = \ee^{\ii a + \ii b}\cdot T_{\bbC}|_p
\end{equation}
where $a\in \Omega^{3}(M)_{\bbR}$, $b\in \Omega^{3}(M)_{\bbR}$. Condition (iv) of definition \ref{def:E6 ECS} is then equivalent to
\begin{equation}\label{eq:type 0 ECS non-linear condition}
    \Tr(K_{a}^{2}) + 6b^{2} <0
\end{equation}
Here, $K_{a}:T\rightarrow T\otimes \det T^{*}$ is the map introduced by Hitchin in \cite{Hitchin00}. 
% The equation above is therefore a section of $(\det T^{*})^{2}$ which has a canonical orientation and hence a well defined notion of a negative section. 
Note in particular that \eqref{eq:type 0 ECS non-linear condition} implies that $\Tr(K_{a}^{2})<0$ and so $a$ defines an $\SL{3,\bbC}$ structure leading to

\begin{cor}\label{cor:class 0 classification}
A type 0 ECS at $p\in M$ is equivalent to an $\SL{3,\bbC}$ structure on $T|_p$, a bounded 6-form, and a generic real $\ER{6}$ transformation of the form $\ee^{c+\tilde{c}}$ where $c\in\Omega^3(M)_\bbR$ and $\tilde{c}\in\Omega^6(M)_\bbR$.
\end{cor}

For type 3, the details of the calculation become more complicated and so we have left them to appendix \ref{app:ECS in E6} and summarise the results here. We can use condition (iv) to put a constraint on $\dim(\Delta \cap \bar{\Delta})$. Indeed, we find that we must have
\begin{equation}\label{eq:class constraint}
    \dim(\Delta\cap\bar{\Delta}) \leq 1
\end{equation}
We make the following definition which refines the classification of ECS
\begin{defn}
The \emph{class} of an almost ECS $L_{1}\subset E_{\bbC}$ at $p\in M$ is the (complex) codimension of $a(L_{1}\oplus L_{-1})|_p$. That is
\begin{equation}
    \class L_{1}|_p = \codim_{\bbC} a(L_{1}\oplus L_{-1})|_p = \codim_{\bbC} (\Delta + \bar{\Delta})|_p
\end{equation}
Allowing $\dim\Delta$ to be 3 or 6, this definition holds for all ECS and it  follows from \eqref{eq:class constraint} that the class can only be 0 or 1. 
\end{defn}

For class 0 type 3, $\Delta|_p$ defines a $\GL{3,\bbC}$ almost complex structure. In this case, one can use the ambiguity in \eqref{eq:complex twist ambiguity} to write any imaginary twist $\ii a$ in the decomposition \eqref{eq:decomp-twist} as a real twist $c$. Hence, without loss of generality we can consider 
\begin{equation}
    L_{1}|_p = \ee^{\ii b}\cdot \left(\Delta \oplus \mathcal{F}^{2}_{1}(\Delta)\right)|_p
\end{equation}
From \eqref{eq:type 3 non-linear constraint}, the condition (iv) of definition \ref{def:E6 ECS} is then equivalent to 
\begin{equation}
    b > 0 
\end{equation}
where we are using the natural orientation defined by the almost complex structure to define the sign of $b$. Concretely if $\xi$ is a section of $\mathcal{F}^3_2(\Delta)$ (which in this case is the space of (0,3)-forms) our convention is that $\ii\bar{\xi}\wedge\xi$ is a positive six-form. Together $\Delta$ and $b$ are stabilised by $\U1\times\SL{3,\bbC}$ and so we have 
\begin{cor}\label{cor:type 30 classification}
A type 3 class 0 ECS at $p\in M$ is equivalent to a $\U1\times\SL{3,\bbC}$ structure on $T|_p$ and a generic real $\ER{6}$ transformation of the form $\ee^{c+\tilde{c}}$ where $c\in\Omega^3(M)_\bbR$ and $\tilde{c}\in\Omega^6(M)_\bbR$.
\end{cor}

For class 1, the situation is more complicated. One finds that the action of $b$ and some parts of $a$ are not independent. In particular, we can use the ambiguity in \eqref{eq:complex twist ambiguity} to remove the 6-form twist and write, without loss of generality
\begin{equation}
    L_{1}|_p = \ee^{\ii a}\cdot \left(\Delta \oplus \mathcal{F}^{2}_{1}(\Delta)\right)|_p
\end{equation}
where $a$ is a section of $\mathcal{A}\cap\bar{\mathcal{A}}$ where  $\mathcal{A}:=\ext^3T_\bbC/\mathcal{F}^3_1(\Delta)$. Using the result and notation of \eqref{eq:type 3 non-linear constraint}, the constraint is then 
\begin{equation}
\label{eq:class1condition}
    \ii (ja\wedge \xi) \wedge (ja \wedge \bar{\xi}) > 0 
\end{equation}
We will explain the content of this condition more explicitly in section \ref{sec:SU*(6) in local SU(2)}. In particular, we will see that together $\Delta$ and $a$ are stabilised by $\U2\times(\GL{1,\bbR}^2\ltimes\bbR)$ and so we have 
\begin{cor}\label{cor:type 31 classification}
A type 3 class 1 ECS at $p\in M$ is equivalent to a $\U2\times\GL{1,\bbR}^2$ structure on $T|_p$ and a real $\ER{6}$ transformation of the form $\ee^{c+\tilde{c}}$ where $c\in\Omega^3(M)_\bbR$ and $\tilde{c}\in\Omega^6(M)_\bbR$.
\end{cor}
\noindent
Note that in this case, unlike that of type 0 or type 3 class 0, the transformation is not generic in that not all $\ee^{c+\tilde{c}}$ transformations define different structures. 

We can, thus summarise the classification of almost ECS in the following proposition.
\begin{propn}
\label{prop:classification}
Any almost ECS must take one of the forms in table \ref{tab:almost ECS classification}.
\begin{table}[h]
    \centering
    \begin{tabular}{lllll}
    \toprule
        Class & Type & Bundle $L_{1}$ & Constraints & $G$-structure on $T$ \\ \midrule
        0 & 0 & $\ee^{\alpha+\beta}\cdot T_{\bbC}$ & $\Tr(K_{a}^{2}) + 6b^{2} < 0$ & $\SL{3,\bbC}$ \\
        0 & 3 & $\ee^{\alpha+\beta}\cdot (\Delta \oplus \mathcal{F}^{2}_{1}(\Delta))$ & $\dim(\Delta\cap\bar{\Delta}) = 0, \ b>0$  & $\U1\times \SL{3,\bbC}$ \\
        1 & 3 & $\ee^{\alpha+\beta}\cdot (\Delta \oplus \mathcal{F}^{2}_{1}(\Delta))$ & $\dim(\Delta\cap\bar{\Delta}) = 1,$ \eqref{eq:class1condition} & $\U2\times(\GL{1,\bbR}^2\ltimes\bbR)$ \\
        \bottomrule
    \end{tabular}
    \caption{Classification of almost ECS in $\ER{6}$ geometry for M-theory backgrounds. Here $\alpha \in \Omega^{3}(M)_{\bbC}$, $\beta \in \Omega^{6}(M)_{\bbC}$, $\Delta \subset T_{\bbC}$, $\dim\Delta = 3$. We also have $a$ and $b$ as defined in \eqref{eq:decomp-twist}}
    \label{tab:almost ECS classification}
\end{table}
\end{propn}
\noindent
Since every $\bbRpl\times\Us{6}$ structure admits a $\USp6$ structure, we can use the explicit spinor bilinear expressions \eqref{eq:general H 1} and \eqref{eq:general H 2} for $J_\alpha$ found in \cite{Ashmore:2016qvs} to give the concrete form for $L_1$ in each of these cases. These will be given in section \ref{sec:SU*(6) in local SU(2)} below. 

We now turn to the conditions for integrability of the $\Us{6}\times \bbR^{+}$ structure. Recall from \eqref{eq:Us(6) integrability} that an ECS is integrable if and only if it is involutive with respect to the Dorfman derivative. We will write $L_{1} = \ee^{\alpha+\beta}\cdot(v+\omega) \in \Gamma(L_{1})$, and similarly for $V'$, where $v\in \Gamma(\Delta)$ and $\omega \in \Gamma(\mathcal{F}^{2}_{1})$. Note that in the case that $\Delta = T_{\bbC}$, $\mathcal{F}^{2}_{1}(\Delta)=0$ and hence this expression covers both type 0 and type 3 (and the cases of type-changing). Involutivity then becomes
\begin{align}
    \begin{split}\label{eq:E6 ECS Dorfman}
       \Gamma(L_{1}) \ni L^{F}_{V}V' &= \ee^{\alpha+\beta}\cdot\left([v,v'] + (\mathcal{L}_{v}\omega' - v'\lrcorner\dd\omega + v'\lrcorner(v\lrcorner(F+ \dd\alpha))) \right. \\
        & \qquad \qquad \quad  \left. -\omega'\wedge \dd\omega + \omega'\wedge (v\lrcorner(F+\dd\alpha)) \right)
    \end{split}
\end{align}
For this to be true we require $[\Delta,\Delta]\subseteq \Delta$. In the global type 0 structure this is trivial, but more generally it implies the existence of a generalised involutive distribution where $\dim\Delta|_p=0,3$. For a global type 3 structure this becomes a three-dimensional foliation. The 2-form piece implies that for all $v,v'\in \Gamma(\Delta)$
\begin{equation}
    v\lrcorner \dd\omega' - v'\lrcorner\dd\omega + v'\lrcorner(v\lrcorner(F+ \dd\alpha)) \in \Gamma(\mathcal{F}^{2}_{1})
\end{equation}
A short calculation shows that, provided $\Delta$ is integrable in the sense of Frobenius, the de Rham differential restricts to $\dd:\Gamma(\mathcal{F}^{k}_{p}) \rightarrow \Gamma(\mathcal{F}^{k+1}_{p})$. Hence, $\Gamma(\mathcal{F}^{\bullet}_{p})$ defines a filtration of the de Rham complex. It is then clear that for all $v \in \Gamma(\Delta)$, $\omega'\in \Gamma(\mathcal{F}^{2}_{1})$ we have $v\lrcorner\dd\omega' \in \Gamma(\mathcal{F}^{2}_{1})$. We further require that $v'\lrcorner(v\lrcorner(F+ \dd\alpha)) \in \Gamma(\mathcal{F}^{2}_{1})$ which we can restate as
\begin{equation}\label{eq:complex flux condition}
   F_{\bbC} = F + \dd\alpha\in \Gamma(\mathcal{F}^{4}_{1})
\end{equation}
Finally, we need the 5-form term to vanish in \eqref{eq:E6 ECS Dorfman}. Since $\omega'\in \Gamma(\mathcal{F}^{2}_{1})$ and $\dd\omega \in \Gamma(\mathcal{F}^{3}_{1})$, you can show that $\omega'\wedge (\dd\omega - v\lrcorner F_{\bbC}) \in \Gamma(\mathcal{F}^{5}_{3})$. This space trivially vanishes because $\Delta$ is of dimension (at least) 3. Therefore, the final term vanishes trivially.

Finally, recall the ambiguity \eqref{eq:complex twist ambiguity} in the definition of the complex 3-form twist $\alpha\sim \alpha+\gamma$, for any $\gamma\in \Gamma(\mathcal{F}^{3}_{1})$. Combining this with \eqref{eq:complex flux condition}, we see that $F_{\bbC}$ is only defined up to $\dd\gamma$. That is, the complex flux is part of the cohomology group
\begin{equation}
    F_{\bbC} \in \frac{\left\{ a \in \Gamma(\mathcal{F}^{4}_{1})\,|\, \dd a = 0 \right\}}{ \left\{ a = \dd b \,|\, b\in \Gamma(\mathcal{F}^{3}_{1}) \right\} } =: H^{4}(M,\mathcal{F}^{\bullet}_{1})
\end{equation}
We can summarise the results as follows.

\begin{propn}\label{prop: ECS integrability E6}
An integrable ECS is of the form $L_{1} = \ee^{\alpha+\beta}\cdot (\Delta\oplus \mathcal{F}^{2}_{1}(\Delta))$ for $\Delta \subset T_{\bbC}$, $\alpha \in \Omega^{3}(M)_{\bbC}$ and $\beta \in \Omega^{6}(M)_{\bbC}$ as in proposition \ref{prop:classification} such that 
\begin{equation}
    [\Delta,\Delta]\subseteq \Delta \qquad F_{\bbC}:= F + \dd\alpha \in H^{4}(M,\mathcal{F}^{\bullet}_{1}) \label{eq:integrability conditions}
\end{equation}
\end{propn}
\noindent
In the case that the ECS is of global type 0, the second condition just implies $F_{\bbC} = 0$, which in turn implies that $F = -\dd\alpha$ and hence $F$ must be in a trivial cohomology class. Since the twisted Dorfman derivative is only defined up to the cohomology, we can take $F=0$ in this case. For global type 3 solutions, the complex flux does not need to vanish and in general, $F$ can be in a non-trivial cohomology class.

We will determine the conditions on $\alpha,\beta$ for integrability of the full $\SUs{6}$ structure in the following sections.

\subsection{The classification of \texorpdfstring{$\SUs6$}{SU*(6)} structures}\label{sec: SU(2) symmetry}

Recall that an $\SUs{6}$ structure can be described by a triplet of adjoint elements $J_{\alpha}$ which form a highest weight $\su{2}$ subalgebra. In passing to the description in terms of ECS, we have broken the explicit $\SU{2}$ symmetry prescribed by this structure. We now take a closer look at the implications of this additional symmetry of the structure.

In the definition of ECS in section \ref{sec:SUs(6) structures}, we took $L_{1}$ to be the $+\ii$ eigenbundle of $\JUs = \kappa^{-1}J_{3}$. However, this is just an arbitrary choice and we could choose any one of $J_{1},J_{2}, J_{3}$. In fact, we can choose any linear combination $\JUs_{u} = \kappa^{-1}u^{\alpha}J_{\alpha}$, where $\mathbf{u} \in \bbR^{3}$, such that
\begin{equation}
    \Tr (\JUs_{u}\JUs_{u}) = -1 \qquad \Leftrightarrow \qquad |\mathbf{u}|^{2} = 1
\end{equation}
We see that we have an $S^{2} \simeq \bbC \mathbb{P}^{1} \simeq \SU{2}\quotient \U{1}$ of possible ECS for each $\SUs{6}$ structure. That is, each $\JUs_{u}$ defines a $\Us{6}\times \bbR^{+}$ structure. (Equivalently we can rotate the original $J_\alpha$ by a global $\SU2$ transformation to define a new $\SUs6$ structure $J'_\alpha$ and set $\JUs_u=J'_3$.) Recall that we could equally define the structure in terms of eigenbundles. Taking the $+\ii$ eigenbundle of $\JUs_{u}$ to be $L_{u}$, we find
\begin{equation}\label{eq:general Lu}
    L_{u} = \begin{cases}
    \left( 1+\frac{\ii(u^{3}-1)}{2(u^{1}+\ii u^{2})}\kappa^{-1}J_{+} \right)\cdot L_{1} & u^{3}\neq \pm 1 \\
    L_{\pm 1} & u^{3} = \pm 1
    \end{cases}
\end{equation}

We saw that we could describe the integrability of an $\SUs{6}$ structure as the vanishing of the moment maps \eqref{eq:H condition 2}, or equivalently as the integrability of $L_{1}$ and the vanishing of $\mu_3$. However, since the choice of $\JUs_u$ just corresponds to a global $\SU2$ rotation of the $J_\alpha$, integrability defined by $\JUs_u$ implies the vanishing of the same (rotated) moment maps. Hence the integrability of the $\SUs6$ structure can also be written as 
\begin{equation}
    \llbracket L_{u}, L_{u} \rrbracket \subseteq L_{u} \qquad u^{\alpha}\mu_{\alpha} = 0
\end{equation}
for any $u\in S^2$.

We may wonder whether involutivity of some fixed $L_{u}$ implies involutivity of all $L_{u}$. However, it turns out that this condition is not quite strong enough. A quick calculation shows that requiring involutivity for all $\mathbf{u}\in S^{2}$ implies that the intrinsic torsion of the $\SUs{6}$ structure satisfies
\begin{equation}
    T^{\tint}|_{\rep{15}_{+2}\oplus \rep{15}_{0}\oplus \rep{15}_{-2}} = 0
\end{equation}
which is a strictly stronger condition than just involutivity of $L_{1}$ (which just required the $\rep{15}_{\pm 2}$ part of of $T^{\tint}$ to vanish). It is then possible to show that
\begin{equation}
    L_{V}\Jpl = A(V) \Jpl \qquad \forall \,V\in \Gamma(L_{1})
\end{equation}
where $\Gamma(E^{*})\ni A\sim T^{\tint}|_{\repb{6}_{-1}}$. From this it is clear that full integrability of the $\SUs{6}$ structure is given by further imposing $L_{V}\Jpl = 0$ for all $V\in \Gamma(L_{1})$. Defining $\Jpl_u$ as the $\SUs6$ structure defined by an involutive $L_u$, it is easy to see that $L_{V}\Jpl_u=0$ for all $V\in \Gamma(L_{u})$ gives the same condition for any $\mathbf{u}\in S^{2}$. This proves the following statement.
\begin{propn}\label{prop:different integrability}
Let $\Jpl \in \Gamma(\det T^{*} \otimes \ad \tilde{F})$ define an $\SUs{6}$ structure, equivalently defined by weighted adjoint tensors $J_{\alpha}$. Then the following are equivalent.
\begin{enumerate}
    \item The $\SUs{6}$ structure is integrable
    \item $\forall \, \mathbf{u}\in S^{2}$, $u^{\alpha}\mu_{\alpha} = 0$ \label{integrability i}
    \item $\forall \,\mathbf{u}\in S^{2}$, $\llbracket L_{u}, L_{u} \rrbracket \subseteq L_{u}$, $L_{V}\Jpl_{u} = 0$ for all $V\in \Gamma(L_{u})$ \label{integrability ii}
    \item $\exists\,\mathbf{u} \in S^{2}$, $\llbracket L_{u}, L_{u} \rrbracket \subseteq L_{u}$,  $u^{\alpha}\mu_{\alpha} = 0$ \label{integrability iii}
\end{enumerate}
\end{propn}

Turning to the form of the $\SUs6$ structure, since $L_{u}$ defines an ECS for any $\mathbf{u}$, it must take the form of one of the cases in table \ref{tab:almost ECS classification}. Since a generic $L_{u}$ mixes together parts of $L_{\pm 1}$, the type may vary as we change $\mathbf{u}$. However, since $L_{1}\oplus L_{-1}$ is $\SU{2}$ invariant, the class will remain fixed as we vary $\mathbf{u}$. Hence, the following is well-defined.
\begin{defn}
The \emph{class} of an $\SUs{6}$ structure at point $p\in M$ is the class of any of its associated ECS. That is
\begin{equation}
    \class \Jpl|_p = \class L_{u}|_p \quad \text{any} \quad \mathbf{u}
\end{equation}
\end{defn}
It then immediately follows that we only have two categories of $\SUs6$ structures
\begin{propn}
Let $\Jpl \in \Gamma(\det T^{*} \otimes \ad \tilde{F})$ define an $\SUs{6}$ structure, equivalently defined by weighted adjoint tensors $J_{\alpha}$. Then
\begin{equation}\label{eq:class 0 or 1}
    \begin{aligned}
    \text{either} \quad \class \Jpl|_p &= 0  & & \Leftrightarrow & \exists\, &\mathbf{u} \;\; \type L_{u}|_p = 0 \\
    \text{or} \quad \class \Jpl|_p &= 1 & &  \Leftrightarrow  & \forall\, &\mathbf{u} \;\; \type L_{u}|_p = 3
    \end{aligned}
\end{equation}
\end{propn}
The second of these statements is clear. Indeed, if $\codim \anchor( L_{1}\oplus L_{-1})|_p > 0$ then $\codim \anchor(L_{u})|_p >0$ for any $L_{u}\subset L_{1}\oplus L_{-1}$. Hence $\type L_{u}|_p \neq 0$ so it must be type 3 for all $\mathbf{u}$. The first statement can be seen by applying \eqref{eq:general Lu} to the class 0, type 3 solution in table \ref{tab:almost ECS classification}.

Finally, we can apply the conditions in proposition \ref{prop:different integrability} to find the form of integrable $\SUs6$ structures. For a class 0 structure, choosing a $\mathbf{u}$ such that the $L_u$ is type 0, then from \eqref{eq:UJbundle} we have 
\begin{equation}
\label{eq:class0X}
    \Jpl_u = \ee^{\alpha+\beta}\cdot \h = \ee^{c+\tilde{c}}\ee^{\ii a+\ii b} \cdot \h 
\end{equation}
where $\h $ is a function and we have decomposed $\alpha=c+\ii a$ and $\beta=\tilde{c}+\ii(b-\frac12 c\wedge a)$. Recall that integrability of $L_u$ implies
\begin{equation}
    \dd a = \dd c = 0 
\end{equation}
It is then relatively straightforward to show that the vanishing of the moment map $u^\alpha\mu_\alpha$ is equivalent to
\begin{equation}
    \dd \hat{a} = 0 , \qquad
    \dd \h  = 0 , \qquad
    \dd \left(b/(a\wedge\hat{a})\right) = 0 
\end{equation}
where $\hat{a}$ is the conjugate three-form to $a$ defined by Hitchin \cite{Hitchin00a}. This means that $a+\ii\hat{a}$ is a holomorphic three-form defining an integrable $\SL{3,\bbC}$ structure. Thus we have 
\begin{propn}
\label{prop:class0}
An integrable class 0 $\SUs{6}$ structure is equivalent to 
\begin{itemize}
    \item[i)] an integrable $\SL{3,\bbC}$ structure $\tilde{\Omega}=a+\ii\hat{a}$
    \item[ii)] a closed real three-form $c$ and real six-form $\tilde{c}$, a complex constant $\h $ and a real constant $b_0$ with $6b_0^2<1$
%     \item[iii)] an exact 4-form flux $F=\dd A$
\end{itemize}
where there exists a $\mathbf{u}$ such that $\Jpl_u$ takes the form \eqref{eq:class0X} with $b=\frac14 b_0\,a\wedge \hat{a}$. 
\end{propn}
\noindent 
Recall that for a type 0 structure integrability implied the $F$ appearing in the twisted Dorfman derivative was trivial in cohomology. Hence we can choose a twist where $F=0$, and this is why no cohomological data needs to be specified in the proposition. When $\mathbf{u}$ is such that we have a class 0 type 3 ECS, this is not immediately clear, since by proposition \ref{prop: ECS integrability E6} the integrability of the ECS allows a non-trivial $F$. However, once one imposes the moment map condition, so that the $\SUs6$ structure is integrable, one again finds $F$ is trivial in cohomology. 

In the next section we will see explicitly how these ingredients are related to the local $\SU{2}$ structure and spinor bilinears defined in section \ref{sec: ECY review}. This actually gives the simplest way to summarise the geometrical structure determined by an integrable class 1 structure. We find
\begin{propn}
\label{prop:class1}
An integrable class 1 $\SUs{6}$ structure is equivalent to 
\begin{itemize}
    \item[i)] an $\SU2\times\GL{1,\bbR}\times\bbR$ structure defined by  $(\omega_\alpha,\zeta_1,\zeta_2)$ where we make the identification  $(s\omega_\alpha,s^{-1}\zeta_1,s^{-1}\zeta_2)\sim(\omega_\alpha,\zeta_1,\zeta_2)$ and $\zeta_1+t\zeta_2\sim\zeta_1$ for $(s,t) \in \GL{1,\bbR}\times \bbR$, and a function $\ee^{3\warp}$ satisfying \eqref{eq:class 1 inv} and \eqref{eq:class 1 flux 2}
    \item[ii)] a six-form potential $\tilde{A}$, and a form-form flux $F$ satisfying \eqref{eq:class 1 flux 1}
\end{itemize}
For general $\mathbf{u}$ we can write the $\SUs6$ structure as 
\begin{equation}
\label{eq:class1X}
    \Jpl_u = \ee^{A+\tilde{A}}\,\ee^{-\ii\omega_{u}\wedge \zeta_{1}}\cdot \ee^{3\lambda}\Omega_{u}\wedge\zeta_2
\end{equation}
with $\omega_u=u^\alpha\omega_\alpha$ and $\Omega_u$ given by \eqref{eq:Omega u}
\end{propn}
\noindent
Comparing \eqref{eq:class0X} and \eqref{eq:class1X} for generic $\mathbf{u}$, we see that for a class 0 structure the leading term in $\Jpl_u$ is a function $h$ while for a class 1 structure it is the three-form $\ee^{3\lambda}\Omega_{u}\wedge\zeta_2$. The condition $\dd h=0$ means that $h$ is constant and so we the $\SUs6$ structure cannot smoothly change from class 0 to class 1 as one moves from point to point in the manifold. Thus we have 
\begin{cor}
\label{cor:class change}
There are no (smooth) class-changing integrable $\SUs6$ structures. 
\end{cor}
\noindent
One can on the other hand have a smooth type-changing ECS, as it only specifies $\Jpl$ up to a local $\bbC^*$ action.

\subsection{ECS and \texorpdfstring{$\SUs6$}{SU*(6)} structures and the local \texorpdfstring{$\SU2$}{SU(2)} structure}\label{sec:SU*(6) in local SU(2)}

We are now in a position to see how this classification applies to the general expressions \eqref{eq:general H 1}, \eqref{eq:general H 2} for the $J_{\alpha}$ in terms of the local $\SU{2}$ structure that were found in \cite{Ashmore:2016qvs}. Since every $\SUs6$ structure admits a $\USp6$ structure, the objects defined in \eqref{eq:general H 1}, \eqref{eq:general H 2} are in fact generic, and so the $\SU2$ bilinears give a good parameterisation of any ECS.

One can ask under what conditions we have a class 0 $\SUs6$ structure. From \eqref{eq:class 0 or 1} we can see that this equates to the condition that there exists a type 0  $L_{u}$, or equivalently, from \eqref{eq:UJbundle}, that  $\Jpl_u= \ee^{\alpha+\beta}\h $ for some function $\h $. Examining \eqref{eq:general H 1}, \eqref{eq:general H 2}, we see that this is true provided either $\sin\theta$ or $f$ is non vanishing. That is, we have
\begin{equation}
    \class \Jpl = \begin{cases}
    1 & \sin\theta = f = 0 \\
    0 & \text{otherwise}
    \end{cases}
\end{equation}
We will study these two cases separately.

\subsubsection*{Class 0}

Let us first focus on the class 0 case. Using \eqref{eq:general Lu} and the expressions for the $J_{\alpha}$ given in \eqref{eq:general H 1}, \eqref{eq:general H 2}, we find that the generic class 0 ECS is given by
\begin{equation}
    L_{u} = \ee^{A+\tilde{A}}\ee^{\Lambda_{u}\quotient f_{u}}\cdot T_{\bbC}
\end{equation}
where we have included the form-field gauge potentials $A$ and $\tilde{A}$ and defined
\begin{align}\label{eq:generic LX and f}
    \Lambda_{u} &= \begin{cases}
     V-\frac{\ii(u^{1} - \ii u^{2})}{2(1+u^{3})}\Lambda - \frac{\ii(u^{1} + \ii u^{2})}{2(1-u^{3})}\bar{\Lambda} \\
    \bar{\Lambda} \\
    \Lambda
    \end{cases}\!\!\! f_{u} = \begin{cases}
    \sin\theta - \frac{\ii f (u^{1}-\ii u^{2})}{2(1+u^{3})} - \frac{\ii \bar{f} (u^{1}+\ii u^{2})}{2(1-u^{3})} & u^{3} \neq \pm 1 \\
     \bar{f} & u^{3} = 1 \\
     f & u^{3} = -1
    \end{cases}
\end{align}
These expressions hold for generic $\mathbf{u}$ and generic points on the manifold where $f_{u} \neq 0$. Wherever 
\begin{equation}
\label{eq:type3u}
    \mathbf{u} = r^{-1}(\im f, -\re f, \sin\theta)|_{p}
\end{equation}
where $r^2 = \sin^2\theta + f\bar{f}$ one has $f_u=0$ and the ECS $L_{u}|_{p}$ degenerates to a type 3 (but still class 0) structure.

Comparing with \eqref{eq:type03L} we can read off
\begin{equation}
\begin{aligned}
    \alpha &= \left(A+\re(\Lambda_u/f_u)\right) + \ii \im(\Lambda_u/f_u) \\
    \beta &= \tilde{A} - \tfrac{1}{2}A\wedge (\Lambda_u/f_u)
\end{aligned}
\end{equation}
and hence from proposition \ref{prop:classification} we have that $a_u=\im(\Lambda_u/f_u)$ defines an $\SL{3,\bbC}$ structure. As we vary $\mathbf{u}$, this 3-form changes and so we can compare the induced $\SL{3,\bbC}$ structures. If we do not impose integrability of the $\SUs{6}$ structure, then $\sin\theta$ and $f$ are two independent functions on the manifold and in general, $\im(\Lambda_{u} \quotient f_{u})$ will define inequivalent $\SL{3,\bbC}$ structures for different values $\mathbf{u}$. Imposing $\SUs6$ integrability, however, we find that involutivity of the $L_{u}$ $\forall \, \mathbf{u}$ is equivalent to
\begin{equation}\label{eq:class 0 involutive all u}
    \dd(\Lambda\quotient f) = -F \qquad \dd(V\csc\theta + \Lambda\quotient f) = 0 \qquad \dd(f\csc\theta) = 0
\end{equation}
In particular, we see that there is some constant $c$ such that $f=c\sin\theta$. This is enough to show that the $\SL{3,\bbC}$ structures induced from $\im(\Lambda_{u}\quotient f_{u})$ will be equivalent up to rescaling by a non-vanishing constant. Concretely, one finds a closed holomorphic three-form, as in proposition \ref{prop:class0}, given by 
\begin{equation}
    \tilde{\Omega} = \begin{cases} 
        V - \sin\theta\re(\LX/f) + \ii r \im(\LX/f) & f\neq 0 \\
        \LX/\sin\theta & f=0
        \end{cases}
\end{equation}
Moreover, imposing this we find that the right hand side of \eqref{eq:type3u} is a constant vector. Then if \eqref{eq:type3u} holds at a point $p$, it will hold everywhere on $M$ and the solution will be globally type 3.

Finally, imposing the condition that $L_{V}\Jpl = 0$ for $V\in \Gamma(L_{1})$ gives
\begin{equation}
    \dd(\ee^{3\lambda}f) = 0
\end{equation}
which, combined with \eqref{eq:class 0 involutive all u}, gives the subset of the Killing spinor equations in \eqref{eq:Killing spinor eqn}.

\subsubsection*{Class 1}

Let us now consider the class 1 case. Here we have $\sin\theta = f =0$ and so the expressions for the $J_{\alpha}$ in \eqref{eq:general H 1}, \eqref{eq:general H 2} simplify. In terms of the local $\SU2$ structure given in terms of the triplet of 2-forms $\omega_{1},\omega_{2},\omega_{3}$, and the 1-forms $\zeta_{1},\zeta_{2}$, we have 
\begin{equation}
    J_{\alpha} = \tfrac{1}{2}\kappa \omega_{\alpha\, R} + \tfrac{1}{2}\kappa \big(\omega_{\alpha}\wedge \zeta_{1} - \omega_{\alpha}^{\#}\wedge \zeta_{1}^{\#} \big)
\end{equation}
where we have realigned\footnote{We send $J_{3} \rightarrow - J_{3}$, $J_{2} \rightarrow - J_{1}$, $J_{1} \rightarrow - J_{2}$. This is an $\SU{2}$ rotation so defines an equivalent $\SUs{6}$ structure} the basis of $J_{\alpha}$ relative to \cite{Ashmore:2016qvs} so that they match with the $\SU{2}$ structure. 

It is clear then that choosing a $\JUs_{u}$ is equivalent to choosing some distinguished complex structure $\omega_{u\, R} = u^{\alpha}\omega_{\alpha\, R}$, three-form $u^{\alpha}\omega_{\alpha}\wedge \zeta_1$ and three-vector $u^{\alpha}\omega^\#_{\alpha}\wedge \zeta^\#_1$. Using these, and including the form-field gauge potentials $A$ and $\tilde{A}$, we find that the ECS for any $\mathbf{u}$ is
\begin{align}
    L_{u} &= \ee^{A+\tilde{A}}\,\ee^{-\ii\omega_{u}\wedge \zeta_{1}}\cdot [\Delta \oplus \mathcal{F}^{2}_{1}(\Delta)] 
    \label{eq:SU2 type 0}\\
    \Delta &= T^{1,0}_{u} \oplus \bbC \zeta_{1}^{\#}
\end{align}
where $T^{1,0}_{u} \subset T_{\bbC}$ has a $+\ii$ eigenvalue under the action complex structure $\omega_{u\,R}$ and $\bbC \zeta_{1}^{\#}$ is defined by having zero wedge product with $u^{\alpha}\omega^\#_{\alpha}\wedge \zeta^\#_1$. Note that in writing $L_u$, the imaginary exponential is not uniquely determined since there is a kernel for the action on $\Delta\oplus\mathcal{F}^2_1(\Delta)$. The kernel is the space of sections $\mathcal{F}^3_1\cap\overline{\mathcal{F}^3_1}$, that is elements of the form  $\gamma\wedge\zeta_2$ where $\gamma$ is a $(1,1)$ form. Thus we should identify 
\begin{equation}
\label{eq:equiv omega}
    \omega_{u}\wedge \zeta_{1} \sim \omega_{u}\wedge \zeta_{1} + \gamma\wedge\zeta_2
\end{equation}
Comparing with \eqref{eq:type03L} we can read off
 \begin{equation}
 \begin{aligned}
    \alpha &= A -  \ii \omega_u \wedge \zeta_1 \\
    \beta &= \tilde{A} + \tfrac{1}{2}\ii A\wedge \omega_u \wedge \zeta_1
\end{aligned}
\end{equation}
and hence we see that the general solution of the constraint \eqref{eq:class1condition} is $a=-\omega_u\wedge\zeta_1$ up to the equivalence \eqref{eq:equiv omega}. It is then relatively straightforward to show that the group that stabilises the pair $(\Delta,a)$, up to the equivalence \eqref{eq:equiv omega}, is indeed $\U2\times(\GL{1,\bbR}^2\ltimes\bbR)$ as claimed in corollary \ref{cor:type 31 classification}. The $\GL{1,\bbR}^2$ factors act as $\omega_u\to s\omega_u$, $\zeta_1\to s^{-1}\zeta_1$ and $\zeta_2\to\ s'\zeta_2$. The $\bbR$ action is the shift  $\zeta_1\to\zeta_1+t\zeta_2$ for some function $t$, which in \eqref{eq:equiv omega} corresponds to $\gamma=t\omega_u$. 

The corresponding $\SUs6$ structure is given by
\begin{equation}
    \Jpl_u = \ee^{A+\tilde{A}}\,\ee^{-\ii\omega_{u}\wedge \zeta_{1}}\cdot \ee^{3\lambda}\Omega_{u}\wedge\zeta_2 
\end{equation}
where
\begin{equation}
\label{eq:Omega u}
    \Omega_{u} = \sqrt{1-(u^{3})^{2}}\left( \omega_{3} + \frac{u^{1}-\ii u^{2}}{2(1+u^{3})}(\omega_{1} + \ii \omega_{2}) - \frac{u^{1}+\ii u^{2}}{2(1-u^{3})}(\omega_{1}-\ii \omega_{2}) \right)
\end{equation}
Note that $\Jpl_u$ fixes a conventional $\SU{2}\times\GL{1,\bbR}\times\bbR$ structure on $T$. Involutivity of $L_u$ $\forall \mathbf{u}$ together with $L_{V}\Jpl = 0$ for $V\in \Gamma(L_{1})$ implies 
\begin{equation}
\label{eq:class 1 inv}
    \dd \big( \ee^{3\lambda} \omega_\alpha \wedge\zeta_2 \big) = 0
\end{equation}
which is equivalent (given $f=\sin\theta=0$) to the conditions \eqref{eq:Killing spinor eqn}. It also implies that $F_\bbC=F-\ii\dd(\omega_u\wedge\zeta_1)\in\Gamma(\mathcal{F}^4_1)$ for all $\mathbf{u}$ or equivalently 
\begin{gather}
    (\omega_\alpha\wedge \zeta_2)^\#\lrcorner F
        = \tfrac12 \epsilon_{\alpha\beta\gamma}
            (\omega_\beta\wedge \zeta_2)^\#\lrcorner \dd(\omega_\gamma\wedge\zeta_1) \label{eq:class 1 flux 1} \\*[0.3em]
    (\omega_\alpha\wedge \zeta_2)^\#\lrcorner
       \dd(\omega_\beta\wedge\zeta_1) 
    + (\omega_\beta\wedge \zeta_2)^\#\lrcorner
       \dd(\omega_\alpha\wedge\zeta_1) 
       = \tfrac{2}{3}\delta_{\alpha\beta}
         (\omega_\gamma\wedge \zeta_2)^\#\lrcorner
       \dd(\omega_\gamma\wedge\zeta_1) \label{eq:class 1 flux 2}
\end{gather}
where here $F=\dd A$ is the physical flux. Note that the condition \eqref{eq:class 1 inv} means that these equations are invariant under $\zeta_1\to\zeta_1+t\zeta_2$ as expected. We see, in particular, that the full $\SU2\times\GL{1,\bbR}\times\bbR$ structure is not integrable as a conventional $G$-structure, although the larger structure defined by the the set of 3-forms $\omega_\alpha \wedge\zeta_2$ is. 

Including the V-structure $K$, fixes $\zeta_1$ and $\zeta_2$ and hence a global conventional  $\SU2$-structure. As discussed in \cite{Gauntlett:2004zh}, the corresponding additional supersymmetry conditions \eqref{eq:Killing spinor eqn b} imply that the six-dimensional space has is locally a product of $\bbR$ with a real line bundle over a four-dimensional hyperk\"ahler base. The $\bbR$ factor is spanned by $\zeta_2^\#$ and the fiber of the line bundle by $\zeta_1^\#$, while the $\omega_\alpha$ define the hyperk\"ahler structure on the base. The full solution can be interpreted as the back-reacted geometry of an M5 brane wrapped on the line bundle fiber. Thus physically a class 1 $\SUs6$ structure should be viewed as a particular to six-dimensional generalisation of a four-dimensional hyperk\"ahler structure, that captures the geometry of the wrapped brane.

\section{Moduli of H-structures}\label{sec:moduli of H-structures}

As previously mentioned, an $\SUs{6}$ structure does not define a generalised metric and hence does not define a supergravity background. However, much as the moduli space of a Calabi--Yau locally splits into K\"ahler and complex moduli, the moduli space of a $\USp{6}$ structure splits locally into H-structure and V-structure moduli. Therefore, by studying the moduli of the H-structure, we will be able to retrieve some information about the spectrum of the effective theory on $\bbR^{4,1}$. From the classification of the previous section we have shown that H-structures in M theory are characterised by their class and that, furthermore, the class is a global notion in that it is the same at all points on the manifold. Thus we expect two different moduli spaces problems, one for class 0 and one for class 1.  

The moduli space $\mathcal{M}_{H}$ was described in \cite{Ashmore:2015joa} in terms of a hyperk\"ahler quotient of the space of $\SUs{6}$ structure by generalised diffeomorphisms. This description comes from the condition for integrability given by \ref{integrability i} in proposition \ref{prop:different integrability} as the vanishing of the triplet of moment maps $\mu_{\alpha}$. Here, we will instead exploit the structure implied by the integrability conditions \ref{integrability iii} of proposition \ref{prop:different integrability} and choose a particular ECS. This will lead to explicit statements about the moduli of the structure in terms of natural cohomology groups. It comes at the cost of losing the explicit hyperk\"ahler construction, since only one of the K\"ahler structures is manifest. However, the moduli space should be the same independent particular choice of ECS that we make. We will see that reinstating the $\SU{2}$ symmetry thus implies an interesting structure on the cohomologies defined by these different choices.

Let us first review the hyperk\"ahler geometry of the moduli space as given in \cite{Ashmore:2015joa}. Recall from the discussion around \eqref{eq:U(6) Kahler Potential} that the space of $\SUs{6}$ structures $\ZSUs$ has a natural hyperk\"ahler structure. The moduli space $\mathcal{M}_{H}$ is defined to be the space of integrable $\SUs{6}$ structures up to generalised diffeomorphisms. From condition \ref{integrability i} in proposition \ref{prop:different integrability}, integrability is given by the vanishing of the three maps $\mu_{\alpha}$ and hence
\begin{equation}
    \mathcal{M}_{H} = \{\Jpl \in \ZSUs \,|\, \mu_{\alpha} = 0\} \quotient \GDiff = \ZSUs\qqquotient \GDiff
\end{equation}
This construction, called a hyperk\"ahler quotient, keeps the hyperk\"ahler nature of $\mathcal{M}_{H}$ manifest. In fact $\ZSUs$ and hence also  $\mathcal{M}_{H}$ has the structure of a hyperk\"ahler cone and the physical moduli space of hypermultiplet scalars is given by
\begin{equation}
    \mathcal{M}_{\text{hyper}} = \mathcal{M}_{H}\quotient \bbH^{*} = \mathcal{M}_{H} \quotient (\SU{2} \times \bbR^{+})
\end{equation}
This is because there are deformations of the $\SUs{6}$ structure that do not deform the generalised metric and hence should not be regarded as physical moduli. Specifically, the $\SU{2}$ comes from the ambiguity in the definition of the orthonormal internal spinors $\eta^{i}$ in \eqref{eq:Killing Spinor Form} which can be absorbed into the definition of the external component of the full 11 dimensional spinor. Hence this $\SU{2}$ symmetry is related to the R-symmetry of the 5 dimensional $\mathcal{N} = 1$ theory. The $\bbR^{+}$ corresponds to shifting the warp factor $\warp$ by a constant, but this can be absorbed into the definition of the external flat metric. This defines the hyperk\"ahler cone structure of $\ZSUs$ which becomes the internal $\SU{2}$ symmetry of the $J_{\alpha}$ along with $\bbR^{+}$ rescalings of $\kappa$. However, this descends to $\mathcal{M}_{H}$ since the action of $\GDiff$ commutes with the $\bbH^{*}$ action.

In general, a hyperk\"ahler cone $\mathcal{M}$ can be viewed as a real cone over a tri-Sasaki space $\mathcal{L}$ \cite{Swann1991,Boyer1998}. This is an $\SO{3}$ bundle over the quarternionic K\"ahler base $\mathcal{M}\quotient \bbH^{*}$ and is defined by setting the hyperk\"ahler potential (in our case given by $\mathcal{K}$ in \eqref{eq:U(6) Kahler Potential}) to a constant. Selecting some $\U{1}\subset \SO{3}$, we can consider the quotient $\mathcal{L}\quotient \U{1}$. This defines a K\"ahler space called the \emph{twistor space}
\begin{equation}
    \twistor = \mathcal{L}\quotient \U{1} = \mathcal{M}\quotient \bbC^{*}
\end{equation}
The twistor space defines an $S^{2}$ bundle over the base space $\mathcal{M}\quotient \bbH^{*}$ via the following commuting diagram.
\begin{equation}
    \begin{tikzcd}[column sep=huge, row sep = huge]
        \mathcal{M} \arrow[r] \arrow[rd] & \twistor = \mathcal{L}\quotient \U{1} = \mathcal{M}\quotient \bbC^{*} \arrow[d] \\
        &  \mathcal{M}\quotient \bbH^{*}
    \end{tikzcd}
\end{equation}

For the moduli space of $\SUs{6}$ structures the associated twistor space $\twistor_{H}$ has a natural description in terms of ECS. Indeed, the conditions for integrability given by \ref{integrability iii} of proposition \ref{prop:different integrability} allows us to write $\mathcal{M}_{H}$ in a different, but equivalent way. Fixing some $\mathbf{u} \in S^{2}$, we have
\begin{equation}
    \mathcal{M}_{H} = \{ X \in \ZSUs \,|\, L_{u} \text{ involutive}, \; u^{\alpha}\mu_{\alpha} = 0 \} \quotient \GDiff = \ZinvSUs\qquotient \GDiff
\end{equation}
where $\ZinvSUs = \{X \in \ZSUs\,|\, L_{u} \text{ involutive}\}$. While this K\"ahler quotient construction hides the manifest hyperk\"ahler structure, we can now exploit a general result about group actions that preserve a K\"ahler structure: the space can be viewed as either a K\"ahler quotient, or a quotient by the complexified group $\GDiff_{\bbC}$\footnote{One has to be careful in defining this complexified group since the natural complexification is not well defined. What we mean by $\GDiff_{\bbC}$ is the group generated by $\rho_{V}, \mathcal{I}\rho_{V}\in \Gamma(T\ZSUs)$, where $\mathcal{I}$ is the complex structure on $\ZSUs$.} (see for example \cite{Hitchin:1986ea}). We can therefore write the moduli space of H-structures in the convenient form\footnote{In fact, one really needs to consider the space $\ZinvSUs^{\text{ps}}$ of `polystable' points in $\ZinvSUs$. This has interesting links to geometric invariant theory but we won't go into more detail. Here, we are just interested in the infinitesimal structure of the moduli space for which this technicality is not important. The links between ECS and geometric invariant theory were explored in more detail in \cite{Ashmore:2019qii,Ashmore:2019rkx}.}
\begin{equation}
    \mathcal{M}_{H} = \ZinvSUs\quotient \GDiff_{\bbC}
\end{equation}
Finally, from the discussion around \eqref{eq:equiv SU*(6) for U*(6)}, we know that choosing an ECS $L\subset E_{\bbC}$ defines the $\SUs{6}$ structure $\Jpl$ up to a complex scaling. Hence, the twistor space can be defined via ECS as
\begin{equation}\label{eq:H-structure twistor space}
    \twistor_{H} = \mathcal{M}_{H}\quotient \bbC^{*} = \{ L_{u} \text{ an ECS}\,|\, L_{u} \text{ involutive}\}\quotient \GDiff_{\bbC} = \ZECSinv\quotient \GDiff_{\bbC}
\end{equation}
where we have defined $\ZECSinv$ to be the space of integrable ECS.

The space \eqref{eq:H-structure twistor space} is now in a form that allows analysis very similar to the analysis of the moduli of conventional complex structures \cite{kodaira2006complex}. In the following sections we will use the deformation theory of ECS to get a local dimension of the twistor space around an arbitrary point. We should highlight that $\twistor_{H}$ is not the moduli space of ECS, but rather is the space of structures satisfying the additional condition that the moment map $u^\alpha\mu_\alpha$ vanishes. The moduli space of ECS is given by $\mathcal{M}_{\text{ECS}} = \ZECSinv\quotient \GDiff_{\bbR}$ which is infinite dimensional and does not have nice properties. By imposing the vanishing of the moment map, the moduli we find  are, in general, moduli of the full $\SUs{6}$ structure rather than of the associated ECS. Thus, once we have analysed the structure of $\twistor_{H}$, we should project on the $S^{2}$ fiber to recover the physical moduli space $\mathcal{M}_{\text{hyper}}$. This $S^{2}$ fiber has a natural interpretation as the $\SU{2}/\U{1} = S^{2}$ of ECS associated to any $\SUs{6}$ structure, as was laid out previously in section \ref{sec: SU(2) symmetry}. Projecting on this $S^{2}$ corresponds to removing one complex modulus from the infinitesimal analysis which we will explain in more detail in the following. Finally we note that the resulting space $\mathcal{M}_{\text{hyper}}$ should be independent of the choice of $u^\alpha$ we made in defining \eqref{eq:H-structure twistor space}.

\subsection{Deformation Theory and Moduli of \texorpdfstring{$\SUs{6}$}{SU*(6)} Structures}

The form of \eqref{eq:H-structure twistor space} means we can understand the local structure of $\twistor_{H}$ by analysing the deformation theory of ECS. By identifying deformations up to local complexified generalised diffeomorphisms, we will find a finite-dimensional result in terms of natural cohomology groups. The dimension of these gives a local dimension of $\twistor_{H}$ which we can use to find the moduli of $\mathcal{M}_{H}\quotient \bbH^{*}$ by removing a particular complex modulus associated to the $S^{2}$ fiber of $\twistor_{H} \rightarrow \mathcal{M}_{H}\quotient \bbH^{*}$. Let us start by outlining the general deformation theory.

At a point $p\in M$, the space of almost ECS is given by the coset
\begin{equation}
    Q_{\bbR^{+}\times \Us{6}} = E_{6(6)}\cdot \JUs_{0} =  \mathrm{E}_{6(6)}\quotient \Us{6} = \mathrm{E}_{6,\bbC}\quotient P
\end{equation}
where $\JUs_{0}$ is some fixed ECS and $P$ is the parabolic subgroup that stabilises $L_{1}$
\begin{equation}
    P = \Stab L_{1} = \GL{6,\bbC}\ltimes \bbC^{21}
\end{equation}
By considering all possible $p\in M$ we find that $\JUs$ must be a section of the bundle
\begin{equation}
    Q_{\bbR^{+}\times \Us{6}}\longrightarrow \mathcal{Q}_{\bbR^{+} \times \Us{6}} \longrightarrow M
\end{equation}
Infinitesimally, the deformations are given by sections of the bundle
\begin{equation}
    \mathfrak{e}_{6,\bbC}\quotient \mathfrak{p} \longrightarrow \mathfrak{Q}_{\bbR^{+}\times \Us{6}} \longrightarrow M
\end{equation}
In practice, we choose an embedding $\mathfrak{e}_{6,\bbC}\quotient \mathfrak{p} \hookrightarrow \mathfrak{e}_{6,\bbC}$. Then, given some section $A\in \Gamma(\mathfrak{Q}_{\bbR^{+}\times \Us{6}})$, we can define the deformed $L_{1}$ bundle $L_{1}'$ by
\begin{equation}
    L_{1}' = (1+\epsilon A)\cdot L_{1}
\end{equation}
for some small parameter $\epsilon \ll 1$ and we view $A$ as a map $:L_{1}\rightarrow E_{\bbC}\quotient L_{1}$. Through the embedding $\mathfrak{e}_{6,\bbC}\quotient \mathfrak{p} \hookrightarrow \mathfrak{e}_{6,\bbC}$, we get an embedding $E_{\bbC}\quotient L_{1}\hookrightarrow E_{\bbC}$.

By assumption, the original bundle $L_{1}$ is involutive and hence the intrinsic torsion vanishes. For a generic deformation $A$, $L_{1}'$ will have some non-zero intrinsic torsion that appears as an obstruction to the involutivity of the bundle with respect to the (flux-twisted) Dorfman derivative. By expanding the involutivity condition to first order in $\epsilon$, we find a map
\begin{equation}
    \dd_{2}:\Gamma(\mathfrak{Q}_{\bbR^{+}\times \Us{6}}) \longrightarrow \Gamma(W_{\tint}^{\bbR^{+}\times \Us{6}})
\end{equation}
The integrable deformations are determined by the kernel of this map. That is, $L_{1}'$ is integrable if and only if $A\in \ker \dd_{2}$.

We also have a notion of trivial deformation given by complexified generalised diffeomorphisms. To linear order, these are given by the action of the Dorfman derivative along some complexified vector $V\in \Gamma(E_{\bbC})$. That is, $L_{1}'$ is said to be a trivial deformation if
\begin{equation}
    L_{1}' = (1+\epsilon L^{F}_{V})L_{1} \quad \text{some } V\in \Gamma(E_{\bbC})
\end{equation}
This defines a second map
\begin{equation}
    \dd_{1}: \Gamma(E_{\bbC}) \longrightarrow \Gamma(\mathfrak{Q}_{\bbR^{+}\times \Us{6}})
\end{equation}
where a deformation $A$ is trivial if an only if $A\in \image \dd_{1}$. It is an easy check that any trivial deformation is involutive to linear order in $\epsilon$. Indeed,
\begin{align}
    \begin{split}
        L^{F}_{W+\epsilon L^{F}_{V}W}(W' + \epsilon L^{F}_{V}W') &= L_{W}W' + \epsilon(L^{F}_{L^{F}_{V}W}W' + L^{F}_{W}L^{F}_{V}W') + O(\epsilon^{2}) \\
        &= (1+\epsilon L^{F}_{V})L_{W}W' + O(\epsilon^{2})
    \end{split}
\end{align}
This implies that $\dd_{2}\circ \dd_{1} = 0$, and hence we have a three-term complex
\begin{equation}
    \Gamma(E_{\bbC}) \xrightarrow{ \quad \dd_{1} \quad} \Gamma(\mathfrak{Q}_{\bbR^{+}\times \Us{6}}) \xrightarrow{ \quad \dd_{2}\quad} \Gamma(W_{\tint}^{\bbR^{+}\times \Us{6}}) \label{eq:E6 deformation complex 1}
\end{equation}
where the cohomology of \eqref{eq:E6 deformation complex 1} gives the tangent space $T_{\JUs}\twistor_{H}$.

To get the physical moduli, we need to remove the modulus associated to the $S^{2}$ fiber. From section \ref{sec: SU(2) symmetry}, and particularly the discussion around \eqref{eq:general Lu}, it is clear that this $S^{2}$ is generated by $\eta (\kappa^{-1}J_{+})$ for some constant $\eta\in \bbC$ and that deformations of this kind are always integrable provided we start at a fully integrable $\SUs{6}$ structure. We must therefore remove the complex modulus associated to the image of $\kappa^{-1}J_{+}$ under the projection $\ad\tilde{F}_{\bbC} \rightarrow \mathfrak{Q}_{\bbR^{+}\times \Us{6}}$.

\subsection{Class 0 Structures}\label{sec:class 0 moduli}

A generic ECS associated to a class 0 $\SUs6$ structure is of type 0\footnote{Generic in the sense that any the space of $L_{1}$ with non-surjective projection onto $T$ are measure 0 in the Grassmannian of all $L_{1}$.} and so is of the form
\begin{equation}
    L_{1} = \ee^{\alpha+\beta}\cdot T_{\bbC} \qquad \alpha\in \Omega^{3}(M)_{\bbC},\, \beta \in \Omega^{6}(M)_{\bbC} \label{eq:E6 type 0}
\end{equation}
The conditions arising from definition \ref{def:E6 ECS} put algebraic conditions on $\alpha,\beta$ which we derived in \eqref{eq:type 0 ECS non-linear condition}. In particular, we saw that $\im \alpha$ defines an $\SL{3,\bbC}$ structure as in \cite{Hitchin00}. It is easy to see from \eqref{eq:integrability conditions} that this is an integrable $\bbR^{+}\times \Us{6}$ structure iff $\dd\alpha = 0$.

To study the deformations we can choose the following embeddings:
\begin{align}
    E_{\bbC}\quotient L_{1} &= \ext^{2}T^{*}_{\bbC}\oplus \ext^{5}T^{*}_{\bbC} \\
    \mathfrak{Q}_{\bbR^{+}\times \Us{6}} &= \ext^{3}T^{*}_{\bbC}\oplus \ext^{6}T^{*}_{\bbC}
\end{align}
Then a generic deformation of $L_{1}$ of the form \eqref{eq:E6 type 0} will be
\begin{equation}
    L_{1}' = (1+\epsilon(a+b))L_{1} = \ee^{\alpha+\beta+\epsilon (a+\tilde{b})}T_{\bbC}
\end{equation}
where the formula on the right hand side is to be taken to first order in $a,b$, and where $\tilde{b} = b- \frac{1}{2}a\wedge\alpha$. From this it is clear that
\begin{equation}
    L_{1}' \text{ integrable} \quad \Leftrightarrow \quad \dd a = 0
\end{equation}
since the condition $\dd b=0$ is trivial. We then want to consider when a deformation is trivial. That is, when we can write it in the form\footnote{As we saw in section \ref{sec:closer analysis of ECS}, the flux must be in the trivial cohomology class for class 0 structures. Hence we can use the untwisted Dorfman derivative in this case.}
\begin{equation}
    L_{1}' = (1+\epsilon L_{V})L_{1} \qquad \text{some } V\in E_{\bbC}
\end{equation}
Writing $V = \ee^{\alpha + \beta}(V+\omega+\sigma)$, we find that the trivial $L_{1}'$ can be written as
\begin{equation}
    L_{1}' = \ee^{\alpha+\beta -\dd\omega - \dd\tilde{\sigma}}T_{\bbC}
\end{equation}
where $\tilde{\sigma} = \sigma + \frac{1}{2}\alpha\wedge \omega$. Hence, the deformation is trivial if and only if $a,b$ are exact. From this it is clear to see that the deformations are counted by the complex de Rham cohomology groups
\begin{equation}
    T_{\JUs}\twistor_{H} = H^{3}(M,\bbC)\oplus H^{6}(M,\bbC)
\end{equation}

We now must remove the modulus associated to $J_{+}$ to find the physical moduli. To do so, we need to know how $\kappa^{-1}J_{+}$ projects onto $\mathfrak{Q}_{\bbR^{+}\times \Us{6}} = \ext^{3}T^{*}_{\bbC}\oplus \ext^{6}T^{*}_{\bbC}$. Fortunately, this projection is quite simple and we just take the 3 and 6-form components of $\ee^{-\alpha-\beta}\kappa^{-1}J_{+}\ee^{\alpha+\beta}$. Moreover, since we have chosen our representative ECS to be class 0, one can show that the 6-form component in particular never vanishes. Hence, the associated modulus we should remove is a particular combination of classes in $H^{3}$ and $H^{6}$ which are related through the $\SL{3,\bbC}$ structure, $\im\alpha$. We can use this element to write any 6-form deformation in terms of 3-form deformations and write the physical moduli as
\begin{equation}
    \text{moduli} = \bigl(H^{3}(M,\bbC)\oplus H^{6}(M,\bbC)\bigr)\quotient [J_{+}] \simeq H^{3}(M,\bbC)
\end{equation}
As a sanity check, $H^{3}(M,\bbC)$ has a natural symplectic structure and hence must be $2n$ complex dimensional. It is therefore $4n$ real dimensional - the required dimension of a quarternionic K\"ahler space.

One may ask whether we would obtain the same result if we had taken our representative ECS to be of type 3. This will follow from the analysis we do in the next section.

\subsection{Class 1 Structures}\label{sec:E6 class 1 Moduli}

For class 1 solutions, all associated ECS are of type 3. We therefore want to understand the deformations of 
\begin{equation}
    L_{1} = \ee^{\alpha+\beta}\cdot(\Delta\oplus \mathcal{F}^{2}_{1}(\Delta))
\end{equation}
We will keep $\Delta$ general for now, and hence the results of this section apply to class 0 type 3 as well. For convenience, we will define a dual filtration of multivectors $\mathfrak{F}^{k}_{p}(\Delta)\subset \ext^{k}T_{\bbC}$ given by $\xi\lrcorner\phi = 0$ for all $\xi \in \mathfrak{F}^{k}_{p}(\Delta)$, and for all $\phi \in \mathcal{F}^{k}_{p}(\Delta)$. It is possible to show that one can choose the following for the quotient spaces.
\begin{align}
    E_{\bbC}\quotient L_{1} 
    &= \left(T_{\bbC} \quotient \mathfrak{F}^{1}_{0} \right) 
    \oplus \left( \ext^{2}T^{*}_{\bbC}\quotient \mathcal{F}^{2}_{1} \right) 
    \oplus \ext^{5}T^{*}_{\bbC} 
    \label{eq:type 3 quotient 1} \\
    \mathfrak{Q}_{\bbR^{+}\times \Us{6}} 
    &= \left[ \left(T_{\bbC}\quotient \mathfrak{F}^{1}_{0}\right) 
    \otimes \left(T^{*}_{\bbC}\quotient \mathcal{F}^{1}_{0}\right) \right] 
    \oplus \left( \ext^{3}T_{\bbC} \quotient \mathfrak{F}^{3}_{2} \right) 
    \oplus \left(\ext^{3}T^{*}_{\bbC} \quotient \mathcal{F}^{3}_{1}\right) 
    \oplus \ext^{6}T^{*}_{\bbC} 
    \label{eq:type 3 quotient 2}
\end{align}
While these spaces may seem confusing at first, things are made easier by choosing some space $\Sigma \subset T_{\bbC}$ that is complement to $\Delta$. If the structure is class 0 then $\Delta\cap\bar{\Delta}=0$ and so there is a canonical choice of $\Sigma = \bar{\Delta}$. Interestingly, as we will see in section \ref{sec: Minkowski Backgrounds}, there is a canonical choice of $\Sigma$ even when $\Delta\cap\bar{\Delta} \neq 0$. We can use this split $T_{\bbC}=\Delta \oplus \Sigma$ to simplify the quotients to
\begin{align}
    E_{\bbC}\quotient L_{1} &= \Sigma \oplus (\Sigma^{*}\otimes \Delta^{*})\oplus \ext^{2}\Delta^{*} \oplus \ext^{5}T^{*} \\
    \mathfrak{Q}_{\bbR^{+}\times \Us{6}} &= [\Sigma\otimes \Delta^{*}] \oplus \ext^{3}\Sigma \oplus \ext^{3}\Delta^{*} \oplus (\ext^{2}\Delta^{*}\otimes \Sigma^{*}) \oplus \ext^{6}T^{*}
\end{align}
The final result should be independent of this choice of splitting and so we will work with the general form \eqref{eq:type 3 quotient 1}, \eqref{eq:type 3 quotient 2}.

An important consideration to make in the type 3 case is the possibility of non-trivial flux. As we saw in proposition \ref{prop: ECS integrability E6}, the complex flux locally defined by $\dd\alpha$ does not need to vanish. Instead, it falls into some, possibly non-trivial, cohomology class in $H^{4}(M,\mathcal{F}^{\bullet}_{1})$. This, in turn, implies that the physical flux $F$ need not be in a trivial cohomology class. As we will discuss in the following section, the cohomology class of $F$ represents something physical, related to the number of M5 branes wrapping a cycle. In this case, it is easiest to work with the flux-twisted Dorfman derivative\footnote{The expression for the flux-twisted Dorfman derivative is given in appendix \ref{app:conventions}. This formulation of generalised geometry is equivalent to the original formulation with flux-twisted bundles.}. We will find that the moduli are therefore counted by the cohomology of a `flux-twisted' differential. To find such a differential that squares to 0, it will be convenient to work with the complex-flux twisted Dorfman derivative $L^{F_{\bbC}}_{V}$ and consider deformations of the untwisted bundle $\tilde{L}_{1} = \Delta \oplus \mathcal{F}^{2}_{1}(\Delta)$. This has the same quotient bundles as \eqref{eq:type 3 quotient 1}, \eqref{eq:type 3 quotient 2}.

Consider a general deformation element $R=r + \trivector+\theta+\tau \in \Gamma(\mathfrak{Q}_{\bbR^{+}\times \Us{6}})$ where $r\in  \Gamma\left[\left(T_{\bbC}\quotient \mathfrak{F}^{1}_{0}\right) \otimes \left(T^{*}_{\bbC}\quotient \mathcal{F}^{1}_{0}\right)\right]$, $\trivector\in \Gamma(\ext^{3}T_{\bbC}\quotient \mathfrak{F}^{3}_{2})$, etc. We then have the deformed bundle
\begin{equation}
    \tilde{L}'_{1} = (1+R)\tilde{L}_{1} = \ee^{\theta+\tau}(1+r+\trivector)\cdot(\Delta + \mathcal{F}^{2}_{1}(\Delta))
\end{equation}
What are the conditions for $\tilde{L}_{1}'$ to be involutive under $L^{F_{\bbC}}_{V}$? We will leave the detailed calculation to the appendix and for now just note that the moduli are controlled by two cohomology groups related to $\Delta$. First, since $\Delta$ is involutive with respect to the Lie bracket, this defines a Lie algebroid and has an associated differential
\begin{equation}
    \dd_{\Delta} : \ext^{p}\left(T^{*}_{\bbC}\quotient \mathcal{F}^{1}_{1}\right) \longrightarrow \ext^{p+1}\left(T^{*}_{\bbC}\quotient \mathcal{F}^{1}_{1}\right) \qquad \dd_{\Delta}^{2} = 0
\end{equation}
If we take $i:\Delta \hookrightarrow T_{\bbC}$ to be the natural inclusion, then $i^{*}:T^{*}_{\bbC} \rightarrow (T^{*}_{\bbC}\quotient \mathcal{F}^{1}_{1})$. We can define the differential above via $i^{*}\circ \dd = \dd_{\Delta}\circ i^{*}$, where we take the natural extension of $i^{*}$ to $\ext^{p}T^{*}_{\bbC}$. This will define cohomology groups which we will denote by $H^{p}_{\Delta}$. We will further denote by $H^{p}_{\Delta}(M,B)$ the cohomology group of $\ext^{p}(T^{*}_{\bbC}\quotient \mathcal{F}^{1}_{1})$ evaluated in the bundle $B$.

The second cohomology group of interest is defined in terms of the filtration $\mathcal{F}^{p}_{k}(\Delta)$. Recall that $\dd:\mathcal{F}^{p}_{k}(\Delta)\rightarrow \mathcal{F}^{p+1}_{k}(\Delta)$ if $\Delta$ is an integrable distribution. Hence, the de Rham differential descends to the following complex.
\begin{equation}
    \left(\ext^{1}T^{*}\quotient \mathcal{F}^{1}_{k}\right) \xrightarrow{\; \dd\;}\left(\ext^{2}T^{*}\quotient \mathcal{F}^{2}_{k}\right) \xrightarrow{\;\dd\;} ... \xrightarrow{\;\dd\;} \left(\ext^{6}T^{*}\quotient \mathcal{F}^{6}_{k}\right)
\end{equation}
We then denote the cohomology groups associated to this complex as $H^{p}(M,\ext^{\bullet}T^{*}\quotient \mathcal{F}^{\bullet}_{k})$. (It is probably worth noting that neither of these cohomologies are the basic cohomology of foliated spaces defined in e.g. \cite{habib2013modified}.)

After a lengthy calculation, one finds that the deformations are counted by the cohomology of a differential that we will label $\dd_{\Delta,F}$ which creates the following complex
\begin{align}
\label{eq:dDeltaF complex}
\begin{split}
    & \Gamma\left( \left(T_{\bbC}\quotient \mathfrak{F}^{1}_{0}\right) \oplus \left(\ext^{2}T^{*}_{\bbC}\quotient \mathcal{F}^{2}_{1}\right) \oplus \ext^{5}T^{*}_{\bbC} \right) \\
    \xrightarrow{\quad \dd_{\Delta,F} \quad} \, & \Gamma\left(\left( \ext^{3}T_{\bbC}\quotient \mathfrak{F}^{3}_{2} \right) \oplus \left[ \left(T_{\bbC}\quotient \mathfrak{F}^{1}_{0}\right)\otimes \left(T^{*}_{\bbC}\quotient \mathcal{F}^{1}_{0}\right)\right] \oplus \left(\ext^{3}T^{*}_{\bbC}\quotient \mathcal{F}^{3}_{1}\right) \oplus \ext^{6}T^{*}_{\bbC} \right) \\
    \xrightarrow{\quad \dd_{\Delta,F} \quad} \, & \Gamma\left( \left[\left( \ext^{3}T_{\bbC}\quotient \mathfrak{F}^{3}_{2} \right)\otimes \left( T^{*}_{\bbC}\quotient \mathcal{F}^{1}_{0} \right)\right] \oplus \left[ \left(T_{\bbC}\quotient \mathfrak{F}^{1}_{0}\right)\otimes \ext^{2}\left(T^{*}_{\bbC}\quotient \mathcal{F}^{1}_{0}\right)\right] \oplus \left(\ext^{4}T^{*}_{\bbC}\quotient \mathcal{F}^{4}_{1}\right) \right)
\end{split}
\end{align}
If we take $R = \trivector+r+\theta+\tau \in \Gamma(\mathfrak{Q}_{\bbR^{+}\times \Us{6}})$, and $V = v+\omega + \sigma \in \Gamma(E_{\bbC}\quotient L_{1})$, then the closure conditions are\footnote{The definition of $j\trivector\lrcorner j^{2}F_{\bbC}$ can be found in appendix \ref{app:conventions}}
\begin{align}
    0 &= \dd_{\Delta} \trivector \label{eq:type 3 closure 1} \\
    0 &= \dd_{\Delta}r - j\trivector\lrcorner j^{2}F_{\bbC} \label{eq:type 3 closure 2} \\
    0 &= \dd\theta - r\cdot F_{\bbC} \label{eq:type 3 closure 3}
\end{align}
and the exactness conditions are
\begin{align}
    r &= \dd_{\Delta}v \label{eq:type 3 exact 1} \\
    \theta &= \dd\omega - v\lrcorner F_{\bbC} \label{eq:type 3 exact 2} \\
    \tau &= \dd\sigma + \omega \wedge F_{\bbC} \label{eq:type 3 exact 3}
\end{align}
We are implicitly taking projections onto relevant quotient spaces where needed above. It is an easy check to see that $\dd_{\Delta,F}^{2} = 0$. If $F_{\bbC}$ is in a trivial cohomology class in $H^{4}(M,\mathcal{F}^{\bullet}_{1})$, which is true for class 0 backgrounds in particular, then the deformations are counted by
\begin{equation}\label{eq:type 3 deformation structure}
    T_{\JUs}\twistor_{H} = H^{0}_{\Delta}(M,\ext^{3}T_{\bbC}\quotient \mathfrak{F}^{3}_{2}) \oplus H^{1}_{\Delta}(M,T_{\bbC}\quotient \mathfrak{F}^{1}_{0}) \oplus H^{3}(M,\ext^{\bullet}T_{\bbC}^{*} \quotient \mathcal{F}^{\bullet}_{1})\oplus H^{6}_{\dd}(M,\bbC)
\end{equation}

To find the physical moduli, we need to remove the modulus associated $J_{+}$. Again, this is done by finding the projection of $\kappa^{-1}J_{+}$ onto the space \eqref{eq:type 3 quotient 2}. The precise form of this projection is complicated but one can show that the projection onto the following space is always non-vanishing
\begin{equation}
    \ext^{3}\Delta^{*} \simeq \ext^{3}T^{*}\quotient \mathcal{F}^{3}_{2} \subseteq \ext^{3}T^{*}\quotient \mathcal{F}^{3}_{1} \subseteq \mathfrak{Q}_{\bbR^{+}\times \Us{6}}
\end{equation}
In the first equality we have chosen a decomposition $T=\Delta \oplus \Sigma$. We can therefore use $J_{+}$ to remove deformations along $\ext^{3}\Delta^{*}$ to obtain the physical moduli of the background.

\subsection{Exceptional Dolbeault Operators}\label{sec:E6 generic moduli}

The moduli found in the previous sections determine the moduli of all structures of constant type. This works well for class 1 $\SUs6$ structures where the notion of type is unambiguous. However, as we noted in section \ref{sec: SU(2) symmetry}, the type of an ECS associated to a class 0 $\SUs6$ structure is not uniquely specified. Although it is generically type 0 there are two points on the $S^2$ of structures where it becomes type 3.
% Moreover, we have not yet ruled out the possibility of class/type changing solutions, although we shall see that this is not possible for Minkowski\footnote{It is possible to show that AdS backgrounds are in fact class-changing but can be described with a global type 3 structure. Nonetheless, this section may provide a possible way to analyse AdS moduli as well.} backgrounds in section \ref{sec: Minkowski Backgrounds}. 
We would therefore like to characterise the moduli in a way that is independent of the type of particular ECS used to analyse the problem, and treats both class 0 and class 1 in a single formalism. This will lead us to defining an `exceptional Dolbeault operator' whose cohomology groups then capture the moduli.

To allow analysis for arbitrary type, we would like to be able to find the deformations of $L_{u}$ for arbitrary $\mathbf{u}\in S^{2}$. Recall from proposition \ref{prop:different integrability} that integrability can be defined in terms of any $L_{u}$ and hence our results should be independent of this choice. Being able to find the deformations requires a `nice' choice of embedding $E_{\bbC}\quotient L_{u} \hookrightarrow E_{\bbC}$, and $\mathfrak{e}_{6,\bbC}\quotient \mathfrak{p}_{u} \hookrightarrow \mathfrak{e}_{6,\bbC}$. Fortunately, one such nice embedding is naturally selected by the $\Us{6}$ structure independent of class. Decomposing into eigenbundles of $\JUs_{u}$, one finds
\begin{align}
    E_{\bbC} & = \mathfrak{X}_{1}\oplus \mathfrak{X}_{-1}\oplus \ext^{2}\mathfrak{X}^{*} \\
    \ad\tilde{F}_{\bbC} &= \ad P_{\bbR^{+}\times \Us{6}} \oplus \ext^{3}\mathfrak{X}^{*}_{-1} \oplus \ext^{6}\mathfrak{X}^{*}_{-2} \oplus \ext^{3}\mathfrak{X}^{*}_{1}\oplus \ext^{6}\mathfrak{X}^{*}_{+2} \\
    W_{\tint}^{\bbR^{+}\times \Us{6}} &= \ext^{4}\mathfrak{X}^{*}_{-2} \oplus \ext^{4}\mathfrak{X}^{*}_{2}
\end{align}
where $\mathfrak{X}$ is a bundle that transforms in the $\rep{6}$ of $\SUs{6}$. The subscript denotes the $\U{1}$ charge under $\JUs_{u}$, so $\mathfrak{X}_{1}\simeq L_{u}$. A natural choice of embeddings is then
\begin{equation}
    E_{\bbC}\quotient L_{u} = \ext^{2}\mathfrak{X}^{*} \oplus \ext^{5}\mathfrak{X}^{*}_{-1} \qquad \mathfrak{Q}_{\bbR^{+}\times \Us{6}} = \ext^{3}\mathfrak{X}^{*}_{-1} \oplus \ext^{6}\mathfrak{X}^{*}_{-2}
\end{equation}

We assume that we start from a fully integrable $\SUs{6}$ structure and hence there exists a torsion-free compatible connection $D$. Using \eqref{eq:torsion} with vanishing torsion, we know that we can replace the definitions of $\dd_{1},\dd_{2}$ in terms of the Dorfman derivative with expressions involving $L_{V}^{D}$. This means that we can write the maps $\dd_{1},\dd_{2}$ in terms of the connection $D$. Moreover, viewing $D:\Gamma(\mathcal{T})\rightarrow \Gamma(E^{*}\otimes \mathcal{T})$, we can decompose $E^{*}$ into $\JUs_{u}$ eigenbundles. The compatibility of $D$ implies that it is consistent to define a decomposition of $D$ as
\begin{equation}
    D = D_{u} + D_{-u} + D_{0}
\end{equation}
where $D_{nu} = \pi_{n}D$ where $\pi_{n}$ is the projection of $E^{*}$ to the subspace with $\JUs_{u}$ charge $n\ii$. Note that, while $D_{\pm u}$ depend on the choice of $L_{u}\subset L_{1}\oplus L_{-1}$, the operator $D_{0}$ is independent of that choice.

With these decompositions, we find that the complex \eqref{eq:E6 deformation complex 1} can be written
\begin{equation}
\begin{tikzcd}[column sep=huge, row sep = huge]
\Gamma(\ext^{2}\mathfrak{X}^{*})  \arrow[r, "D_{-u}"] & \Gamma(\ext^{3}\mathfrak{X}^{*}_{-1}) \arrow[r, "D_{-u}"] & \Gamma(\ext^{4}\mathfrak{X}^{*}_{-2}) \\
\Gamma(\ext^{5}\mathfrak{X}^{*}_{-1}) \arrow[r,"D_{-u}"] \arrow[ur,"D_{0}"] & \Gamma(\ext^{6}\mathfrak{X}^{*}_{-2}) \arrow[ur,"D_{0}"] &
\end{tikzcd} \label{eq:E6 deformation complex 2}
\end{equation}
Note that the involutivity of $L_{u}$ implies that $D_{-u}^{2}=0$. In fact, it is possible to show that $L_{u}$ defines a Lie algebroid and that $D_{-u}$ is the associated differential
\begin{equation}
    D_{-u}:\ext^{p}\mathfrak{X}^{*}_{q} \longrightarrow \ext^{p+1}\mathfrak{X}^{*}_{q-1}
\end{equation}
In full generality, not much can be said about the cohomology of \eqref{eq:E6 deformation complex 2} without more knowledge of the maps $D_{-u},D_{0}$. However, if we make the following assumption, we can give a generic result about the moduli of the $\SUs{6}$ structures.
\begin{defn}
$D_{0}$, $D_{-u}$ are said to satisfy the \emph{exceptional $\del\delb$-lemma} if they satisfy the following
\begin{equation}
    \image D_{0}\cap \ker D_{-u} \subseteq \image D_{-u}D_{0}
\end{equation}
\end{defn}
\noindent We show in appendix \ref{app:ECS in E6} that provided the exceptional $\del\delb$-lemma holds, and $D_{0}$ is a cochain homomorphism, then the cohomology $\mathcal{H}$ of the complex \eqref{eq:E6 deformation complex 1} is given by the cohomology of $D_{-u}$. More precisely we have
\begin{equation}\label{eq:generic moduli}
    \mathcal{H} = H^{3}_{D_{-u}}\oplus H^{6}_{D_{-u}}
\end{equation}
Recall that the cohomology of \eqref{eq:E6 deformation complex 1} is precisely the tangent space to the twistor space $\mathcal{H}=T_{\JUs}\twistor_{H}$. To find the physical moduli, we therefore need to remove the modulus associated to $J_{+}$. This is particularly easy with the embeddings chosen since $\ext^{6}\mathfrak{X}^{*}_{-2}$ is precisely the line bundle generated by $J_{+}$. Hence, the deformations associated to $J_{+}$ are simply $H^{6}_{D_{-u}}$ and we get the following result

\begin{propn}
Provided a background satisfies the generalised $\del\delb$-lemma, the hypermultiplet moduli are given by $H^{3}_{D_{-u}}$
\end{propn}

The definition of $D_{-u}$ as the differential associated to the Lie algebroid $L_{u}$ mirrors the properties of the Dolbeault operator $\del$ with the complex distribution $T^{1,0}$. Moreover, through complex conjugation, one can show that $D_{u}$ squares to 0, is the differential associated to $L_{-u}$, and $D_{u} = \overline{D}_{-u}$. We therefore make the following definition.

\begin{defn}\label{def:exceptional Dolbeault}
The \emph{exceptional Dolbeault operators} associated to the integrable ECS $\JUs_{u}$ are the operators $D_{\pm u}$.
\end{defn}

\noindent It is interesting to note that, if the $\SUs{6}$ structure is fully integrable, then $L_{u}$ is involutive for all $\mathbf{u}\in S^{2}$. This implies that we have a set of differentials $D_{u}$ labelled by $\mathbf{u}\in S^{2}\simeq \mathbb{CP}^{1}$. Moreover, as the analysis above was independent of the choice of $\mathbf{u}$, these differentials should be quasi-isomorphic\footnote{This should at least hold in the sense that $H^{3}_{u} \cong H^{3}_{\tilde{u}}$}. Finally, we note that, unlike conventional Dolbeault operators, we do not have $D_{(u}D_{-u)} = 0$. Instead, using the fact that any connection must satisfy $D\times_{N} D = 0$, we have
\begin{equation}
    (D_{u})^{2} = 0 \qquad (D_{-u})^{2} = 0 \qquad D_{u}D_{-u} + D_{-u}D_{u} \sim (D_{0})^{2}
\end{equation}
This should be expected as $L_{u}\oplus L_{-u}$ does not define a Lie algebroid, and hence we cannot form a differential out of $D_{u} + D_{-u}$.

\subsubsection{Example: Calabi--Yau and Class 0}

We return to the explicit example of compactification on a Calabi--Yau. Following the method set out above, we decompose $E_{\bbC}$, $\ad \tilde{F}_{\bbC}$ into eigenspaces of $\JUs$. This is outlined in appendix \ref{app:ECS in E6} but for now, we just note that there is an isomorphism between this complex and the following, using the holomorphic three-form $\Omega$.
\begin{equation}\label{eq:CY exceptional Dolbeault}
    \begin{tikzcd}[column sep=huge, row sep = huge]
        \Omega^{2}(M)_{\bbC} \arrow[r, "\del"] & \Omega^{3}(M)_{\bbC} \arrow[r, "\del"] & \Omega^{4}(M)_{\bbC} \\
        \Omega^{5}(M)_{\bbC} \arrow[r, "\del"] \arrow[ru, "D_{0}"] & \Omega^{6}(M)_{\bbC} \arrow[ru, "D_{0}"]
    \end{tikzcd}
\end{equation}
where $D_{0} = \Omega^{\#}\lrcorner\delb + \bar{\Omega}^{\#}\lrcorner \del$. One can show that this satisfies the generalised $\del\delb$-lemma and hence the moduli are counted by
\begin{equation}
    H^{3}_{\del}(M)
\end{equation}
Note that this contains all the hypermultiplet moduli of deformations of the Calabi--Yau manifolds, namely the complex structure moduli $H^{2,1}_{\delb}$ and the deformations of the three-form and six-form potential $A$ and $\tilde{A}$. The latter lie in de Rham cohomology classes $H^3_\dd(M,\bbR)$ and $H^6_\dd(M,\bbR)$. Using the holomorphic three-form these can be associated with $H^{3,0}_{\delb}$, $H^{0,3}_{\delb}$ and $H^{1,2}_{\delb}$ thus filling out $H^3_\dd(M,\bbR)$. 

This was calculated for $\JUs = \tilde{J}_{3}$, but as was noted above, we should be able to do the analysis for arbitrary $\mathbf{u}\in S^{2}$. Since the Calabi--Yau is class 0, the generic ECS for a Calabi--Yau is in fact type 0. One can use this to show that the generic differential $D_{u}$ will be quasi-isomorphic to the de Rham differential. That is
\begin{equation}
     D_{u} \sim \begin{cases} 
        \dd & u^{3} \neq \pm 1 \\
        \del & u^{3} = 1 \\
        \delb & u^{3} = -1
    \end{cases}
\end{equation}
Therefore, either using the results of section \ref{sec:class 0 moduli} or the results from \eqref{eq:generic moduli}, at a generic point $\mathbf{u}\in S^{2}$, the moduli will be counted by $H^{3}_{\dd}(M,\bbC)$. The fact that these calculations give the same result, i.e. the quasi-isomorphism of the various $D_{u}$, is equivalent to Hodge's Theorem on a Calabi--Yau manifold.

These arguments can be easily extended to a generic class 0 background, which by corollary \ref{cor:class 0 classification} differs from the above by a 6-form and an irrelevant $\ER{6}$ twist. It is easy to then incorporate the 6-form into the isomorphisms \eqref{eq:CY exceptional Dolbeault} and into the definition $D_{0}$. The generalised $\del\delb$-lemma is then equivalent to the conventional $\del\delb$-lemma for the associated $\SL{3,\bbC}$ structure. If this is satisfied then $H^{3}_{D_{-u}}$ equals $H^{3}_{\del}(M) = H^{3}_{\dd}(M,\bbC)$ depending on the choice of $\mathbf{u}$.

\section{Hypermultiplet Moduli for Minkowski Backgrounds}\label{sec: Minkowski Backgrounds}

We have showed in the previous section how to calculated the infinitesimal moduli for an arbitrary integrable M theory H-structure $J_\alpha$. As we discussed in section \ref{sec: ECY review} this should allow us to calculate the hypermultiplet moduli of a general supersymmetric Minkowski compactification. 

Recall that the full supersymmetric background includes a compatible V-structure defined by a generalised vector $K$. Together $(J_\alpha,K)$ define a generalised metric encoding the physical metric and form-field potentials, and supersymmetry implies the background satisfies the supergravity equations of motion. However, there are well-known no-go theorems \cite{Maldacena:2000mw,Giddings:2001yu,Gauntlett:2003cy,Gauntlett:2002sc} that exclude compact solutions with non-zero flux, so that the only allowed compact background is, in our case, a Calabi--Yau manifold, although non-compact backgrounds are also of significant interest, such as for geometrical engineering. 

The basic way to avoid the no-go theorems is to include sources for the fluxes coming from branes and orientifold planes. Thus generically we should consider deformations on spaces with boundaries where the sources have been removed. However, it is also possible that the sources enter only in the V-structure equations, such that the H-structure remains well defined even at the source. To see how this can work, recall that, in terms of bilinears, a generic type 0 structure had the form 
\begin{equation}
     L_{u} = \ee^{\alpha+\beta}\cdot T_{\bbC}
\end{equation}
with $\alpha$ and $\beta$ given by 
\begin{equation}
\begin{aligned}
    \alpha &= \left(A+\re(\Lambda_u/f_u)\right) + \ii \im(\Lambda_u/f_u) \\
    \beta &= \tilde{A} - \tfrac{1}{2}A\wedge (\Lambda_u/f_u)
\end{aligned}
\end{equation}
We see that the ECS does not determine the gauge potentials $A$ and $\tilde{A}$ but only combinations of potential and bilinear such as $A+\re(\Lambda_u/f_u)$. It is only once one specifies $K$ that the separate terms are picked out. Thus the solution may be singular such that $A$ and $\re(\Lambda_u/f_u)$ both diverge but $\alpha$ remains finite and hence the ECS remains well defined. The H-structure is well defined, but there is no compatible supersymmetric V-structure that is not divergent at some point. If this is the case then one can calculate the hypermultiplet moduli without having to make any excision of the sources. 

Recall that the general analysis of \cite{Gauntlett:2004zh} showed that the local supersymmetric solutions split into two classes, that directly correspond to class 0 and class 1 $\SUs6$ structures. The former class includes the fluxless Calabi--Yau background but more generally can be considered as a deformation to `Calabi--Yau with flux' since class 0 backgrounds define an integrable conventional $\SL{3,\bbC}$ structure (albeit non-metric-compatible). The second class of supersymmetric solutions could actually be reduced to solving for a single function, and have the interpretation of the back-reacted geometry of an M5-brane wrapped on a circle fibered over a five dimensional space of the form $M_{\text{HK}}\times \bbR$ where $M_{\text{HK}}$ is a four-dimensional hyperk\"ahler space. The two classes are distinguished by $\sin\theta = f = 0$ for class 1/M5-brane, and otherwise the background is class 0/Calabi--Yau with flux\footnote{As was observed in \cite{Gauntlett:2004zh}, the full set of supersymmetry equations sets $f=0$ always. However, this cannot be seen from integrability of the $\SUs{6}$ structure alone and so we shall keep it general.}. 
% Note that the differential condition \eqref{eq:Killing spinor eqn} imply that the class cannot change at some point on the manifold without the warp factor $\ee^{2\warp}$ diverging and hence the solution becoming singular. 

In section \ref{sec:moduli of H-structures} we saw how to calculate the infinitesimal  moduli for global class 0 and 1 in terms of cohomologies, and so we can simply apply those results here to find the hypermultiplet degrees of freedom for the Minkowski backgrounds. As noted in corollary \ref{cor:class change}, we cannot have smooth class-changing structures so the global analysis of section \ref{sec:moduli of H-structures} is valid away from singular points. In fact, it may be of slightly broader applicability in that the corresponding ECS may be smooth and non-singular.  If we choose the corresponding ECS as type 3 at both the class 0 and class 1 points, the moduli space analysis is still then captured by the discussion in section \ref{sec:E6 class 1 Moduli} though now with a distribution $\Delta$ that changes class.

\subsection{Class 0}

Recall from proposition \ref{prop:class0}, an integrable class 0 $\SUs6$ structure defines an integrable conventional $\SL{3,\bbC}$ structure $\tilde{\Omega}$. Furthermore, the generic ECS associated to a class 0 $\SUs{6}$ structure is type 0 and so from the analysis of section \ref{sec:class 0 moduli} we have 
\begin{equation}\label{eq:class 0 type 0 moduli}
    \text{moduli} \; = \; H^{3}_{\dd}(M,\bbC)
\end{equation}
We see that the physical hypermultiplet moduli space is locally the (complexified) moduli space of $\SL{3,\bbC}$ structures. One can view this as deformations of the complex structure associated to $\tilde{\Omega}$, the constant $b_0$ and of the closed three- and six-forms $c$ and $\tilde{c}$ that appear in the proposition. 

We emphasise that this result was obtained without introducing a specific $V$-structure $K$ and hence is independent of the precise value of the flux. Instead, we only required integrability of the $\SUs{6}$ structure which constrains the flux to be $\dd$-exact. Different choice of $K$ will give different values of the flux within the trivial cohomology class. Provided the background admits a compatible K\"ahler metric, one could choose $K$ such that the fluxes vanish and the background is genuinely Calabi--Yau (i.e. taking $\sin\theta \equiv -1$). As discussed at the end of section \ref{sec:E6 generic moduli}, \eqref{eq:class 0 type 0 moduli} then gives the expected result. On the other hand, one could choose $K$ such that the fluxes do not vanish and hence we are necessarily away from the Calabi--Yau solution. Despite this, we find that the hypermultiplet moduli are given by the same cohomology groups as in the fluxless case. This remarkable fact is non-trivial as one cannot smoothly deform a Calabi--Yau solution to one with flux and so there is no guarantee that the moduli will be the same.

In finding the moduli, we could have alternatively deformed around the type 3, class 0 ECS. In this case, we could use the results of section \ref{sec:E6 class 1 Moduli} to write the moduli in terms of the cohomology of the differential associated to the $\Delta = T^{1,0}$, i.e. the Dolbeault operator $\del$. The moment map condition implies that the complex flux $F_{\bbC}$ is exact and hence we can find an exact decomposition into Dolbeault cohomology groups. We find
\begin{align}\label{eq:class 0 type 3 moduli}
    T_{\JUs}\twistor_{H} \; &=\; H^{0}_{\delb}(M,\ext^{3,0}T) \oplus H^{1}_{\delb}(M,T^{1,0}) \oplus H^{1,2}_{\delb}(M) \oplus H^{0,3}_{\delb}(M) \oplus H^{3,3}_{\delb}(M) \\
    &= \bigoplus_{k = 0}^{3} H^{k,3-k}_{\delb}(M) \oplus H^{3,3}_{\delb} (M)
\end{align}
In the second line, we have formed isomorphisms using the holomorphic 3-form $\tilde{\Omega}$. As argued at the end of section \ref{sec:E6 class 1 Moduli}, the physical moduli are found by removing deformations along $\ext^{0,3}T^{*}$. Therefore, the physical moduli are
\begin{equation}
    \text{moduli} = \bigoplus_{k=1}^{3}H^{k,3-k}_{\delb}(M) \oplus H^{3,3}_{\delb}(M)
      = \bigoplus_{k=0}^{3}H^{k,3-k}_{\delb}(M)
\end{equation}
where again we have used an isomorphism induced by the holomorphic 3-form $\tilde{\Omega}$ in the last expression. 

Strikingly, the statement that these two calculations of the moduli are equal appears to suggest that these backgrounds should always satisfy some kind of Hodge theorem. However, this is not quite correct due to a technicality in the way one derives the moduli from the moment map picture. One needs that the K\"ahler metric on the space of structures is non-degenerate transverse to the action of $\GDiff$ \cite{Hitchin00,Hitchin02}. One sufficient condition is to say the backgrounds satisfies the $\del\delb$-lemma which is enough to guarantee the isomorphism of \eqref{eq:class 0 type 0 moduli} and \eqref{eq:class 0 type 3 moduli}. In the same way, that the existence of a K\"ahler metric implies the conventional $\del\delb$-lemma, we expect that the existence of a compatible supersymmetric $K$ is also sufficient.

\subsection{Class 1: \texorpdfstring{$\sin\theta=f= 0$}{sin(theta)=f=0}}

In this case, the internal spinors are of fixed norm and hence the local $\SU{2}$ structure is in fact global. As discussed in section \ref{sec:SU*(6) in local SU(2)}, the $\SUs{6}$ structure defines a $\SU2\times\GL{1,\bbR}\times\bbR$ structure satisfying the equations \eqref{eq:class 1 inv}-\eqref{eq:class 1 flux 2}. All the corresponding ECS are of type 3. We can then use the results of section \ref{sec:E6 class 1 Moduli} to find the moduli of such a background as the cohomology of the complex \eqref{eq:dDeltaF complex} for the differential $\dd_{\Delta, F}$. Crucially  $F$ does not need to be in a trivial cohomology class and hence we cannot decompose the moduli into the cohomology of $\dd_{\Delta}$ alone. 

% so that it aligns with the triplet of 2-forms $\omega_{\alpha}$ and we find any associated ECS has the form
% \begin{align}
%     L_{u} &= \ee^{A+\tilde{A}}\ee^{-\ii\omega_{u}\wedge \zeta_{1}}\cdot[\Delta \oplus \mathcal{F}^{2}_{1}(\Delta)] \\
%    \Delta &= \hat{T}^{1,0}_{u}\oplus \bbC\zeta_{1}^{\#}
% \end{align}
% where $\hat{T}^{1,0}_{u} \subset T_{\bbC}$ is the the complex structure determined by $(\omega_{u},\Omega_{u})$ as in \eqref{eq:class 1 omega_u} on the 4-dimensional space orthogonal to $\zeta_{i}$.

Recall that the cohomology class of $F$ encodes some physical information about the background namely the number of M5 branes wrapping the circle fibered over  $M_{\text{HK}}\times \bbR$. Under infinitesimal deformations of the background, this cohomology class should remain unchanged. However we have already seen that class 0 structures have trivial flux. Hence, when $F$ is non-trivial, there should be no deformation from class 1 to class 0. Such deformations were parameterised by  $\trivector$. Thus physically we expect these moduli are obstructed by $0\neq[F] \in H^{4}(M)$.

To see this more directly from the moduli equations \eqref{eq:type 3 closure 1}-\eqref{eq:type 3 closure 3}, it is useful to use the full structure of the supergravity background. In particular, the metric naturally selects a complement to $\Delta$ given by
\begin{equation}
    \Sigma = \hat{T}^{0,1} \oplus \bbC\xi
\end{equation}
where $\xi$ is a real Killing vector that also satisfies $\mathcal{L}_{\xi}F = \mathcal{L}_{\xi}\lambda = \mathcal{L}_{\xi}\omega_{\alpha} = \mathcal{L}_{\xi}\zeta_{i} = 0$. This particular choice of complement is convenient as one can use the Killing spinor equations\footnote{This is not the case for AdS backgrounds.} to show $[\Sigma,\Sigma] \subseteq \Sigma$. We can therefore decompose the exterior derivative as
\begin{equation}
    \dd = \dd_{\Delta} + \dd_{\Sigma} \qquad \dd_{\Delta} = \Pr{}_{p+1,q}\circ \dd \qquad \dd_{\Sigma} = \Pr{}_{p,q+1}\circ\dd
\end{equation}
where $\Pr{}_{p,q}:\ext^{n}T^{*} \rightarrow \ext^{p}\Delta^{*} \otimes \ext^{q}\Sigma^{*}$. We will use the abuse of notation $\ext^{p,q}T^{*} = \ext^{p}\Delta^{*} \otimes \ext^{q}\Sigma^{*}$. Since $\bar{\Delta}\neq \Sigma$, the reader should not be confused and think of $\ext^{p,q}T^{*}$ as a decomposition under some complex structure.

We can already put constraints on the real flux $F$. Recall that the integrability of a class 1 $\SUs{6}$ structure implies that $F_{\bbC} \in \mathcal{F}^{4}_{1}(\Delta)$. As discussed in section \ref{sec:SU*(6) in local SU(2)}, in terms of the global $\SU2$ structure we have $F_\bbC=F-\ii\dd(\omega_u\wedge\zeta_1)$ and hence, using subscripts ${p,q}$ to denote the $\ext^{p,q}T^{*}$ component of the form, we must have
\begin{align}
    0 &= (F_{\bbC})_{3,1} \\
    &= F_{3,1} - \ii (\dd(\omega_{u}\wedge \zeta_{1}))_{3,1} \\
    &= F_{3,1} - \ii \dd_{\Delta}(\omega_{u}\wedge \zeta_{1})_{2,1}  \\
 \Rightarrow F_{3,1} &= \ii \hat{\del}(\omega_{u}\wedge \zeta_{1})_{2,1}
\end{align}
where we have decomposed  a tangential Dolbeault operator $\hat{\del}$ coming from the integrable hyperk\"ahler structure on $M_{\text{HK}}$. Since $F$ is a real form\footnote{Again, since we do not have a complex structure $\bar{F}_{3,1} \neq F_{1,3}$ and so we need to be careful. One can use the complex structure on $M_{\text{HK}}$ however to form constraints.}, we can put constraints on $F_{1,3}$.
 
Now suppose we have some integrable deformation $R=\trivector + r+...$ such that $\trivector$ is globally non-vanishing\footnote{This restricts us to deformations into class 0 backgrounds with a smooth metric.}. Since this is an integrable deformation, we have
\begin{equation}\label{eq:trivector obstruction 1}
    \dd_{\Delta}\trivector = 0 \qquad \dd_{\Delta}r - j\trivector \lrcorner j^{2}(F_{\bbC})_{2,2} = 0
\end{equation}
Since $\trivector$ is non-vanishing we can define $\trivector^{-1} \in \Gamma(\ext^{3}\Sigma^{*})$ to be the unique section such that $\trivector \lrcorner \trivector^{-1} = 1$. \eqref{eq:trivector obstruction 1} then implies
\begin{equation}
    (F-\ii\dd(\omega_{u}\wedge \zeta_{1}))_{2,2} = -\dd_{\Delta}(r\cdot \trivector^{-1})
\end{equation}
Again, using reality conditions, one can now completely determine the form of $F$ and we find
\begin{equation}\label{eq:trivector obstruction 2}
    F-\ii\dd(\omega_{u}\wedge \zeta_{1}) + \dd(r\cdot \trivector^{-1}) = \dd_{\Sigma}\rho
\end{equation}
some $\rho\in \Gamma(\ext^{1,2}T^{*})$ that can be described explicitly in terms of $r,\trivector$. The left hand side of \eqref{eq:trivector obstruction 2} is clearly $\dd$-closed and hence the right hand side must be $\dd_{\Delta}$-closed. Careful consideration of this equation and application of the $\hat{\del}\hat{\delb}$-lemma induced from integrable hyperk\"ahler base implies we can therefore write $\dd_{\Sigma}\rho = \dd_{\Delta} \sigma$ for some $\sigma \in \Gamma(\ext^{0,3}T^{*})$. Hence we have
\begin{equation}
    F-\ii\dd(\omega_{u}\wedge \zeta_{1}) + \dd(r\cdot\trivector^{-1}) - \dd\sigma = 0
\end{equation}
We then see that
\begin{equation}
    \Xi \neq 0 \qquad \Rightarrow \qquad [F]=0 \in H^{4}(M)
\end{equation}
or, put another way, the cohomology class of the flux acts as an obstruction to the trivector deformation as expected.

\section{Discussion}

In this paper we defined and classified a new object in $\ER{6}$ generalised geometry which we called an exceptional complex structure and used it to analyse generic supersymmetric $D=5$ Minkowski backgrounds of M-theory. These are the analogue of $\SL{3,\bbC}$ structures in conventional geometry, or $\SU{3,3}$ structures in Hitchin's generalised geometry, and they extend the definition of exceptional complex structures in \cite{Ashmore:2019rkx,Ashmore:2019qii} to $D=5$ backgrounds. 
% In this case, ECS are given by a complex subbundle $L_{1}\subset E_{\bbC}$ of dimension 6 obeying certain algebraic conditions, or alternatively, an adjoint element $\JUs$ in a particular $\ER{6}$ orbit. It defines an $\SUs{6}\times \bbR^{+}$ structure on the generalised tangent bundle, and integrability of the $G$-structure corresponds to involutivity of the complex subbundle $L_{1}$ under the Dorfman derivative. 
We saw that ECSs fell into three families labelled \emph{type} and \emph{class}. In each case, along with some extra data, the ECS defined a conventional $G$-structure on $T$; in particular for class 0 type 0 backgrounds this was simply an $\SL{3,\bbC}$ structure. Integrability of the ECS did not however, in general, lead to integrability of the conventional $G$-structure but instead constrained part of the intrinsic torsion.

The $D=5$ supersymmetric background define a $\SUs6$ structure \cite{Ashmore:2015joa} that encodes the massless hypermultiplet scalar degrees of freedom in the five-dimensional effective theory. We showed that for each such `H-structure' there is an $S^2$ of associated ECSs. 
%We saw that an ECS defines a complex line bundle $\mathcal{U}_{\JUs}\subset \ad\tilde{F}_{\bbC}$, a section of which defines a refinement of the structure to an $\SUs{6}$ structure. 
%This is the more physically relevant structure, as it defines the H-structure of M-theory backgrounds as defined in  Moreover, the moduli of this structure 
%We analysed the $\SUs{6}$-structures and found that there is always an $S^{2} \simeq \SU{2} \quotient \U{1}$ of associated ECS. 
While the type of the associated ECS may vary, we found that the class is fixed and constant on the manifold, so that $\SUs{6}$ structures were of either class 0 and class 1 matching the two types of solution identified in \cite{Gauntlett:2004zh}. Class 0 $\SUs{6}$ structures describe flux-deformed Calabi--Yau solutions in the sense of section \ref{sec: Minkowski Backgrounds}, while class 1 correspond to a an M5 brane wrapped on a circle transverse to $M_{\text{HK}}\times \bbR$, where $M_{\text{KH}}$ is hyperk\"ahler. We analysed the integrability conditions of the $\SUs{6}$ structures in multiple ways and found that class 0 solutions always have an integrable $\SL{3,\bbC}$ structure on $T$, but it is not necessarily metric-compatible. We also found that class 1 structures had a particular non-integrable $\SU{2}\times \GL{1,\bbR}\times \bbR$ structure. In each case, the integrability conditions had natural interpretations as K\"ahler quotients, or equivalently as the extremisation of a Hitchin functional. 

Using the classification of $\SUs{6}$ structures and their integrability, we were able to find the hypermultiplet moduli of arbitrary backgrounds. 
% More specifically, we saw that we could use the ECS to describe the twistor space associated to the moduli space of H-structures. This is an $S^{2}$ bundle over the physical hypermultiplet moduli space $\mathcal{M}_{\text{hyper}}$, where the $S^{2}$ corresponds to the choice of ECS associated to the $\SUs{6}$ structure. Removing this non-physical modulus corresponded to removing the modulus associated to $J_{+}$ which then fixed us on the physical moduli. We found, 
For class 0 structures, that the moduli correspond to the complexification of the moduli space of the associated $\SL{3,\bbC}$ structure. This was true independent of the flux, showing that flux deformed Calabi--Yau solutions remarkably have the same hypermultiplet moduli as the fluxless Calabi--Yau background. For class 1 solutions, we found the moduli in terms of the cohomology of some differential $\dd_{\Delta,F}$ and used this to show that a non-trivial flux obstructs some of the moduli. This can be viewed as the statement that some deformations would break the supersymmetry of the wrapped brane configuration and so are lifted. 

We saw that an integrable $\SUs6$ structure defined a rich set of cohomologies, with an $S^2$ of natural `exceptional Dolbeault operators' dependent on which compatible ECS one chooses. Provided the deformation problem for the $\SUs6$ structure  is well-defined, there must be relations between the corresponding cohomologies since we can parameterise the deformation using any of the different ECSs. In the class 0 case, this was equivalent to the Hodge Theorem relating de Rham and Dolbeault cohomology groups.  

Although we focused on the M-theory case, the general formalism is equally applicable to type II theories. The main difference would be the classification of the structures in terms of type and class and the corresponding modification of the moduli space calculations. Type IIA should follow straightforwardly from M-theory with the generic ECS being again type 0. For type IIB however it is easy to see that the generic case is type 1, and so the distribution $\Delta$ in the analogue of proposition \ref{prop:form of ECS} is never the whole of $T_\bbC$. Nonetheless one expects that the tools used to analyse type 3 ECS here would carry over to the type IIB case. 

Another very natural extension is to try to understand obstructions in this theory. In general the K\"ahler quotient only matches the moduli space if the quotient group has a free action on the space of structures. While there are generally no generalised diffeomorphisms that fix the full supergravity solutions, the analogue of the statement that a Calabi--Yau has no isometries, this may not be the case for ECS alone. For example, in Calabi--Yau backgrounds, we have $L_{V}\JUs = 0$ for $V=v$ where $v$ is a holomorphic vector field, i.e. a real vector field such that $\mathcal{L}_{v}\Omega$ is proportional to $\Omega$. While on a generic complex manifold there can be an infinite number of solutions to this equation, Calabi--Yau manifolds are unobstructed \cite{Tian88,1989CMaPh.126..325T} and hence have no holomorphic vector fields. It would be interesting to see if a similar statement holds when we move to arbitrary flux backgrounds.

One of the intriguing observations in this work is that, for generic class 0 backgrounds, one simply needs to specify an integrable $\SL{3,\bbC}$ structure. The full supersymmetric background is completed by finding an compatible V-structure. In the Calabi--Yau case, this corresponds to specifying the Ricci-flat K\"ahler metric. More generally it requires identifying a holomorphic vector $\xi$ together a exact flux $F$ satisfying \eqref{eq:Killing spinor eqn b}. This opens up a relatively straightforward way of searching for new supersymmetric flux backgrounds given any complex manifold with vanishing first Chern class. One might wonder if there was an analogue of the Calabi--Yau theorem in this case. In addition, as we have also stressed, there is the possibility that the sources in compact background only effect the V-structure, so that the $\SUs6$ structure is globally well-defined and we can directly identify the hypermultiplet moduli as $H^3(M,\bbC)$. 

Another obvious extension of this work is to try to apply it AdS backgrounds. These are described by $\USp{6}$ structures with weak generalised holonomy \cite{Coimbra:2017fqv,Ashmore:2016qvs,Ashmore:2016oug,Ashmore:2018npi}. In particular, they have an $\SUs{6}$ structure that is not quite integrable but has intrinsic torsion in a $\USp{6}$ singlet. One can always choose the ECS such that this singlet appears in the moment map and not the involutivity condition. Moreover, if we impose the condition from the Killing spinor equations that $f=0$, then this choice of ECS is globally of type 3. Unlike the Minkowski case, however, we find that the class of the structure is not constant on $M$. Indeed, this follows from the AdS Killing spinor equation $\dd(\ee^{3\warp}\sin\theta) = 2m\ee^{2\warp}\tilde{\zeta}_{1}$ where $m$ is the inverse AdS radius. This implies there can be class-changing solutions. Even though the moment map is non-zero, it has been shown by Ashmore, Petrini, Tasker and Waldram \cite{APTW} that one can still interpret the moduli space as a suitable K\"ahler quotient and so find the moduli in terms of the cohomology of $\dd_{\Delta,F}$. Grading this cohomology by the R-charge,  we should get finite dimensional results, which collectively give the whole spectrum of chiral operators on the associated CFT. It would be very interesting to investigate these ideas in the simple cases such as the Maldacena-Nunez solutions \cite{Maldacena:2000mw}.

Another direction, that we hope to address soon, is to try and apply the theory of ECS to topological theories. Recall that the 1-loop corrections to topological string theories could be calculated by quantising the Hitchin functional for generalised complex structures \cite{Pestun:2005rp}. It is natural to wonder if there is an  analogous calculation using the ECS Hitchin function  \eqref{eq:U(6) Kahler Potential}. One could then compare with one-loop corrections to the universal hypermultiplet in five dimensions as in \cite{Strominger:1997eb,Anguelova:2004sj,Antoniadis:2003sw}. 

Finally, one may use the mathematical structure analysed here as a stepping stone for understanding $D=4$ backgrounds of M-theory with non-trivial flux. We already know from \cite{Ashmore:2019qii} that such backgrounds are given in terms of an analogous ECS and we expect backgrounds with non-trivial flux to be also given by type 3 structures. The moduli will therefore be broadly similar, this time counted by a differential associated to some 4-dimensional $\Delta \subset T_{\bbC}$ which defines a transverse holomorphic foliation. These will encode the back-reacted geometry of M5 branes wrapping cycles in the internal manifold and are of interest in model-building.

\acknowledgments

DT and DW are supported in part by the EPSRC New Horizons Grant ``New geometry from string dualities'' EP/V049089/1. DW is also supported in part by the STFC Consolidated Grants ST/P000762/1 and ST/T000791/1. We would like to thank Alex Arvanitakis for helpful comments on the paper.

\appendix

\section{Conventions}\label{app:conventions}

\subsection{Exterior and Interior Products}\label{sec:ext and int product}

We use the following conventions for the exterior and interior products of differential forms and multivectors. Here $u\in \Gamma(\ext^{p}T)$, $v\in \Gamma(\ext^{q}T)$, $\lambda \in \Omega^{p}(M)$, $\rho \in \Omega^{q}(M)$, and we take $p\geq q$ without loss of generality.
\begin{align}
    (u\wedge v)^{a_{1}...a_{p+q}} & = \frac{(p+q)!}{(p!q!}v^{[a_{1}...a_{p}}u^{\_{p+1}...a_{p+q}]} \\
    (\lambda\wedge \rho)_{a_{1}...a_{p+q}} &= \frac{(p+q)!}{p!q!}\lambda_{[a_{1}...a_{p}}\rho_{a_{p+1}...a_{p+q}]} \\
    (v\lrcorner \lambda)_{a_{1}...a_{p-q}} &= \frac{1}{q!}v^{b_{1}...b_{q}}\lambda_{b_{1}...b_{q}a_{1}...a_{p-q}} \qquad \\
    (u\lrcorner \rho)^{a_{1}...a_{p-q}} &= \frac{1}{q!}u^{a_{1}...a_{p-q}b_{1}...b_{q}}\rho_{b_{1}...b_{q}} \\
    (jv\lrcorner j\rho)^{a}{}_{b} &= \frac{1}{(q-1)!}v^{ac_{1}...c_{q-1}}\rho_{bc_{1}...c_{q-1}} \\
    (j\lambda\wedge \rho)_{a,a_{1}...a_{p+q-1}} &= \frac{(p+q-1)!}{(p-1)!q!}\lambda_{a[a_{1}...a_{p-1}}\rho_{a_{p}...a_{p+q-1}]}
\end{align}
Also, for $\trivector\in \Gamma(\ext^{3}T)$, $F \in \Omega^{4}(M)$, we define
\begin{equation}
    (j\trivector\lrcorner j^{2}F)^{a}{}_{bc} = \frac{1}{2}X^{apq}F_{bcpq}
\end{equation}
we can also define for 3-forms $\alpha,\beta,\gamma,\delta$
\begin{equation}
    (j\alpha\wedge \beta)\wedge (j\gamma\wedge\delta) = \left(\frac{6!}{3!\,2!}\right)^{2}\alpha_{a_{1}b_{2}b_{3}}\beta_{b_{4}b_{5}b_{6}}\gamma_{b_{1}a_{2}a_{3}} \delta_{a_{4}a_{5}a_{6}}
\end{equation}
Indices with the same letter are antisymmetrised.

\subsection{\texorpdfstring{$\ER{6}$}{E(6(6)} Generalised Geometry}\label{sec:Ed algebra for M-theory}

The generalised tangent bundle, the adjoint bundle, and the $N$ bundle for $\ER{6}$ geometry are as follows
\begin{align}
    E &= T\oplus \ext^{2}T^{*} \oplus \ext^{5}T^{*} \\
    \ad \tilde{F} &= \bbR \oplus (T\otimes T^{*}) \oplus \ext^{3}T^{*}\oplus \ext^{6}T^{*} \oplus \ext^{3}T \oplus \ext^{6}T \\
    N &= T^{*} \oplus \ext^{4}T^{*}\oplus (T^{*}\otimes \ext^{6}T^{*})
\end{align}
We take the following sections of these bundles, where each term matches with the expressions above in the obvious way.
\begin{equation}
    V = v+\omega +\sigma+\tau \qquad R = l + r + a+\tilde{a} + \alpha + \tilde{\alpha}
\end{equation}
The following gives the adjoint action $R\cdot V = V'$
\begin{align}
    v' &= lv +r\cdot v + \alpha\lrcorner\omega - \tilde{\alpha}\lrcorner\sigma \\
    \omega' &= l\omega + r\cdot \omega + v\lrcorner a + \alpha\lrcorner\sigma \\
    \sigma' &= l\sigma + r\cdot \sigma + v\lrcorner\tilde{a} + a\wedge \omega
\end{align}
The following gives the Lie algebra bracket $[R,R']=R''$
\begin{align}
    l'' &= \tfrac{1}{3}(\alpha\lrcorner a' - \alpha'\lrcorner a) + \tfrac{2}{3}(\tilde{\alpha}'\lrcorner\tilde{a} - \tilde{\alpha}\lrcorner\tilde{a}') \\
    \begin{split}
    r'' &= [r,r'] + j\alpha\lrcorner ja' - j\alpha'\lrcorner ja - \tfrac{1}{3}\mathbb{I}(\alpha\lrcorner a' - \alpha'\lrcorner a) \\
    & \qquad +j\tilde{\alpha}'\lrcorner j\tilde{a} - j\tilde{\alpha}\lrcorner j\tilde{a}' - \tfrac{2}{3}\mathbb{I}(\tilde{\alpha}'\lrcorner\tilde{a} - \tilde{\alpha}\lrcorner\tilde{a}')
    \end{split} \\
    a'' &= r\cdot a' - r'\cdot a + \alpha'\lrcorner\tilde{a} - \alpha\lrcorner \tilde{a}' \\
    \tilde{a}'' &= r\cdot \tilde{a}' - r'\cdot \tilde{a} - a\wedge a' \\
    \alpha'' &= r\cdot \alpha' - r'\cdot \alpha + \tilde{\alpha}'\lrcorner a - \tilde{\alpha}\lrcorner a' \\
    \tilde{\alpha}'' &= r\cdot \tilde{\alpha}' - r'\cdot \tilde{\alpha} - \alpha\wedge \alpha'
\end{align}

\subsubsection{Dorfman Derivative}
The following is the Dorfman derivative on vectors.
\begin{align}
    \begin{split}
        L_{V}V' &= \mathcal{L}_{v}v' + (\mathcal{L}_{v}\omega' - v'\lrcorner\dd\omega) + (\mathcal{L}_{v}\sigma' - v'\lrcorner\dd\sigma - \omega'\wedge \dd\omega)
    \end{split}
\end{align}
The following is the Dorfman derivative on adjoint elements.
\begin{align}
    \begin{split}
        L_{V}R &= \mathcal{L}_{v} l + (\mathcal{L}_{v}r + j\alpha\lrcorner j\dd\omega -\tfrac{1}{3}\mathbb{I}\alpha\lrcorner\dd\omega -j\tilde{\alpha}\lrcorner j\dd\sigma + \tfrac{2}{3}\mathbb{I}\tilde{\alpha}\lrcorner\dd\sigma) \\
        & \qquad + (\mathcal{L}_{v}a +r\cdot \dd\omega - \alpha \lrcorner \dd\sigma) +(\mathcal{L}_{v}\tilde{a} + r\cdot \dd \sigma + \dd\omega\wedge a) \\
        & \qquad + (\mathcal{L}_{v}\alpha - \tilde{\alpha}\lrcorner\dd\omega) + \mathcal{L}_{v}\tilde{\alpha} 
    \end{split}
\end{align}
To obtain the twisted Dorfman derivative we make the following substitutions.
\begin{equation}
    \dd\omega \rightarrow \dd\omega - v\lrcorner F \qquad \dd\sigma \rightarrow \dd\sigma - v\lrcorner \tilde{F} + \omega\wedge F
\end{equation}

\subsubsection{The cubic invariant}
The following is the cubic invariant for $\E{6}$
\begin{equation}
    c(V,V,V) = -\left(v\lrcorner \omega\wedge \sigma + \frac{1}{3!}\omega\wedge\omega\wedge\omega \right)
\end{equation}

\subsubsection{The Killing Form}
The Killing form for $\E{d}$ is
\begin{equation}
    \Tr(R,R') = \frac{1}{2}\left( \frac{1}{9-d}\Tr{r}\Tr{r'} + \Tr{rr'} + \alpha\lrcorner a' + \alpha'\lrcorner a - \tilde{\alpha}\lrcorner\tilde{a}' - \tilde{\alpha}'\lrcorner \tilde{a} \right)
\end{equation}

\subsubsection{Projections}
Let $Z = \zeta + u+ s\in \Gamma(E^{*})$. Then the projection $E\times E^{*}\rightarrow \ad \tilde{F}$ is given by
\begin{align}
    l &= -\tfrac{1}{3}u\lrcorner\omega -\tfrac{2}{3}s\lrcorner \sigma \\
    r &= v\otimes \zeta -ju\lrcorner j\omega +\tfrac{1}{3}u\lrcorner\omega \mathbb{I} -js\lrcorner j \sigma + \tfrac{2}{3}s\lrcorner\sigma \mathbb{I} \\
    a &= \zeta \wedge \omega + u\lrcorner\sigma \\
    \tilde{a} &= \zeta\wedge \sigma \\
    \alpha &= v\wedge u + s\lrcorner\omega \\
    \tilde{\alpha} &= -v\wedge s
\end{align}
If $Y = \lambda + \kappa+\mu\in \Gamma(N)$ then the projection $E\times E\rightarrow N$ is given by
\begin{align}\label{eq:N projection}
    \lambda &= v\lrcorner\omega' + v'\lrcorner\omega \\
    \kappa &= v\lrcorner\sigma' + v'\lrcorner\sigma - \omega\wedge\omega' \\
    \begin{split}
    \mu &= (j\omega\wedge \sigma' + j\omega'\wedge \sigma) - \tfrac{1}{4}(\sigma\wedge \omega' + \sigma'\wedge \omega)
    \end{split}
\end{align}

\section{Spinor Bilinears and the Local \texorpdfstring{$\SU{2}$}{SU(2)} Structure}\label{app:bilinears}

We will build the spinor bilinears of the local $\SU{2}$ structure, and the associated generalised $\USp{6}$ structure, from the completely generic spinor
\begin{equation}
    \tilde{\eta} = \sqrt{2}(\cos\alpha \, \eta^{1} + \sin\alpha\, (a\eta^{1} + \sqrt{1-|a|^{2}}\eta^{2})^{*})
\end{equation}
We can parameterise $a=\ee^{\ii\psi}\sin\phi$. Then we take $\epsilon^{+},\epsilon^{-}$ as in section \ref{sec:D=5 geometry}. The non-vanishing bilinears we can form are as follows. \\
\\
\textbf{Scalars}
\begin{align}
    \bar{\epsilon}^{+} \epsilon^{+} = \bar{\epsilon}^{-}\bar{\epsilon}^{-} &= 1 \\
    \bar{\epsilon}^{+}\epsilon^{-} = \bar{\epsilon}^{-}\epsilon^{+} &= \sin\theta \\
    \epsilon^{+\,\mathrm{T}} \epsilon^{+} = -\epsilon^{-\,\mathrm{T}}\epsilon^{-} &= \ee^{-\ii\psi}\sin\phi \, \cos\theta
\end{align}

\noindent \textbf{1-forms}
\begin{align}
    \bar{\epsilon}^{+} \gamma_{(1)} \epsilon^{+} = - \bar{\epsilon}^{-} \gamma_{(1)} \epsilon^{-} &= \cos\theta\,\cos\phi\, \zeta_{1} \\
    \ii\bar{\epsilon}^{+}\gamma_{(1)}\epsilon^{-} = -\ii\bar{\epsilon}^{-}\gamma_{(1)} \epsilon^{+} &= \cos\theta \, \cos\phi\, \zeta_{2}
\end{align}

\noindent \textbf{2-forms}
\begin{align}
    \begin{split}
        -\ii \bar{\epsilon}^{+}\gamma_{(2)} \epsilon^{+}  = -\ii\bar{\epsilon}^{-}\gamma_{(2)} \epsilon^{-} &= \omega_{3}(1-\sin^{2}\phi \, (1-\sin\theta)) - \sin\theta \zeta_{1} \wedge \zeta_{2} \\
        & \qquad - \frac{1}{2}(1-\sin\theta) \sin 2\phi \, (\cos\psi \, \omega_{1} + \sin\psi\, \omega_{2})
    \end{split} \\
    \begin{split}
        \ii \bar{\epsilon}^{+}\gamma_{(2)} \epsilon^{-} = \ii \bar{\epsilon}^{-}\gamma_{(2)} \epsilon^{+} &= \zeta_{1}\wedge \zeta_{2} - \omega_{3} (\sin\theta + \sin^{2}\phi\, (1-\sin\theta)) \\
        & \qquad + \frac{1}{2}(1-\sin\theta) \sin 2\phi \, (\cos\psi \, \omega_{1} + \sin \psi \, \omega_{2} )
    \end{split} \\
    \epsilon^{+\,\mathrm{T}} \gamma_{(2)} \epsilon^{-} = - \epsilon^{-\,\mathrm{T}} \gamma_{(2)} \epsilon^{+} &= \cos\theta \,( \cos \phi \, (\omega_{2} + \ii \omega_{1}) + \ii \ee^{-\ii\psi} \sin \phi \, (\omega_{3} - \zeta_{1} \wedge \zeta_{2}))
\end{align}

\noindent \textbf{3-forms}
\begin{align}
    \begin{split}
        \ii\bar{\epsilon}^{+} \gamma_{(3)} \epsilon^{+} = -\ii\bar{\epsilon}^{-} \gamma_{(3)} \epsilon^{-} &= \cos\theta \, \sin\phi \, \cos\psi\, (\omega_{1}\wedge \zeta_{1} - \omega_{2}\wedge \zeta_{2})  \\
        & \qquad +\cos\theta\, \sin \phi\, \sin\psi \, (\omega_{2} \wedge \zeta_{1} + \omega_{1} \wedge \zeta_{2}) \\
        & \qquad -\cos\theta\, \cos \phi\, \omega_{3}\wedge \zeta_{1}
    \end{split} \\
    \begin{split}
        \bar{\epsilon}^{+}\gamma_{(3)} \epsilon^{-} = -\bar{\epsilon}^{-} \gamma_{(3)} \epsilon^{+} &= \cos\theta \, \sin\phi\, \sin\psi\, (\omega_{1}\wedge \zeta_{1} - \omega_{2}\wedge \zeta_{2}) \\
        & \qquad - \cos\theta \, \sin\phi\, \cos\psi\, (\omega_{2}\wedge \zeta_{1} + \omega_{1} \wedge \zeta_{2}) \\
        & \qquad + \cos\theta\, \cos\phi\, \omega_{3} \wedge \zeta_{2}
    \end{split} \\
    \begin{split}
        \epsilon^{+\,\mathrm{T}}\gamma_{(3)} \epsilon^{+} = \epsilon^{-\,\mathrm{T}}\gamma_{(3)}\epsilon^{-} &= -(\omega_{2} + \ii\omega_{1})\wedge(\sin\theta\,\zeta_{1} - \ii \zeta_{2})  \\
        & \qquad - \frac{1}{2}(1-\sin\theta)\sin^{2}\phi \, (\omega_{2} + \ii\omega_{1}) \wedge (\zeta_{1} + \ii \zeta_{2}) \\
        & \qquad -\frac{1}{2}(1-\sin\theta )\sin^{2}\phi \, \ee^{-2\ii\psi}(\omega_{2} - \ii \omega_{1}) \wedge (\zeta_{1} + \ii \zeta_{2}) \\
        & \qquad + \frac{\ii}{2}(1-\sin\theta) \ee^{-\ii\psi}\sin2\phi\, \omega_{3}\wedge (\zeta_{1} + \ii \zeta_{2})
    \end{split} \\
    \begin{split}
        \epsilon^{+\,\mathrm{T}}\gamma_{(3)} \epsilon^{-} = \epsilon^{-\,\mathrm{T}}\gamma_{(3)} \epsilon^{+} &= -(\omega_{2} + \ii\omega_{1}) \wedge (\zeta_{1} - \ii \sin\theta\, \zeta_{2}) \\
        & \qquad + \frac{1}{2}(1-\sin\theta) \sin^{2}\phi \, (\omega_{2} + \ii \omega_{1}) \wedge (\zeta_{1} + \ii \zeta_{2}) \\
        & \qquad + \frac{1}{2}(1-\sin\theta ) \sin^{2}\phi \, \ee^{-2\ii\psi} (\omega_{2}-\ii\omega_{1} ) \wedge (\zeta_{1} + \ii\zeta_{2}) \\
        & \qquad - \frac{\ii}{2}(1-\sin\theta) \sin 2\phi \, \ee^{-\ii\psi} \omega_{3}\wedge (\zeta_{1} + \ii\zeta_{2})
    \end{split}
\end{align}

\noindent \textbf{4-forms}
\begin{align}
    \begin{split}
        \bar{\epsilon}^{+}\gamma_{(4)} \epsilon^{+} = \bar{\epsilon}^{-} \gamma_{(4)} \epsilon^{-} &= *(\ii \bar{\epsilon}^{+}\gamma_{(2)} \epsilon^{-}) = * (\ii \bar{\epsilon}^{-} \gamma_{(2)} \epsilon^{+} )
    \end{split} \\
    \begin{split}
        - \bar{\epsilon}^{+} \gamma_{(4)} \epsilon^{-} = -\bar{\epsilon}^{-} \gamma_{(4)} \epsilon^{+} &= *(-\ii\bar{\epsilon}^{+} \gamma_{(2)} \epsilon^{+}) = *(-\ii\bar{\epsilon}^{-} \gamma_{(2)} \epsilon^{-})
    \end{split} \\
    \begin{split}
        -\ii\epsilon^{+\,\mathrm{T}} \gamma_{(4)} \epsilon^{+} = \ii \epsilon^{-\,\mathrm{T}}\gamma_{(4)} \epsilon^{-} &= *(\epsilon^{+\,\mathrm{T}}\gamma_{(2)} \epsilon^{-} ) = *(-\epsilon^{-\,\mathrm{T}} \gamma_{(2)} \epsilon^{+})
    \end{split}
\end{align}

\noindent \textbf{5-forms}
\begin{align}
\begin{split}
    \bar{\epsilon}^{+}\gamma_{(5)} \epsilon^{+} = -\bar{\epsilon}^{-}\gamma_{(5)} \epsilon^{-} &= *(\ii\bar{\epsilon}^{+}\gamma_{(1)}\epsilon^{-}) = *(-\ii\bar{\epsilon}^{-} \gamma_{(1)} \epsilon^{+}) \\
    &= \cos\theta\,\cos\phi\, *\zeta_{2}
\end{split} \\
\begin{split}
    -\ii \bar{\epsilon}^{+}\gamma_{(5)} \epsilon^{-} = \ii \bar{\epsilon}^{-} \gamma_{(5)} \epsilon^{+} &= *(\bar{\epsilon}^{+}\gamma_{(1)} \epsilon^{+} ) = *(-\bar{\epsilon}^{-}\gamma_{(1)} \epsilon^{-}) \\
    &= \cos\theta \, \cos\phi\,*\zeta_{1}
\end{split}
\end{align}

\noindent \textbf{6-forms}
\begin{align}
    -\ii\bar{\epsilon}^{+}\gamma_{(6)}\epsilon^{+} = -\ii \bar{\epsilon}^{-}\gamma_{(6)} \epsilon^{-} &= \sin\theta\, \vol \\
    -\ii\bar{\epsilon}^{+} \gamma_{(6)} \epsilon^{-} = -\ii \bar{\epsilon}^{-} \gamma_{(6)} \epsilon^{+} &= \vol \\
    -\ii\epsilon^{+\,\mathrm{T}} \gamma_{(6)} \epsilon^{-} = \ii \epsilon^{-\,\mathrm{T}}\gamma_{(6)}\epsilon^{+} &= \ee^{-\ii\psi}\sin\phi\, \cos\theta\, \vol
\end{align}

\section{ECS in \texorpdfstring{$\ER{6}$}{RxE6(6)} Geometry}\label{app:ECS in E6}

\begin{propn}
Any isotropic subbundle $L\subset E_{\bbC}$ has the form
\begin{equation}
    \ee^{\alpha+\beta}\cdot(\Delta \oplus S_{2} \oplus S_{5}) \label{eq: general iso bundle}
\end{equation}
where $\alpha\in\Omega^{3}(M)$ and$\beta\in \Omega^{6}(M)$ are arbitrary but fixed, and where $\Delta\subset T$, $S_{2}\subset \ext^{2}T^{*}$, $S_{5}\subset \ext^{5}T^{*}$ satisfy the following. For all $v\in \Delta$, $\omega,\omega'\in S_{2}$and $\sigma \in S_{5}$ we have
\begin{equation}
    \arraycolsep = 1.4pt
    \begin{array}{rclcrcl}
    v\lrcorner \omega &=& 0 & \qquad & v\lrcorner\sigma &=& 0 \\
    \omega\wedge \omega' &=& 0 & & j\omega\wedge \sigma &=& 0
    \end{array}
\end{equation}
\end{propn}

\noindent To prove this, we follow a similar proof for isotropic bundles in $\mathrm{O}(d,d)$ geometry laid out in \cite{Gualtieri04}.

\begin{proof}

The condition for isotropy is $V_{1}\times_{N} V_{2} = 0$ for all $V_{1},V_{2}\in L$ which translates to
\begin{align}
    v_{1}\lrcorner\omega_{1} + v_{2}\lrcorner\omega_{1} &= 0 \label{eq:iso condtition 1} \\
    j\omega_{1}\wedge \sigma_{2} + j\omega_{2}\wedge \sigma_{1} &= 0 \label{eq:iso condition 2} \\
    \omega_{1}\wedge \omega_{2} - v_{1}\lrcorner\sigma_{2} - v_{2}\lrcorner\sigma_{1} &= 0 \label{eq:iso condition 3}
\end{align}
It is a simple check to see that any $L$ of the form \eqref{eq: general iso bundle} satisfies these conditions and hence defines an isotropic bundle. Hence it is left to show that any isotropic bundle takes that form.

Clearly we have $\Delta = \anchor(L)$. Suppose we have some
\begin{equation}
    \omega_{1},\omega_{2} \in \pi_{\wedge^{2}T^{*}} \left( (\ext^{2}T^{*}\oplus \ext^{5}T^{*}) \cap L\right)
\end{equation}
where $\pi_{\wedge^{2}T^{*}}:E\rightarrow \ext^{2}T^{*}$, and similarly for $\pi_{T}$, $\pi_{\wedge^{5}T^{*}}$. From \eqref{eq:iso condtition 1} and \eqref{eq:iso condition 3} we see that for any $v\in \Delta$
\begin{equation}
\arraycolsep = 1.4pt
    \begin{array}{rclcrcl}
    v\lrcorner\omega_{i} &=& 0 & \quad \Rightarrow \quad & \omega_{i}&\in& \mathcal{F}^{2}_{1}(\Delta) \\
    \omega_{1}\wedge \omega_{2} &=& 0 & \quad \Rightarrow \quad & \omega_{i} &\in& S_{2}
    \end{array}
\end{equation}
Now consider the element
\begin{equation}
    \alpha(v) := \pi_{\wedge^{2}T^{*}}\left( \pi_{T}^{-1}(v)\cap L \right) \in \frac{\ext^{2}T^{*}}{\pi_{\wedge^{2}T^{*}}( (\ext^{2}T^{*}\oplus \ext^{5}T^{*})\cap L )}
\end{equation}
From \eqref{eq:iso condtition 1} for $V\times_{N} V = 0$, we see that we need
\begin{equation}
    v\lrcorner \alpha(v) \quad \forall \, v\in \Delta \quad \Rightarrow \quad \alpha\in \ext^{3}T^{*} \text{ (WLOG)}
\end{equation}
Then we can write any element $\lambda \in \pi_{\wedge^{2}T^{*}}(L)$ as
\begin{equation}
    \lambda = v\lrcorner\alpha + \omega \qquad v\in \Delta, \, \omega \in S_{2}
\end{equation}

No we consider any $\sigma \in \ext^{5}T^{*}\cap L$. From \eqref{eq:iso condition 2} and \eqref{eq:iso condition 3} we see that for all $v\in \Delta$, $\omega \in S_{2}$ we need
\begin{equation}
    \arraycolsep = 1.4pt
    \begin{array}{rclcrcl}
       v\lrcorner\sigma &=& 0 & \quad \Rightarrow \quad & \sigma &\in& \mathcal{F}^{5}_{4}(\Delta) \\
       j\omega\wedge \sigma &=& 0 & \quad \Rightarrow \quad & \sigma &\in& S_{5}
    \end{array}
\end{equation}
Note that we also need
\begin{equation}
    j(v\lrcorner\alpha)\wedge \sigma = 0 \quad \Leftrightarrow \quad (\vol^{\#}\lrcorner \sigma)\lrcorner(v\lrcorner\alpha) = 0
\end{equation}
However, since $\mathcal{F}^{5}_{4}(\Delta) = 0$ if $\dim\Delta > 1$, one can check that
\begin{equation}
    (\vol^{\#}\lrcorner \sigma)\lrcorner(v\lrcorner\alpha) \propto v\lrcorner(v\lrcorner\alpha) = 0
\end{equation}
Now consider the element
\begin{equation}
    \theta(v,\omega) := \pi_{\wedge^{5}T^{*}}\left((\pi^{-1}_{T}(v) + \pi_{\wedge^{2}T^{*}}^{-1}(\omega))\cap L\right) \in \frac{\ext^{5}T^{*}}{\ext^{5}T^{*}\cap L}
\end{equation}
From \eqref{eq:iso condition 3} we need
\begin{equation}
    (\omega_{1}+ v_{1}\lrcorner\alpha) \wedge (\omega_{2}+v_{2}\lrcorner\alpha)  - v_{1}\lrcorner\theta(v_{2},\omega_{2}) - v_{2}\lrcorner\theta(v_{1},\omega_{1})
\end{equation}
which has the general solution
\begin{equation}
    \theta(v,\omega) = \frac{1}{2}v\lrcorner\alpha\wedge \alpha + v\lrcorner\beta + \lambda\wedge \alpha
\end{equation}
where $\beta \in \ext^{6}T^{*}$ is arbitrary. It is a simple check to see that this also satisfies \eqref{eq:iso condition 2}. Checking the action of $\ee^{\alpha+\beta}$ we see that we have
\begin{equation}
    L = \ee^{\alpha+\beta}\cdot(\Delta \oplus S_{2}\oplus S_{5})
\end{equation}
\end{proof}

\begin{propn}
\begin{equation}
    \dim_{\bbC}L = 6 \qquad \Leftrightarrow \qquad \type L = 0,3,6
\end{equation}
\end{propn}

\begin{proof}
We will consider each $\type L = k$ for $k=0,1,...,6$ \\
\\
\textbf{k=0}\\
All type 0 bundles are of the form $L = \ee^{\alpha+\beta}\cdot T$ which is clearly 6 dimensional \\
\\
\textbf{k=1}\\
If we have a type 1 bundle then $\dim\Delta = 5$ and $\dim \mathcal{F}^{2}_{1}(\Delta) = \dim \mathcal{F}^{5}_{4}(\Delta) = 0$. Hence the bundle looks like $\ee^{\alpha+\beta}\cdot \Delta$ again. However, this is just 5 dimensional. \\
\\
\textbf{k=2} \\
In this case we have $rk\Delta = 4$, $\dim\mathcal{F}^{2}_{1}(\Delta) = 1$, $\dim \mathcal{F}^{5}_{4}(\Delta) = 0$ and hence the isotropic bundle is of the form $\ee^{\alpha+\beta}\cdot(\Delta \oplus \mathcal{F}^{2}_{1}(\Delta))$ which is 5 dimensional. \\
\\
\textbf{k=3} \\
We have $\dim\Delta = 3$, $\dim\mathcal{F}^{2}_{1} = 3$, $\dim\mathcal{F}^{5}_{4} = 0$. Hence we can take $L = \ee^{\alpha+\beta}\cot(\Delta\oplus \mathcal{F}^{2}_{1}(\Delta))$ which is 6 dimensional. \\
\\
\textbf{k = 4}\\
We have $\dim \Delta = 2$, $\dim\mathcal{F}^{2}_{1}(\Delta) = 6$, $\dim\mathcal{F}^{5}_{4}(\Delta) = 0$. However, it is not that case that $\omega\wedge \omega' = 0 $ for all $\omega,\omega'\in \mathcal{F}^{2}_{1}(\Delta)$. We take any subspace which satisfies this condition which has maximal dimension 3. Hence the isotropic bundle of the form $\ee^{\alpha+\beta}\cdot(\Delta\oplus S_{2})$ has maximal dimension 5\\
\\
\textbf{k = 5}\\
We have $\dim\Delta = 1$, $\dim \mathcal{F}^{2}_{1} = 10$, $\dim\mathcal{F}^{5}_{4}(\Delta) = 1$. Again, we choose a maximal $S_{2}\subset \mathcal{F}^{2}_{1}(\Delta)$ satisfying $\omega\wedge \omega' = 0$. This will have dimension 3 and so the isotropic bundle $\ee^{\alpha+\beta}\cdot(\Delta\oplus S_{2}\oplus S_{5})$ has dimension 5. \\
\\
\textbf{k = 6}
In this case $\dim\Delta = 0$. It will be convenient to parameterise $S_{5} = \Gamma\lrcorner \vol$ where $\Gamma\subset T$. We will also choose a basis $e^{i}$ of $T^{*}$ with dual basis $\hat{e}_{i}$ of $T$. The only possible type 6 solutions are given in the table below.
\begin{table}[h]
    \centering
    \begin{tabular}{cccc}
      $\Gamma$ & $S_{2}$ & $L$ & $\dim L$ \\
      \hline
      $T$ & 0 & $\ext^{5}T^{*}$ & 6 \\
      $\left<\hat{e}_{1},...,\hat{e}_{4}\right>$ & $\left<e^{5}\wedge e^{6}\right>$ & $S_{2}\oplus S_{5}$ & 5 \\
      $\left<\hat{e}_{1},\hat{e}_{2}, \hat{e}_{3}\right>$ & $\left<e^{i}\wedge e^{j}\,|\,i,j=4,5,6 \right>$ & $s_{2}\oplus S_{5}$ & 6
    \end{tabular}
\end{table}

\end{proof}

\begin{propn}
There are no ECS of type 6
\end{propn}

\begin{proof}
There are two different 6 dimensional isotropic spaces of type 6, which shown in the table above. We will show that these do not satisfy the remaining conditions in definition \ref{def:E6 ECS}.

Firstly, let's consider $L_{1}= \ext^{5}T^{*}$. Clearly, this does not satisfy condition (iii) as $\bar{L}_{1} = L_{1}$. Therefore, this cannot be an ECS.

Secondly, let's consider $L_{1} = \ee^{\alpha}\cdot (S_{2}\oplus S_{5})$, where $S_{2},S_{5}$ are as in the third row of the table above. We will show that $L_{0}\cap(L_{1}\oplus L_{-1}) \neq 0$ and hence this does not satisfy condition (iii). To find $L_{0}$ we need to find the null space
\begin{equation}
    A = \left\{ Z \in E^{*} \,|\, \left<V,Z\right> = 0 \; \forall\,V\in L_{1}\oplus L_{-1} \right\}
\end{equation}
Using the same notation as above, it is easy to see that
\begin{equation}
    T^{*}\oplus \left< \hat{e}^{i}\wedge \hat{e}^{j}\,|\,i = 1,2,3,\, j = 4,5,6\right> \subseteq A
\end{equation}
The left hand side of this is 15 dimensional. If $L$ is to define an ECS then $A$ must be 15 dimensional too and hence this must be the whole of $A$. In particular, this implies that
\begin{equation}
    S_{5}\oplus \bar{S}_{5} = \ext^{5}T^{*} \subset L_{1}\oplus L_{-1}
\end{equation}
Now taking any $\nu\in T^{*}$, and some $\ee^{\alpha}\cdot \omega \in L_{1}$. Then we have
\begin{equation}
    \ee^{\alpha}\omega \times_{\ad} \nu = \ee^{\alpha}\omega\wedge \nu
\end{equation}
However, we have that
\begin{equation}
    (\ee^{\alpha}\omega\wedge \nu)\cdot L_{-1} \subseteq \ext^{5}T^{*} \subset L_{1}\oplus L_{-1}
\end{equation}
Hence we see that $L_{0}\cap(L_{1}\oplus L_{1}) \neq 0$ and so this cannot be an ECS
\end{proof}

\begin{propn}
Suppose $L_{1} = \ee^{\alpha+\beta}(\Delta \oplus \mathcal{F}^{2}_{1}(\Delta))$ is an almost ECS of type 3. Then $\dim(\Delta \cap \bar{\Delta}) \leq 1$ 
\end{propn}

\begin{proof}
If $\dim{\Delta} = 3$, then we can choose a local non-vanishing $\xi \in \Gamma(\mathcal{F}^{3}_{2})$. That is, $v\lrcorner\xi = 0$ for all $v\in \Gamma(\Delta)$. Then we can take a local non-vanishing section $\Jpl \in \Gamma(\mathcal{U}_{J})$ as
\begin{equation}
    \Jpl = \ee^{\alpha+\beta}\cdot \xi
\end{equation}
Condition (iv) of definition \ref{def:E6 ECS} then reads
\begin{equation}
    \Tr((\ee^{\alpha + \beta}\cdot \xi)(\ee^{\bar{\alpha}+\bar{\beta}}\cdot \bar{\xi})) < 0
\end{equation}
Careful evaluation of this finds
\begin{align}
\begin{split}
    0 &> (\xi\wedge\bar{\xi})\otimes (\beta - \bar{\beta}) - (\alpha\wedge\xi)\otimes(\bar{\alpha}\wedge\bar{\xi}) + \frac{1}{2}\bigl( (\alpha\wedge\xi)\otimes(\alpha\wedge\bar{\xi}) + (\bar{\alpha}\wedge\xi)\otimes(\bar{\alpha}\wedge\bar{\xi}) \bigr) \\
    & \quad + (j\alpha\wedge\xi)\wedge(j\bar{\alpha}\wedge\bar{\xi}) - \frac{1}{2}\bigl( (j\alpha\wedge\xi)\wedge(j\alpha\wedge\bar{\xi}) + (j\bar{\alpha}\wedge\xi)\wedge(j\bar{\alpha}\wedge\bar{\xi}) \bigr) \\
\end{split}\label{eq:type 3 non-linear constraint}
\end{align}
We have also used the notation
\begin{equation}
    (j\alpha\wedge \xi)\wedge (j\bar{\alpha}\wedge\bar{\xi}) = \left(\frac{6!}{3!\,2!}\right)^{2}\alpha_{a_{1}b_{2}b_{3}}\xi_{b_{4}b_{5}b_{6}}\bar{\alpha}_{b_{1}a_{2}a_{3}} \bar{\xi}_{a_{4}a_{5}a_{6}}
\end{equation}
Indices with the same letter are antisymmetrised. The other expressions containing $j$'s are defined similarly. One can check that the right hand side of \eqref{eq:type 3 non-linear constraint} is a section of $(\det T^{*}_{\bbR})^{2}$, and depends on the choice of $\xi$ only by multiplication by a positive real scalar. Moreover, it is invariant under the change
\begin{equation}
    \alpha+\beta \longrightarrow (\alpha+\gamma) + (\beta - \frac{1}{2}\alpha\wedge\gamma) \qquad \gamma\in \mathcal{F}^{3}_{1}(\Delta)
\end{equation}
(i.e. a shift of $\alpha,\beta$ that leaves $L_{1}$ invariant). Hence the constraint is well defined. 

Now suppose we have $\dim(\Delta\cap\bar{\Delta}) \geq 2$. Then $\xi\wedge\bar{\xi} = 0$ and so the $\beta$ term drops out. Moreover, one can show that
\begin{align}
    (j\alpha\wedge\xi)\wedge(j\alpha\wedge\bar{\xi}) &= 3(\alpha\wedge\xi)\otimes(\alpha\wedge\bar{\xi}) \\
    (j\bar{\alpha}\wedge\xi)\wedge(j\bar{\alpha}\wedge\bar{\xi}) &= 3(\bar{\alpha}\wedge\xi)\otimes(\bar{\alpha}\wedge\bar{\xi}) \\
    (j\alpha\wedge \xi)\wedge (j\bar{\alpha}\wedge\bar{\xi}) &= (\bar{\alpha}\wedge\xi)\otimes(\alpha\wedge\bar{\xi}) + 2(\alpha\wedge\xi)\otimes(\bar{\alpha}\wedge\bar{\xi})
\end{align}
Hence, the right hand side of \eqref{eq:type 3 non-linear constraint} becomes
\begin{equation}
    (\alpha\wedge\xi)\otimes (\bar{\alpha}\wedge\bar{\xi}) + (\bar{\alpha}\wedge\xi)\otimes (\alpha\wedge \bar{\xi}) - (\alpha\wedge\xi)\otimes (\alpha\wedge\bar{\xi}) - (\bar{\alpha}\wedge\xi)\otimes (\bar{\alpha}\wedge\bar{\xi})
\end{equation}
writing $\alpha\wedge\xi = a + \ii b$, $\alpha\wedge\bar{\xi}= c + \ii d$ for some $a,b,c,d \in \Omega^{6}(M)_{\bbR}$, we find that this equals
\begin{equation}
    (a-c)^{2} + (b+d)^{2} \geq 0
\end{equation}
Hence \eqref{eq:type 3 non-linear constraint} does not hold for $\dim(\Delta\cap\bar{\Delta})\geq 2$.
\end{proof}

\section{Proof of the Local Structure of Moduli of \texorpdfstring{$\SUs{6}$}{SU*(6)} Structures}\label{app:local SU*(6) structure}

\subsection{Type 3 Moduli}

For the type 3 problem we have the subbundle
\begin{equation}
    L_{1} = \ee^{\alpha+\beta}\cdot(\Delta \oplus \mathcal{F}^{2}_{1}(\Delta))
\end{equation}
where $\alpha\in \Omega^{3}(M)_{\bbC}$, $\beta\in \Omega^{6}(M)_{\bbC}$, and $\Delta \subset T$ all satisfy
\begin{equation}
    [\Delta,\Delta]\subseteq \Delta \qquad v\lrcorner(w\lrcorner(x\lrcorner \dd\alpha)) = 0
\end{equation}
for all $v,w,x\in \Gamma(\Delta)$. In what follows, it will be convenient to work with the $F_{\bbC}$-twisted Dorfman derivative, where locally $F_{\bbC} = \dd\alpha$, and the untwisted bundle $\tilde{L}_{1} = \Delta\oplus \mathcal{F}^{2}_{1}(\Delta)$. This is because, for type 3 solutions, the physical flux may be in a non-trivial cohomology class and hence the gauge potential $A$, which is implicit in the definition of $\alpha$, may not be global. By working with $L_{V}^{F_{\bbC}}$, we can work only with globally defined objects.

We have the quotient spaces
\begin{align}
    E_{\bbC}\quotient \tilde{L}_{1} &= \left(T\quotient \mathfrak{F}^{1}_{0} \right) \oplus \left( \ext^{2}T^{*}\quotient \mathcal{F}^{2}_{1} \right) \oplus \ext^{5}T^{*} \\
    \mathfrak{Q}_{\bbR^{+}\times \Us{6}} &= \left[\left(T\quotient \mathfrak{F}^{1}_{0}\right) \otimes \left(T^{*}\quotient \mathcal{F}^{1}_{0}\right)\right] \oplus \left(\ext^{3}T\quotient \mathfrak{F}^{3}_{2} \right) \oplus \left(\ext^{3}T^{*} \quotient \mathcal{F}^{3}_{1}\right) \oplus \ext^{6}T^{*}
\end{align}
We shall pick the following elements of the deformation space $\mathfrak{Q}$
\begin{equation}
    \arraycolsep = 1.4pt
    \begin{array}{rclcrcl}
       r & \in & \Gamma\left( (T\quotient\mathfrak{F}^{1}_{0})\otimes (T^{*}\quotient \mathcal{F}^{1}_{0}) \right) & \qquad & \trivector &\in & \ext^{3}T\quotient \mathfrak{F}^{3}_{2} \\
       \theta &\in& \ext^{3}T^{*}\quotient \mathcal{F}^{3}_{1} & & \tau &\in & \ext^{6}T^{*}
    \end{array}
\end{equation}
and write the deformation parameter as $R = \trivector+r+\theta+\tau$. Hence the deformed bundle becomes
\begin{equation}
    \tilde{L}'_{1} = (1+R)\cdot \tilde{L}_{1} = \ee^{\theta+\tau}(1+r+\trivector)\cdot (\Delta \oplus \mathcal{F}^{2}_{1}(\Delta))
\end{equation}
where we are working to linear order in the deformation parameters only. We take sections $V',W'\in \Gamma(\tilde{L}'_{1})$ which are of the form
\begin{equation}
    V' = \ee^{\theta+\tau}\cdot (v + r\cdot v + \trivector\lrcorner \lambda + \lambda + r\cdot \lambda) \qquad W' = \ee^{\theta+\tau}\cdot(w+r\cdot w + x\lrcorner \mu + \mu +r\cdot \mu)
\end{equation}
where $v,w\in \Gamma(\Delta)$, $\lambda,\mu\in \Gamma(\mathcal{F}^{2}_{1})$. We will also denote by $\hat{V} = \hat{v}+\hat{\lambda} = (1+r+\trivector)\cdot (v+\lambda)$, and similarly for $\hat{W}$.

We want to determine when $\tilde{L}'_{1}$ is involutive under $L_{V}^{F_{\bbC}}$, to linear order in $R$. This is the statement that for all $v,w\in \Gamma(\Delta)$, $\lambda,\mu\in \Gamma(\mathcal{F}^{2}_{1})$ we have
\begin{align}
    \begin{split}
        L_{V'}^{F_{\bbC}}W' &= \ee^{\theta+\tau}\cdot\bigl(L_{\hat{V}}\hat{W} + \hat{w}(\lrcorner\hat{v}\lrcorner(F_{\bbC}+\dd\theta)) +\hat{\mu}\wedge(\hat{v}\lrcorner(F_{\bbC}+\dd\theta))\bigr)
    \end{split} \\
    \begin{split}
        &= \ee^{\theta+\tau}\cdot \Bigl( [\hat{v},\hat{w}] \\
        & \qquad \qquad +\mathcal{L}_{\hat{v}}\hat{\mu} - \hat{w}\lrcorner\dd\hat{\lambda} + \hat{w}\lrcorner(\hat{v}\lrcorner(F_{\bbC}+\dd\theta)) \\
        & \qquad \qquad - \hat{\mu}\wedge \dd\hat{\lambda} +\hat{\mu}\wedge (\hat{v}\lrcorner(F_{\bbC}+\dd\theta)) \Bigr)
    \end{split} \\
    & \in \Gamma(\tilde{L}'_{1})
\end{align}
Let us consider this term by term. We will use Greek letters $\alpha,\beta,\gamma,...$ for $\Delta$ indices, and Latin letters $a,b,c,...$ for the complement\footnote{Here we are implicitly using the orthogonal complement under some metric. This is just for ease of the proof although it is not strictly needed to prove these results.}. If we consider the vector piece only then we have, to linear order in $R$
\begin{align}
    \bigl(\ee^{-\theta-\tau}\cdot L^{F_{\bbC}}_{V'}W'\bigr)|_{T} &= [v,w] + [r\cdot v + \trivector\lrcorner \lambda, w] + [v,r\cdot w+\trivector\lrcorner \mu] \\
    \begin{split}
        &= [v,w] +\left( r^{b}{}_{\beta}v^{\beta}\del_{b}w^{\alpha} - r^{b}{}_{\gamma}w^{\gamma}\del_{b}v^{\alpha} \right) + \left( (\trivector\lrcorner \nu)^{b}\del_{b} w^{\alpha} - (\trivector\lrcorner \lambda)^{b}\del_{b}v^{\alpha} \right)\\
        & \qquad + r\cdot [v,w] + \trivector\lrcorner(v\lrcorner\dd\mu - w\lrcorner \dd\lambda) \\
        & \qquad w\lrcorner(v\lrcorner \dd_{\Delta} r) + (v\lrcorner \dd_{\Delta}\trivector)\lrcorner \mu - (w\lrcorner\dd_{\Delta}\trivector)\lrcorner \lambda
    \end{split} \\
    & \overset{!}{=} z + r\cdot z + \trivector\lrcorner \zeta
\end{align}
where $z\in \Gamma(\Delta)$, $\zeta\in \Gamma(\mathcal{F}^{2}_{1})$ are of the form\footnote{The 0\textsuperscript{th} order piece of $z,\zeta$ should be given by the Dorfman derivative of the undeformed sections $V = v+\lambda$, $W=w+\mu$.}
\begin{align}
    z &= [v,w] + O(R) \qquad \zeta = v\lrcorner\dd\mu - w\lrcorner\dd\lambda + w\lrcorner(v\lrcorner F_{\bbC}) + O(R)
\end{align}
For this to be true to linear order in $R$ we need
\begin{align}
\begin{split}
    w\lrcorner(v\lrcorner\dd_{\Delta}r) + (v\lrcorner\dd_{\Delta} \trivector)\lrcorner \mu - (w\lrcorner\dd_{\Delta}\trivector)\lrcorner \lambda &= \trivector\lrcorner(w\lrcorner(v\lrcorner(F_{\bbC}))) \\
    &= w\lrcorner(v\lrcorner(j \trivector\lrcorner j^{2}F_{\bbC}))
\end{split}
\end{align}
This must be true for all $v,w,\lambda,\mu$ and hence we have
\begin{equation}
    \dd_{\Delta}\trivector = 0 \qquad \dd_{\Delta}r - j\trivector\lrcorner j^{2}F_{\bbC} = 0
\end{equation}

Now let's consider the 2-form piece. We have
\begin{align}
    \begin{split}
        \bigl(\ee^{-\theta-\tau}\cdot L_{V'}^{F_{\bbC}}W'\bigr)|_{\wedge^{2}T^{*}} &= v\lrcorner\dd\mu - w\lrcorner\dd\lambda + w\lrcorner(v\lrcorner F_{\bbC}) \\
        & \qquad + (r\cdot v + \trivector\lrcorner\lambda)\lrcorner\dd\mu + v\lrcorner\dd(r\cdot \mu) \\
        & \qquad + \dd((r\cdot v + \trivector\lrcorner \lambda)\lrcorner \mu) + \dd(v\lrcorner(r\cdot \mu))  \\
        & \qquad - (r\cdot w + \trivector\lrcorner \mu)\lrcorner \dd\lambda - w\lrcorner\dd(r\cdot \lambda) \\
        & \qquad + (r\cdot w + \trivector\lrcorner \mu) \lrcorner (v\lrcorner F_{\bbC}) + w\lrcorner((r\cdot v +\trivector\lrcorner \lambda)\lrcorner F_{\bbC}) \\
        & \qquad +w\lrcorner(v\lrcorner \dd\theta)
    \end{split} \\
    & \overset{!}{=} \zeta + r\cdot \zeta
\end{align}
For now, let us set $\lambda, \mu = 0$. We are left with
\begin{align}
        & w\lrcorner(v\lrcorner F_{\bbC}) (r\cdot w)\lrcorner(v\lrcorner F_{\bbC}) + w\lrcorner((r\cdot v)\lrcorner F_{\bbC}) + w\lrcorner (v\lrcorner \dd\theta) \\
        =& w\lrcorner(v\lrcorner F_{\bbC}) +  r\cdot(w\lrcorner(v\lrcorner F_{\bbC})) + w\lrcorner(v\lrcorner( - r\cdot F_{\bbC} + \dd\theta)) \\
        \overset{!}{=}& \zeta + r\cdot \zeta
\end{align}
For this to be the case, we need $\dd\theta - r\cdot F_{\bbC} \in \Gamma(\mathcal{F}^{4}_{1})$. This is equivalent to the statement that
\begin{equation}
    \pi_{1}(\dd\theta - r\cdot F_{\bbC}) = 0 \qquad \pi_{k}:\ext^{n}T^{*} \longrightarrow \ext^{n}T^{*}\quotient \mathcal{F}^{n}_{k}
\end{equation}
where we have introduced the projection operator $\pi_{k}$ as defined above for definiteness.

Now let's set $v,\mu = 0$. We have
\begin{align}
    &-w\lrcorner\dd\lambda - (r\cdot w)\lrcorner \dd\lambda - w\lrcorner \dd(r\cdot \lambda) + w\lrcorner((\trivector\lrcorner\lambda)\lrcorner F_{\bbC}) \\
    \begin{split}
        =& -w\lrcorner\dd\lambda - w^{\alpha}(r^{c}{}_{\alpha}\del_{c}\lambda_{ab} - \lambda_{bc}\del_{a}r^{c}{}_{\alpha} - \lambda_{ca}\del_{b}r^{c}{}_{\alpha}) \\
        & - r\cdot(w\lrcorner\dd\lambda) - (w\lrcorner\dd_{\Delta} r)\cdot \lambda + w\lrcorner((\trivector\lrcorner\lambda)\lrcorner F_{\bbC})
    \end{split} \\
    \begin{split}
        =&-w\lrcorner\dd\lambda - w^{\alpha}(r^{c}{}_{\alpha}\del_{c}\lambda_{ab} - \lambda_{bc}\del_{a}r^{c}{}_{\alpha} - \lambda_{ca}\del_{b}r^{c}{}_{\alpha}) \\
        & - r\cdot(w\lrcorner\dd\lambda) - (w\lrcorner(\dd_{\Delta}r - j\trivector\lrcorner j^{2}F_{\bbC}))\cdot\lambda
    \end{split} \\
    \overset{!}{=} & \zeta + r\cdot \zeta
\end{align}
This is implied by the fact that $\dd_{\Delta}r - j\trivector\lrcorner j^{2}F_{\bbC} = 0$.

Next, let's consider when $w = \lambda = 0$. We find
\begin{align}
    & v\lrcorner \dd\mu +(r\cdot v)\lrcorner \dd\mu +v\lrcorner\dd(r\cdot \mu) +\dd((r\cdot v)\lrcorner\mu) + \dd(v\lrcorner(r\cdot \mu)) + (\trivector\lrcorner\mu)\lrcorner(v\lrcorner F_{\bbC}) \\
    \begin{split}
        =& v\lrcorner\dd\mu + v^{\alpha}(r^{c}{}_{\alpha}\del_{c}\mu_{ab} - \mu_{bc}\del_{a}r^{c}{}_{\alpha} - \mu_{ca}\del_{b}r^{c}{}_{\alpha}) \\
        &+ r\cdot(v\lrcorner\dd\mu) + (v\lrcorner\dd_{\Delta}r)\cdot \mu + (\trivector\lrcorner\mu)\lrcorner(v\lrcorner F_{\bbC})
    \end{split} \\
    \begin{split}
        =& v\lrcorner\dd\mu + v^{\alpha}(r^{c}{}_{\alpha}\del_{c}\mu_{ab} - \mu_{bc}\del_{a}r^{c}{}_{\alpha} - \mu_{ca}\del_{b}r^{c}{}_{\alpha}) \\
        &+ r\cdot(v\lrcorner\dd\mu) + (v\lrcorner(\dd_{\Delta}r-j\trivector\lrcorner j^{2}F_{\bbC}))\cdot \mu
    \end{split} \\
    \overset{!}{=} & \zeta + r\cdot \zeta
\end{align}
This is again implied by the fact that $\dd_{\Delta}r - j\trivector\lrcorner j^{2}F_{\bbC} = 0$.

Finally for the 2-forms, we consider the case when $v,w=0$. We find that
\begin{align}
    & \; (\trivector\lrcorner\lambda)\lrcorner \dd\mu - (\trivector\lrcorner \mu)\lrcorner\dd\lambda +\dd((\trivector\lrcorner\lambda)\lrcorner\mu) \\
    \begin{split}
        = &\; 3(\trivector\lrcorner \lambda)^{a}\del_{[a}\mu_{bc]} - 3(\trivector\lrcorner\mu)^{a}\del_{[a}\lambda_{bc]} + 2\del_{[b|}((\trivector\lrcorner\lambda)^{a}\mu_{a|c]} \\
        &\; -((\dd_{\Delta}\trivector)\lrcorner\lambda)\cdot \mu
    \end{split} \\
    \overset{!}{=} &\; \zeta + r\cdot \zeta = \zeta
\end{align}
Note that we can consider $(\dd_{\Delta}\trivector)\lrcorner\lambda$ as an adjoint element and hence it has a natural action on $\mu$. Also, the final equality holds to linear order in $R$ since $\zeta \sim O(R)$ in this case. This case holds because $\dd_{\Delta} \trivector = 0$.

Now we just need to consider the 5-form pieces and show that they vanish. That is, we need
\begin{align}
\begin{split}
    \bigl(\ee^{-\theta-\tau}\cdot L^{F_{\bbC}}_{V'}W' \bigr)|_{\wedge^{5}T^{*}} &= -\mu\wedge \dd\lambda - (r\cdot\mu)\wedge \dd\lambda - \mu\wedge\dd(r\cdot\lambda) \\
    &\qquad + \mu\wedge(v\lrcorner(F_{\bbC} + \dd\theta)) + (r\cdot\mu)\wedge (v\lrcorner F_{\bbC}) \\
    & \qquad + \mu\wedge((r\cdot v + \trivector\lrcorner\lambda)\lrcorner F_{\bbC})
\end{split} \\
    & \overset{!}{=} 0
\end{align}
Let us first set $v = 0$. Then we have
\begin{align}
    & -\mu\wedge \dd\lambda - (r\cdot\mu)\wedge \dd\lambda - \mu\wedge\dd(r\cdot\lambda) +\mu\wedge((\trivector\lrcorner \lambda)\lrcorner F_{\bbC}) \\
    =& -(r\cdot\mu)\wedge \dd_{\Delta}\lambda - \mu\wedge\dd_{\Delta}(r\cdot \lambda) + \mu\wedge((j\trivector\lrcorner j^{2} F_{\bbC})\cdot \lambda) \\
    =& -(r\cdot\mu)\wedge \dd_{\Delta}\lambda -\mu\wedge ((\dd_{\Delta}r)\cdot \lambda) - \mu\wedge r\cdot\dd_{\Delta}\lambda + \mu\wedge ((j\trivector\lrcorner j^{2}F_{\bbC})\cdot \lambda) \\
    =& -r\cdot(\mu\wedge \dd_{\Delta}\lambda) -\mu\wedge((\dd_{\Delta} r - j\trivector\lrcorner j^{2}F_{\bbC})\cdot \lambda) \\
    =& \;0
\end{align}
This holds because any term $\sim \mu\wedge \dd\lambda = 0$ by virtue of the integrability of $\Delta$, and by the fact that $\dd_{\Delta}r - j\trivector\lrcorner j^{2}F_{\bbC} = 0$. Now let's consider instead $\lambda = 0$. We have
\begin{align}
    & \;\mu\wedge(v\lrcorner(F_{\bbC} + \dd\theta)) + (r\cdot \mu)\wedge (v\lrcorner F_{\bbC}) + \mu\wedge((r\cdot v)\lrcorner F_{\bbC}) \\
    =&\; \mu\wedge v\lrcorner \pi_{1}(\dd\theta) + (r\cdot\mu)\wedge (v\lrcorner F_{\bbC}) +\mu\wedge ((r\cdot v)\lrcorner F_{\bbC}) \\
    =&\; \mu\wedge (v\lrcorner (r\cdot F_{\bbC})) + (r\cdot\mu)\wedge (v\lrcorner F_{\bbC}) +\mu\wedge((r\cdot v)\lrcorner F_{\bbC}) \\
    =&\; (r\cdot\mu)\wedge(v\lrcorner F_{\bbC}) + \mu\wedge r\cdot(v\lrcorner F_{\bbC}) \\
    =&\; r\cdot (\mu\wedge(v\lrcorner F_{\bbC})) \\
    =& \; 0
\end{align}
This vanishes because $\mu\wedge(v\lrcorner F_{\bbC}) = 0$ by restrictions on $F_{\bbC}$ imposed by integrability of $L_{1}$.

Hence, we have found the integrability conditions for the deformations, and they are given by
\begin{align}
    0 &= \dd_{\Delta} \trivector \\
    0 &= \dd_{\Delta}r - j\trivector\lrcorner j^{2}F_{\bbC} \\
    0 &= \pi_{1}(\dd\theta - r\cdot F_{\bbC})
\end{align}

\bigskip

Now we need to consider the exactness conditions. These are given by
\begin{equation}
    \tilde{L}'_{1} = (1+L^{F_{\bbC}}_{V})\tilde{L}_{1}
\end{equation}
where $V = -v-\omega -\sigma \in \Gamma((T\quotient \mathfrak{F}^{1}_{0}) \oplus (\ext^{2}T^{*}\quotient \mathcal{F}^{2}_{1}) \oplus \ext^{5}T^{*})$. The minus signs are for convenience. Given $W = w+\mu \in \Gamma(\tilde{L}_{1})$, we have
\begin{align}
\begin{split}
    (1+L^{F_{\bbC}}_{V})W &= w - [v,w] \\
    & \qquad + \mu - \mathcal{L}_{v}\mu + w\lrcorner\dd\omega - w\lrcorner(v\lrcorner F_{\bbC}) \\
    & \qquad + w\lrcorner \dd\sigma + \mu\wedge \dd\omega + w\lrcorner(\omega\wedge F_{\bbC}) - \mu\wedge(v\lrcorner F_{\bbC})
\end{split} \\
\begin{split}
    &= (w - v^{a}\del_{a}w^{\alpha}) + (\dd_{\Delta}v)\cdot w \\
    & \qquad +(\mu - 3v^{a}\del_{[a}\mu_{bc]} - 2\del_{[b|}(v^{a}\mu_{a|c]}) + (\dd_{\Delta}v)\cdot \mu + w\lrcorner(\dd\omega - v\lrcorner F_{\bbC}) \\
    & \qquad +w\lrcorner(\dd\sigma + \omega\wedge F_{\bbC}) + \mu\wedge (\dd\omega - v\lrcorner F_{\bbC})
\end{split} \\
&= \ee^{\pi_{1}(\dd\omega - v\lrcorner F_{\bbC}) + (\dd\sigma + \omega\wedge F_{\bbC})}(1+\dd_{\Delta}v)\cdot(\tilde{w} + \tilde{\mu})
\end{align}
where $\tilde{W} = \tilde{w}+\tilde{\mu} \in \Gamma(\tilde{L}_{1})$, and we have introduced the projection $\pi_{1}:\ext^{n}T^{*}\rightarrow \ext^{n}T^{*}\quotient \mathcal{F}^{n}_{1}$ for definiteness. Hence, the exactness conditions are given by
\begin{align}
    r &= \dd_{\Delta} v \\
    \theta &= \pi_{1}(\dd\omega - v\lrcorner F_{\bbC}) \\
    \tau &= \dd\sigma +\omega\wedge F_{\bbC}
\end{align}
This reproduces the results at the end of section \ref{sec:E6 class 1 Moduli}.

\bigskip

If we assume the flux is trivial, which is not the case for any AdS solution, then we can take the complex twist $\ee^{\alpha+\beta}$ to be globally well-defined. In this case, then it is possible to show using the following deformation parameter
\begin{equation}
    R = \ee^{\alpha+\beta}(r + j\trivector \lrcorner j \alpha + \trivector + \theta + r\cdot \hat{\alpha} - \tfrac{1}{2}j\trivector\lrcorner j\alpha\cdot \hat{\alpha} + \mu - \theta\wedge \hat{\alpha})\ee^{-\alpha-\beta}
\end{equation}
that the cohomology defined by the flux-twisted derivatives above is isomorphic to
\begin{equation}
    H^{0}_{\Delta}(M,\ext^{3}T\quotient \mathfrak{F}^{3}_{2}) \oplus H^{1}_{\Delta}(M,T\quotient \mathfrak{F}^{1}_{0}) \oplus H^{3}_{\mathcal{F}_{1}}(M)\oplus H^{6}_{\dd}(M)
\end{equation}
We have had to assume something slightly stronger about the integrability conditions to show this. Namely that there exists some $\hat{\alpha}\in \Gamma(\mathcal{F}^{3}_{1})$ such that
\begin{equation}
    \dd\alpha = \dd\hat{\alpha}
\end{equation}
This is not directly implied by the involutivity conditions (which just states that $\dd\alpha\in \Gamma(\mathcal{F}^{4}_{1})$) but it may be implied by the vanishing of the moment map.

\subsection{Generalised \texorpdfstring{$\del\delb$}{del-delbar} Moduli}

Following the notation from section \ref{sec:E6 generic moduli}, the moduli of generic $\SUs{6}$ structures is counted by the cohomology of the following complex.
\begin{equation}\label{eq:E6 generic complex}
\begin{tikzcd}[column sep=huge, row sep = huge]
\Gamma(\ext^{2}\mathfrak{X}^{*})  \arrow[r, "D_{-1}"] & \Gamma(\ext^{3}\mathfrak{X}^{*}_{-1}) \arrow[r, "D_{-1}"] & \Gamma(\ext^{4}\mathfrak{X}^{*}_{-2}) \\
\Gamma(\ext^{5}\mathfrak{X}^{*}_{-1}) \arrow[r,"D_{-1}"] \arrow[ur,"D_{0}"] & \Gamma(\ext^{6}\mathfrak{X}^{*}_{-2}) \arrow[ur,"D_{0}"] &
\end{tikzcd}
\end{equation}
Here $D_{0}$ and $D_{-1}$ are operators coming from any torsion free $\USp{6}$ connection\footnote{We will always assume that we are deforming around a full supergravity background} decomposed into $\SUs{6}$ representations.
\begin{equation}
    D = D_{1} + D_{0} + D_{-1}
\end{equation}
In this section, we give give the cohomology of the complex above in terms of the cohomology of $D_{-1}$ provided the background satisfies the generalised $\del\delb$-lemma.

\begin{defn}
$D_{0}$, $D_{-1}$ are said to satisfy the \emph{generalised $\del\delb$-lemma} if the satisfy the following
\begin{equation}
    \image D_{0}\cap \ker D_{-1} \subseteq \image D_{-1}D_{0}
\end{equation}
\end{defn}

\noindent With this we can prove the following result.

\begin{propn}
If a background satisfies the generalised $\del\delb$-lemma, and $D_{0}$ defines a chain homomorphism $D_{0}:\Gamma(\ext^{\bullet}\mathfrak{X}^{*}_{\bullet})\rightarrow \Gamma(\ext^{\bullet - 2}\mathfrak{X}^{*}_{\bullet})$, then the cohomology of the complex \eqref{eq:E6 generic complex} is given by
\begin{equation}
    H^{3}_{D_{-1}}\oplus H^{6}_{D_{-1}}
\end{equation}
where $H^{p}_{D_{-1}}$ is the $p^{\text{th}}$ cohomology of the differential $D_{-1}$
\end{propn}

\begin{proof}
The cohomology of the complex \eqref{eq:E6 generic complex} is given by\footnote{The factor of $\tfrac{1}{2}$ in the quotient is due to the precise form of the projection $D\times_{\ad}V$ in $\SUs{6}$ indices.}
\begin{equation}
    \mathcal{H} = \frac{\{ A + B\in \Gamma(\ext^{3}\mathfrak{X}^{*}_{-1} \oplus \ext^{6}\mathfrak{X}^{*}_{-2}) \, | \, D_{-1}A + D_{0}B = 0 \}}{\{ A = D_{-1}C + D_{0}E,\, B = \tfrac{1}{2}D_{-1}E\,|\, C\in \Gamma(\ext^{2}\mathfrak{X}^{*}_{0}), \, E\in \Gamma(\ext^{5}\mathfrak{X}^{*}_{-1}) \}}
\end{equation}
Let us define a new quotient group by
\begin{equation}
    \mathcal{K} = \frac{\{ B\in \Gamma(\ext^{6}\mathfrak{X}^{*}_{-2}) \, | \, D_{0}B = 0 \}}{\{ B = D_{-1}E \, |\, E\in \Gamma(\ext^{5}\mathfrak{X}^{*}_{-1}),\, D_{0}E = 0 \}}
\end{equation}
and two maps
\begin{equation}
    \begin{array}{rclcrcl}
    \theta:H^{3}_{D_{-}}\oplus \mathcal{K} & \longrightarrow & \mathcal{H} & \qquad & \psi:\mathcal{H}&\longrightarrow & H^{3}_{D_{-}}\oplus \mathcal{K} \\
    \left[A\right]_{3}+[B]_{\mathcal{K}} & \longmapsto & [A+B]_{\mathcal{H}} & & [A+B]_{\mathcal{H}} & \longmapsto & [\tilde{A}]_{3} + [\tilde{B}]_{\mathcal{K}}
    \end{array}
\end{equation}
where $A, \tilde{A}\in \Gamma(\ext^{3}\mathfrak{X}^{*}_{-1})$, $B,\tilde{B} \in \Gamma(\ext^{6}\mathfrak{X}^{6}_{-2})$ and where the subscript denotes the cohomology group that class is a member of. $\tilde{A},\tilde{B}$ are defined from $[A+B]_{\mathcal{H}}$ in the following way. We have
\begin{align}
    0 &= D_{-}A + D_{0}B \\
    \Rightarrow \quad 0 &= D_{-}D_{0}B
\end{align}
So, using the generalised $\del\delb$-lemma we can write $D_{0}B = D_{-}D_{0}E$ for some $E\in \Gamma(\ext^{5}\mathfrak{X}^{*}_{-1})$. We then define
\begin{equation}
    \tilde{A} = A + D_{0}E \qquad \tilde{B} = B + \tfrac{1}{2}D_{-}E
\end{equation}
We need to check that these do define elements of $H^{3}_{D_{-}}$ and $\mathcal{K}$ respectively, and if the map $\psi$ is well defined. Firstly, we note that
\begin{equation}
\arraycolsep = 1.4pt
    \begin{array}{rclcrcl}
    D_{-}\tilde{A} &=& D_{-}A + D_{-}D_{0}E & \qquad & D_{0}\tilde{B} &=& D_{0}B + \tfrac{1}{2}D_{0}D_{-}E \\
    &=& D_{-}A + D_{0}B & & &=& D_{0}B - D_{-}D_{0}E \\
    &=& 0 & & &=& D_{0}B - D_{0}B \\
    &&&&&=&0
    \end{array}
\end{equation}
This shows that $[\tilde{A}]_{3}\in H^{3}_{D_{-}}$ and $[\tilde{B}]_{\mathcal{K}} \in \mathcal{K}$. Note here we have used the fact that, when evaluated on $\Gamma(\ext^{5}\mathfrak{X}^{*}_{-})$
\begin{equation}
    D_{-}D_{0} + \tfrac{1}{2}D_{0}D_{-} = 0
\end{equation}
which follows from the complex \eqref{eq:E6 generic complex}. The factor of $\tfrac{1}{2}$ comes from the way the Dorfman derivative acts. Now suppose that $[A=B]_{\mathcal{H}} = [A'+B']_{\mathcal{H}}$. Then, there exists $c\in \Gamma(\ext^{2}\mathfrak{X}^{*}_{0})$, $e\in \Gamma(\ext^{5}\mathfrak{X}^{*}_{-})$ such that
\begin{equation}
    A' = D_{-}c + D_{0}e \qquad B' = \tfrac{1}{2}D_{-}e
\end{equation}
From these, we define $E'$ such that $D_{0}B' = D_{-}D_{0}E'$. It is a simple check to see that we can choose $E' = E - e$. Then we have
\begin{equation}
\arraycolsep = 1.4pt
    \begin{array}{rclcrcl}
    \tilde{A}' &=& A' + D_{0}E' & \qquad & \tilde{B}' &=& B' + \tfrac{1}{2}D_{-}E' \\
    &=& A + D_{-}c + D_{0}e + D_{0}(E-e) & & &=& B + \tfrac{1}{2}D_{-}e + \tfrac{1}{2}D_{-}(E-e) \\
    &=& A+D_{0}E + D_{-}c &&&=& B + \tfrac{1}{2}D_{-}E \\
    &=& \tilde{A} + D_{-}c &&&=& \tilde{B}
    \end{array}
\end{equation}
Hence we see that
\begin{equation}
    [A+B]_{\mathcal{H}} = [A'+B']_{\mathcal{H}} \qquad \Rightarrow \qquad [\tilde{A}]_{3} = [\tilde{A}']_{3} \quad [\tilde{B}]_{\mathcal{K}} = [\tilde{B}']_{\mathcal{K}}
\end{equation}
Finally, since $E$ as defined above is not unique, we need to check that the map does not depend on the choice. Indeed, suppose
\begin{equation}
    D_{0}B = D_{-}D_{0}E = D_{-}D_{0}E' \qquad \Rightarrow \qquad D_{-}D_{0}(E-E') = 0
\end{equation}
Using the generalised $\del\delb$-lemma again, we can write $D_{0}(E'-E) = D_{-}D_{0}F$ for some $F\in \Gamma(\ext^{4}\mathfrak{X}^{*}_{0})$. Then we have
\begin{equation}
\arraycolsep = 1.4pt
    \begin{array}{rclcrcl}
    \tilde{A}' &=& A + D_{0}E' & \qquad & \tilde{B}' &=& B + \tfrac{1}{2}D_{-}E' \\
    &=& A + D_{0}E + D_{0}(E'-E) & & &=& B + \tfrac{1}{2}D_{-}E + \tfrac{1}{2}D_{-}(E-E') \\
    &=& A + D_{0}E +D_{-}D_{0}F & & &=& \tilde{B} + D_{-}e \\
    &=& \tilde{A} + D_{-}c & & & & 
    \end{array}
\end{equation}
where $c = D_{0}F \in \Gamma(\ext^{2}\mathcal{X}^{*}_{0})$, and $e = \tfrac{1}{2}(E-E' + D_{-}F) \in \Gamma(\ext^{5}\mathfrak{X}^{*}_{-1})$ is such that $D_{0}e = 0$. Hence we have
\begin{equation}
    D_{-}D_{0}E = D_{-}D_{0}E' \qquad \Rightarrow \qquad [\tilde{A}']_{3} = [\tilde{A}]_{3} \quad [\tilde{B}']_{\mathcal{K}} = [\tilde{B}]_{\mathcal{K}}
\end{equation}
Hence, the map $\psi$ is well defined. It is a simple check to see that $\theta$ is also well defined.

Now we show that $\psi,\theta$ are inverses of each other. Firstly,
\begin{align}
    \theta\circ\psi([A+B]_{\mathcal{H}}) &= \theta([\tilde{A}]_{3} + [\tilde{B}]_{\mathcal{K}}) \\
    &= [\tilde{A} + \tilde{B}]_{\mathcal{H}} \\
    &= [A + D_{0}E + B + \tfrac{1}{2}D_{-}E]_{\mathcal{H}} \\
    &= [A + B]
\end{align}
Therefore, $\theta\circ\psi = \mathbb{I}_{\mathcal{H}}$. Next consider,
\begin{align}
    \psi\circ \theta([A]_{3} + [B]_{\mathcal{K}}) &= \psi([A+B]_{\mathcal{H}}) \\
    &= [\tilde{A}]_{3} + [\tilde{B}]_{\mathcal{K}} \\
    &= [ A + D_{0}E]_{3} + [B + \tfrac{1}{2}D_{-}E]
\end{align}
But since $D_{0}B = 0$ by assumption, we can choose $E=0$. Hence,
\begin{equation}
    \psi\circ\theta([A]_{3} + [B]_{\mathcal{K}} ) = [A]_{3} + [B]_{\mathcal{K}} 
\end{equation}
So $\psi\circ\theta = \mathbb{I}_{H^{3}\oplus \mathcal{K}}$ and hence $\psi = \theta^{-1}$. Clearly $\theta$ and $\psi$ are homomorphisms. Hence,
\begin{equation}
    \mathcal{H} \cong H^{3}_{D_{-}}\oplus \mathcal{K}
\end{equation}

Now we want to show that $\mathcal{K}\cong H^{6}_{D_{-}}$. Again, let's define some maps
\begin{equation}
    \begin{array}{rclcrcl}
       \eta : \mathcal{K} & \longrightarrow & H^{6}_{D_{-}} & \qquad & \zeta : H^{6}_{D_{-}} & \longrightarrow & \mathcal{K} \\
       \left[B\right]_{\mathcal{K}} & \longmapsto & [B]_{6} & & [B]_{6} & \longmapsto & [\tilde{B}]_{\mathcal{K}}
    \end{array}
\end{equation}
where $\tilde{B}$ is defined by the following. For any $B\in \Gamma(\ext^{6}\mathfrak{X}^{*}_{-2})$ we have $D_{-}B = 0$. But we also assume that $D_{0}$ is a chain homomorphism, meaning that $D_{-}(D_{0}B) = 0$. Hence, using the generalised $\del\delb$-lemma, we can define an $E$ such that
\begin{equation}
    D_{0}B = D_{-}D_{0}E
\end{equation}
We therefore define $\tilde{B}$ as
\begin{equation}
    \tilde{B} = B + \tfrac{1}{2}D_{-}E \qquad \Rightarrow \qquad D_{0}\tilde{B} = D_{0}B + \tfrac{1}{2}D_{0}D_{-}E = 0
\end{equation}
A similar proof as above shows that these maps are well defined and are inverses of each other. Hence we have
\begin{equation}
    \mathcal{H} \cong H^{3}_{D_{-}} \oplus \mathcal{K} \cong H^{3}_{D_{-}} \oplus H^{6}_{D_{-}}
\end{equation}

\end{proof}

\subsection{Calabi--Yau Moduli}

Here we will show that the Calabi--Yau satisfies the generalised $\del\delb$-lemma and hence we can calculate its moduli using the formula above. The proof involves using a compact Calabi--Yau but the result holds more generally as one can calculate the moduli using a type 0 presentation of the ECS instead.

The ECS for a Calabi--Yau is
\begin{equation}
    \JUs = \frac{1}{2}\left( I - \vol - \vol^{\#} \right) \qquad L_{1} = \ee^{\ii\vol}\cdot(T^{1,0}\oplus \ext^{0,2}T^{*})
\end{equation}
Using the adjoint action of $\JUs$, we can decompose $E_{\bbC}$ and $\ad\tilde{F}_{\bbC}$ into eigenbundles
\begin{equation}
    E_{\bbC} = L_{1}\oplus L_{0} \oplus L_{-1} \qquad \ad\tilde{F}_{\bbC} = \ad P_{\bbR^{\pl}\times \Us{6}} \oplus S_{1}\oplus S_{-1} \oplus S_{2}\oplus S_{-2}
\end{equation}
The eigenbundles needed for the deformation problem laid out in the previous section are given explicitly by
\begin{align}
    \ext^{5}\mathfrak{X}^{*}_{-1} = L_{-} &= \left\{ \begin{array}{c|c}
        \bar{w} - \ii\bar{w}\lrcorner\vol & \bar{w}\in T^{0,1} \\
        \omega & \omega \in \ext^{2,0}T^{*}
    \end{array} \right\} \\
    \nonumber \\
    \ext^{2}\mathfrak{X}^{*}_{0} = L_{0} &= \left\{ \begin{array}{c|c}
        v - \ii v\lrcorner\vol & v\in T^{1,0}  \\
        \bar{v} + \ii\bar{v}\lrcorner \vol & \bar{v}\in T^{0,1} \\
        \theta & \theta \in \ext^{1,1}T^{*}
    \end{array} \right\} \\
    \nonumber \\
    \ext^{3}\mathfrak{X}^{*}_{-1} = S_{-1} &= \left\{ \begin{array}{c|c}
        \alpha(\tfrac{2}{3} + \tfrac{1}{3}\mathbb{I} +\ii\vol -\ii\vol^{\#}) & \alpha\in\bbC \\
        r & r\in T^{0,1}\otimes T^{*1,0} \\
        \beta+\ii\vol^{\#}\lrcorner \beta & \beta \in \ext^{2,1}T^{*} \\
        \gamma - \ii\vol^{\#}\lrcorner \gamma & \gamma \in \ext^{3,0}T^{*}
    \end{array} \right\} \\
    \nonumber \\
    \ext^{6}\mathfrak{X}^{*}_{-2} = S_{-2} &= \left\{ \begin{array}{c|c}
        \lambda + \ii\vol^{\#}\lrcorner \lambda & \lambda \in \ext^{3,0}T^{*} 
    \end{array} \right\}
\end{align}
Using the holomorphic 3-form $\Omega$ of the Calabi--Yau, we can define a chain isomorphism $\mathfrak{X}^{*} \simeq T^{*}$. Indeed, we have
\begin{align}
    \ext^{5}\mathfrak{X}^{*}_{-1} &\rightarrow \left\{ \begin{array}{c}
        \bar{w}\lrcorner\vol \in \ext^{3,2}T^{*} \\
        \omega \wedge \bar{\Omega}\in \ext^{2,3}T^{*}
    \end{array} \right\} \sim \ext^{5}T^{*} \\
    \nonumber \\
    \ext^{2}\mathfrak{X}^{*}_{0} & \rightarrow \left\{ \begin{array}{c}
        v\lrcorner \Omega \in \ext^{2,0}T^{*} \\
        \bar{v}\lrcorner \Omega \in \ext^{0,2}T^{*} \\
        \theta \in \ext^{1,1}T^{*}
    \end{array}\right\} \sim \ext^{2}T^{*} \\
    \nonumber \\
    \ext^{3}\mathfrak{X}^{*}_{-1} & \rightarrow \left\{\begin{array}{c}
    \alpha\Omega \in \ext^{3,0}T^{*} \\
    r\cdot \bar{\Omega} \in \ext^{1,2}T^{*} \\
    \beta \in \ext^{2,1}T^{*} \\
    (\bar{\Omega}^{\#}\lrcorner \gamma)\bar{\Omega} \in \ext^{0,3}T^{*} 
    \end{array}\right\} \sim \ext^{3}T^{*} \\
    \nonumber \\
    \ext^{6}\mathfrak{X}^{*}_{-2} & \rightarrow \left\{ \begin{array}{c}
    \lambda \wedge \bar{\Omega} \in \ext^{3,3}T^{*}
    \end{array} \right\} \sim \ext^{6}T^{*}
\end{align}
We can also take the torsion free compatible connection $\nabla$, and lift it to a generalised connection $D$ as in \cite{Coimbra:2011ky}. With this lift, and with the isomorphism above we find
\begin{equation}
    D_{-} \rightarrow \del \qquad D_{0} \rightarrow \Omega^{\#}\lrcorner \delb + \bar{\Omega}^{\#}\lrcorner \del
\end{equation}
where here $\del,\delb$ denote the projection of $\nabla$ onto the $T^{*1,0}, T^{*0,1}$ piece respectively.

We need to show that these operators satisfy the generalised $\del\delb$-lemma. We just need to show this for elements in $\ext^{5}T{*}$ and $\ext^{6}T^{*}$ for the proof to hold.

\begin{proof}
First take $\alpha\in \ext^{2,3}T^{*}$. Then we have
\begin{equation}
    D_{0}\alpha = (\Omega^{\#}\lrcorner\delb)\lrcorner\alpha + (\bar{\Omega}^{\#}\lrcorner \del)\alpha
\end{equation}
We can consider only the second term which is just $\bar{\Omega}^{\#}\lrcorner(\del\alpha)$. Suppose further that $D_{0}\alpha \in \ker D_{-} \sim \ker \del$. Then each term individually has to be in $\ker \del$. Since $h^{0,3} = 1$, we must have that, up to $\del$ exact terms
\begin{equation}
    \bar{\Omega}^{\#}\lrcorner(\del\alpha) = c\bar{\Omega} \qquad \Rightarrow \qquad \del\alpha = \tilde{c}\vol
\end{equation}
for some constants $c,\tilde{c}$. However, $\vol$ is not $\del$-exact and hence we must have $c= \tilde{c} = 0$. Therefore, $\del\alpha = 0$ and so $[\alpha]\in H^{2,3}_{\del} = 0$. Therefore, $\alpha$ is $\del$-exact and so $D_{0}\alpha= D_{0}D_{-}a \sim D_{-}D_{0}a$ for some $a\in \ext^{4}T^{*}$.

Now take $\beta \in \ext^{3,2}T^{*}$. Here we automatically have $\del\beta = 0$ and so $[\beta] \in H^{3,2}_{\del} = 0$. Therefore $\beta$ is $\del$-exact and so $D_{0}\beta = D_{0}D_{-}b \sim D_{-}D_{0} b$ for some $b \in \ext^{4}T^{*}$.

Finally, we take $\gamma\in \ext^{3,3}T^{*}$ and write this as $\gamma = c\vol + \del \psi$ for some constant $c$ and some $\psi \in \ext^{2,3}T^{*}$. For any constant, we have $D_{0}(c\vol) = 0$ since $D_{0}$ is built from the compatible connection $\nabla$. Therefore, we have
\begin{equation}
    D_{0}\gamma = D_{0}\del\psi = D_{0}D_{-}\psi \sim D_{-}D_{0}\psi
\end{equation}
This gives the result.
\end{proof}

Using the results of the previous section on the moduli of a background satisfying the generalised $\del\delb$-lemma, we see that the moduli of the Calabi--Yau are given by
\begin{equation}
    \mathcal{H} = H^{3}_{D_{-}}\oplus H^{6}_{D_{-}} \cong H^{3}_{\del}\oplus H^{6}_{\del}
\end{equation}
Note that, since $H^{p}_{\del} \cong H^{p}_{\dd}$ for a Calabi--Yau manifold, this agrees with the result obtained for the moduli of the Calabi--Yau calculated through a type 0 ECS, as discussed in section \ref{sec:class 0 moduli}.

\bibliographystyle{JHEP}
\bibliography{citations}

\providecommand{\href}[2]{#2}\begingroup\raggedright\begin{thebibliography}{10}

\bibitem{Hitchin02}
N.~Hitchin, \emph{Generalized {C}alabi--{Y}au manifolds},
  \href{https://doi.org/10.1093/qjmath/54.3.281}{\emph{Quart. J. Math.}
  {\bfseries 54} (2003) 281}
  [\href{https://arxiv.org/abs/math/0209099}{{\ttfamily math/0209099}}].

\bibitem{Gualtieri04}
M.~Gualtieri, \emph{Generalized complex geometry},
  \href{https://arxiv.org/abs/math/0401221}{{\ttfamily math/0401221}}.

\bibitem{Grana:2004bg}
M.~Graña, R.~Minasian, M.~Petrini and A.~Tomasiello, \emph{{Supersymmetric
  backgrounds from generalized Calabi-Yau manifolds}},
  \href{https://doi.org/10.1088/1126-6708/2004/08/046}{\emph{JHEP} {\bfseries
  08} (2004) 046} [\href{https://arxiv.org/abs/hep-th/0406137}{{\ttfamily
  hep-th/0406137}}].

\bibitem{Pacheco:2008ps}
P.~Pires~Pacheco and D.~Waldram, \emph{{M-theory, exceptional generalised
  geometry and superpotentials}},
  \href{https://doi.org/10.1088/1126-6708/2008/09/123}{\emph{JHEP} {\bfseries
  09} (2008) 123} [\href{https://arxiv.org/abs/0804.1362}{{\ttfamily
  0804.1362}}].

\bibitem{Grana:2011nb}
M.~Graña and F.~Orsi, \emph{{N=1 vacua in Exceptional Generalized Geometry}},
  \href{https://doi.org/10.1007/JHEP08(2011)109}{\emph{JHEP} {\bfseries 08}
  (2011) 109} [\href{https://arxiv.org/abs/1105.4855}{{\ttfamily 1105.4855}}].

\bibitem{Coimbra:2014uxa}
A.~Coimbra, C.~Strickland-Constable and D.~Waldram, \emph{{Supersymmetric
  Backgrounds and Generalised Special Holonomy}},
  \href{https://doi.org/10.1088/0264-9381/33/12/125026}{\emph{Class. Quant.
  Grav.} {\bfseries 33} (2016) 125026}
  [\href{https://arxiv.org/abs/1411.5721}{{\ttfamily 1411.5721}}].

\bibitem{Ashmore:2015joa}
A.~Ashmore and D.~Waldram, \emph{{Exceptional Calabi-Yau spaces: the geometry
  of $\mathcal{N}=2$ backgrounds with flux}},
  \href{https://doi.org/10.1002/prop.201600109}{\emph{Fortsch. Phys.}
  {\bfseries 65} (2017) 1600109}
  [\href{https://arxiv.org/abs/1510.00022}{{\ttfamily 1510.00022}}].

\bibitem{Coimbra:2016ydd}
A.~Coimbra and C.~Strickland-Constable, \emph{{Supersymmetric Backgrounds, the
  Killing Superalgebra, and Generalised Special Holonomy}},
  \href{https://doi.org/10.1007/JHEP11(2016)063}{\emph{JHEP} {\bfseries 11}
  (2016) 063} [\href{https://arxiv.org/abs/1606.09304}{{\ttfamily
  1606.09304}}].

\bibitem{Grana:2016dyl}
M.~Gra\~na and P.~Ntokos, \emph{{Generalized geometric vacua with eight
  supercharges}}, \href{https://doi.org/10.1007/JHEP08(2016)107}{\emph{JHEP}
  {\bfseries 08} (2016) 107}
  [\href{https://arxiv.org/abs/1605.06383}{{\ttfamily 1605.06383}}].

\bibitem{Ashmore:2016qvs}
A.~Ashmore, M.~Petrini and D.~Waldram, \emph{{The exceptional generalised
  geometry of supersymmetric AdS flux backgrounds}},
  \href{https://doi.org/10.1007/JHEP12(2016)146}{\emph{JHEP} {\bfseries 12}
  (2016) 146} [\href{https://arxiv.org/abs/1602.02158}{{\ttfamily
  1602.02158}}].

\bibitem{Gauntlett:2002sc}
J.~P. Gauntlett, D.~Martelli, S.~Pakis and D.~Waldram, \emph{{G structures and
  wrapped NS5-branes}},
  \href{https://doi.org/10.1007/s00220-004-1066-y}{\emph{Commun. Math. Phys.}
  {\bfseries 247} (2004) 421}
  [\href{https://arxiv.org/abs/hep-th/0205050}{{\ttfamily hep-th/0205050}}].

\bibitem{Gauntlett:2002fz}
J.~P. Gauntlett and S.~Pakis, \emph{{The Geometry of D = 11 killing spinors}},
  \href{https://doi.org/10.1088/1126-6708/2003/04/039}{\emph{JHEP} {\bfseries
  04} (2003) 039} [\href{https://arxiv.org/abs/hep-th/0212008}{{\ttfamily
  hep-th/0212008}}].

\bibitem{Gauntlett:2002nw}
J.~P. Gauntlett, J.~B. Gutowski, C.~M. Hull, S.~Pakis and H.~S. Reall,
  \emph{{All supersymmetric solutions of minimal supergravity in five-
  dimensions}}, \href{https://doi.org/10.1088/0264-9381/20/21/005}{\emph{Class.
  Quant. Grav.} {\bfseries 20} (2003) 4587}
  [\href{https://arxiv.org/abs/hep-th/0209114}{{\ttfamily hep-th/0209114}}].

\bibitem{Lukas:1998yy}
A.~Lukas, B.~A. Ovrut, K.~S. Stelle and D.~Waldram, \emph{{The Universe as a
  domain wall}}, \href{https://doi.org/10.1103/PhysRevD.59.086001}{\emph{Phys.
  Rev. D} {\bfseries 59} (1999) 086001}
  [\href{https://arxiv.org/abs/hep-th/9803235}{{\ttfamily hep-th/9803235}}].

\bibitem{Lukas:1998tt}
A.~Lukas, B.~A. Ovrut, K.~S. Stelle and D.~Waldram, \emph{{Heterotic M theory
  in five-dimensions}},
  \href{https://doi.org/10.1016/S0550-3213(99)00196-0}{\emph{Nucl. Phys. B}
  {\bfseries 552} (1999) 246}
  [\href{https://arxiv.org/abs/hep-th/9806051}{{\ttfamily hep-th/9806051}}].

\bibitem{Morrison:1996xf}
D.~R. Morrison and N.~Seiberg, \emph{{Extremal transitions and five-dimensional
  supersymmetric field theories}},
  \href{https://doi.org/10.1016/S0550-3213(96)00592-5}{\emph{Nucl. Phys. B}
  {\bfseries 483} (1997) 229}
  [\href{https://arxiv.org/abs/hep-th/9609070}{{\ttfamily hep-th/9609070}}].

\bibitem{Douglas:1996xp}
M.~R. Douglas, S.~H. Katz and C.~Vafa, \emph{{Small instantons, Del Pezzo
  surfaces and type I-prime theory}},
  \href{https://doi.org/10.1016/S0550-3213(97)00281-2}{\emph{Nucl. Phys. B}
  {\bfseries 497} (1997) 155}
  [\href{https://arxiv.org/abs/hep-th/9609071}{{\ttfamily hep-th/9609071}}].

\bibitem{Intriligator:1997pq}
K.~A. Intriligator, D.~R. Morrison and N.~Seiberg, \emph{{Five-dimensional
  supersymmetric gauge theories and degenerations of Calabi-Yau spaces}},
  \href{https://doi.org/10.1016/S0550-3213(97)00279-4}{\emph{Nucl. Phys. B}
  {\bfseries 497} (1997) 56}
  [\href{https://arxiv.org/abs/hep-th/9702198}{{\ttfamily hep-th/9702198}}].

\bibitem{Gopakumar:1998ii}
R.~Gopakumar and C.~Vafa, \emph{{M theory and topological strings. 1.}},
  \href{https://arxiv.org/abs/hep-th/9809187}{{\ttfamily hep-th/9809187}}.

\bibitem{Gopakumar:1998jq}
R.~Gopakumar and C.~Vafa, \emph{{M theory and topological strings. 2.}},
  \href{https://arxiv.org/abs/hep-th/9812127}{{\ttfamily hep-th/9812127}}.

\bibitem{Gauntlett:2004zh}
J.~P. Gauntlett, D.~Martelli, J.~Sparks and D.~Waldram, \emph{{Supersymmetric
  AdS(5) solutions of M theory}},
  \href{https://doi.org/10.1088/0264-9381/21/18/005}{\emph{Class. Quant. Grav.}
  {\bfseries 21} (2004) 4335}
  [\href{https://arxiv.org/abs/hep-th/0402153}{{\ttfamily hep-th/0402153}}].

\bibitem{Hull:2007zu}
C.~Hull, \emph{{Generalised Geometry for M-Theory}},
  \href{https://doi.org/10.1088/1126-6708/2007/07/079}{\emph{JHEP} {\bfseries
  07} (2007) 079} [\href{https://arxiv.org/abs/hep-th/0701203}{{\ttfamily
  hep-th/0701203}}].

\bibitem{Coimbra:2011ky}
A.~Coimbra, C.~Strickland-Constable and D.~Waldram, \emph{{$E_{d(d)} \times
  \mathbb{R}^+$ generalised geometry, connections and M theory}},
  \href{https://doi.org/10.1007/JHEP02(2014)054}{\emph{JHEP} {\bfseries 02}
  (2014) 054} [\href{https://arxiv.org/abs/1112.3989}{{\ttfamily 1112.3989}}].

\bibitem{Ashmore:2019qii}
A.~Ashmore, C.~Strickland-Constable, D.~Tennyson and D.~Waldram,
  \emph{{Generalising G$_\text{2}$ geometry: involutivity, moment maps and
  moduli}}, \href{https://doi.org/10.1007/JHEP01(2021)158}{\emph{JHEP}
  {\bfseries 01} (2021) 158}
  [\href{https://arxiv.org/abs/1910.04795}{{\ttfamily 1910.04795}}].

\bibitem{Ashmore:2019rkx}
A.~Ashmore, C.~Strickland-Constable, D.~Tennyson and D.~Waldram,
  \emph{{Heterotic backgrounds via generalised geometry: moment maps and
  moduli}},  \href{https://arxiv.org/abs/1912.09981}{{\ttfamily 1912.09981}}.

\bibitem{Coimbra:2012af}
A.~Coimbra, C.~Strickland-Constable and D.~Waldram, \emph{{Supergravity as
  Generalised Geometry II: $E_{d(d)} \times \mathbb{R}^+$ and M theory}},
  \href{https://doi.org/10.1007/JHEP03(2014)019}{\emph{JHEP} {\bfseries 03}
  (2014) 019} [\href{https://arxiv.org/abs/1212.1586}{{\ttfamily 1212.1586}}].

\bibitem{Ashmore:2016oug}
A.~Ashmore, M.~Gabella, M.~Graña, M.~Petrini and D.~Waldram, \emph{{Exactly
  marginal deformations from exceptional generalised geometry}},
  \href{https://doi.org/10.1007/JHEP01(2017)124}{\emph{JHEP} {\bfseries 01}
  (2017) 124} [\href{https://arxiv.org/abs/1605.05730}{{\ttfamily
  1605.05730}}].

\bibitem{Coimbra:2017fqv}
A.~Coimbra and C.~Strickland-Constable, \emph{{Supersymmetric AdS backgrounds
  and weak generalised holonomy}},
  \href{https://arxiv.org/abs/1710.04156}{{\ttfamily 1710.04156}}.

\bibitem{Grana:2009im}
M.~Graña, J.~Louis, A.~Sim and D.~Waldram, \emph{{$E_{7(7)}$ formulation of
  $N=2$ backgrounds}},
  \href{https://doi.org/10.1088/1126-6708/2009/07/104}{\emph{JHEP} {\bfseries
  07} (2009) 104} [\href{https://arxiv.org/abs/0904.2333}{{\ttfamily
  0904.2333}}].

\bibitem{Hitchin00}
N.~Hitchin, \emph{The geometry of three-forms in six dimensions}, {\emph{J.
  Diff. Geom.} {\bfseries 55} (2000) 547}.

\bibitem{Hitchin01}
N.~Hitchin, \emph{Stable forms and special metrics},
  \href{https://arxiv.org/abs/math/0107101}{{\ttfamily math/0107101}}.

\bibitem{Maldacena:2000mw}
J.~M. Maldacena and C.~Nunez, \emph{{Supergravity description of field theories
  on curved manifolds and a no go theorem}},
  \href{https://doi.org/10.1142/S0217751X01003937}{\emph{Int. J. Mod. Phys. A}
  {\bfseries 16} (2001) 822}
  [\href{https://arxiv.org/abs/hep-th/0007018}{{\ttfamily hep-th/0007018}}].

\bibitem{Giddings:2001yu}
S.~B. Giddings, S.~Kachru and J.~Polchinski, \emph{{Hierarchies from fluxes in
  string compactifications}},
  \href{https://doi.org/10.1103/PhysRevD.66.106006}{\emph{Phys. Rev. D}
  {\bfseries 66} (2002) 106006}
  [\href{https://arxiv.org/abs/hep-th/0105097}{{\ttfamily hep-th/0105097}}].

\bibitem{Gauntlett:2003cy}
J.~P. Gauntlett, D.~Martelli and D.~Waldram, \emph{{Superstrings with intrinsic
  torsion}}, \href{https://doi.org/10.1103/PhysRevD.69.086002}{\emph{Phys. Rev.
  D} {\bfseries 69} (2004) 086002}
  [\href{https://arxiv.org/abs/hep-th/0302158}{{\ttfamily hep-th/0302158}}].

\bibitem{kodaira2006complex}
K.~Kodaira, \emph{Complex manifolds and deformation of complex structures}.
  Springer, 2006.

\bibitem{Behrndt:2000zh}
K.~Behrndt and S.~Gukov, \emph{{Domain walls and superpotentials from M theory
  on Calabi-Yau three folds}},
  \href{https://doi.org/10.1016/S0550-3213(00)00149-8}{\emph{Nucl. Phys. B}
  {\bfseries 580} (2000) 225}
  [\href{https://arxiv.org/abs/hep-th/0001082}{{\ttfamily hep-th/0001082}}].

\bibitem{GUNAYDIN1985309}
M.~Günaydin, L.~Romans and N.~Warner, \emph{Iib, or not iib: That is the
  question},
  \href{https://doi.org/https://doi.org/10.1016/0370-2693(85)90332-6}{\emph{Physics
  Letters B} {\bfseries 164} (1985) 309 }.

\bibitem{gilmore2008lie}
R.~Gilmore, \emph{Lie groups, physics, and geometry: an introduction for
  physicists, engineers and chemists}. Cambridge University Press, 2008.

\bibitem{Berman:2012vc}
D.~S. Berman, M.~Cederwall, A.~Kleinschmidt and D.~C. Thompson, \emph{{The
  gauge structure of generalised diffeomorphisms}},
  \href{https://doi.org/10.1007/JHEP01(2013)064}{\emph{JHEP} {\bfseries 01}
  (2013) 064} [\href{https://arxiv.org/abs/1208.5884}{{\ttfamily 1208.5884}}].

\bibitem{Hitchin00a}
N.~{Hitchin}, \emph{The geometry of three-forms in six and seven dimensions},
  \href{https://arxiv.org/abs/math/0010054}{{\ttfamily math/0010054}}.

\bibitem{Swann1991}
A.~Swann, \emph{Hyperkähler and quaternionic kähler geometry},
  \href{https://doi.org/10.1007/BF01446581}{\emph{Mathematische Annalen}
  {\bfseries 289} (1991) 421}.

\bibitem{Boyer1998}
C.~P. Boyer and K.~Galicki, \emph{{3 - Sasakian manifolds}}, {\emph{Surveys
  Diff. Geom.} {\bfseries 7} (1999) 123}
  [\href{https://arxiv.org/abs/hep-th/9810250}{{\ttfamily hep-th/9810250}}].

\bibitem{Hitchin:1986ea}
N.~J. Hitchin, A.~Karlhede, U.~Lindstrom and M.~Rocek, \emph{{Hyperkahler
  Metrics and Supersymmetry}},
  \href{https://doi.org/10.1007/BF01214418}{\emph{Commun. Math. Phys.}
  {\bfseries 108} (1987) 535}.

\bibitem{habib2013modified}
G.~Habib and K.~Richardson, \emph{Modified differentials and basic cohomology
  for riemannian foliations}, {\emph{Journal of Geometric Analysis} {\bfseries
  23} (2013) 1314}.

\bibitem{Tian88}
G.~Tian, \emph{Smoothness of the Universal Deformation Space of Compact
  Calabi-Yau Manifolds and Its Peterson-Weil Metric}, pp.~629--646.

\bibitem{1989CMaPh.126..325T}
A.~N. {Todorov}, \emph{{The Weil-Petersson geometry of the moduli space of SU(
  n{\ensuremath{\geqq}}3) (Calabi-Yau) manifolds I}}, {\emph{Communications in
  Mathematical Physics} {\bfseries 126} (1989) 325}.

\bibitem{Ashmore:2018npi}
A.~Ashmore, \emph{{Marginal deformations of 3d $\mathcal{N}=2$ CFTs from
  AdS$_4$ backgrounds in generalised geometry}},
  \href{https://doi.org/10.1007/JHEP12(2018)060}{\emph{JHEP} {\bfseries 12}
  (2018) 060} [\href{https://arxiv.org/abs/1809.03503}{{\ttfamily
  1809.03503}}].

\bibitem{APTW}
A.~Ashmore, M.~Petrini, E.~Tasker and D.~Waldram. To appear.

\bibitem{Pestun:2005rp}
V.~Pestun and E.~Witten, \emph{{The Hitchin functionals and the topological
  B-model at one loop}},
  \href{https://doi.org/10.1007/s11005-005-0007-9}{\emph{Lett. Math. Phys.}
  {\bfseries 74} (2005) 21}
  [\href{https://arxiv.org/abs/hep-th/0503083}{{\ttfamily hep-th/0503083}}].

\bibitem{Strominger:1997eb}
A.~Strominger, \emph{{Loop corrections to the universal hypermultiplet}},
  \href{https://doi.org/10.1016/S0370-2693(98)00015-X}{\emph{Phys. Lett. B}
  {\bfseries 421} (1998) 139}
  [\href{https://arxiv.org/abs/hep-th/9706195}{{\ttfamily hep-th/9706195}}].

\bibitem{Anguelova:2004sj}
L.~Anguelova, M.~Rocek and S.~Vandoren, \emph{{Quantum corrections to the
  universal hypermultiplet and superspace}},
  \href{https://doi.org/10.1103/PhysRevD.70.066001}{\emph{Phys. Rev. D}
  {\bfseries 70} (2004) 066001}
  [\href{https://arxiv.org/abs/hep-th/0402132}{{\ttfamily hep-th/0402132}}].

\bibitem{Antoniadis:2003sw}
I.~Antoniadis, R.~Minasian, S.~Theisen and P.~Vanhove, \emph{{String loop
  corrections to the universal hypermultiplet}},
  \href{https://doi.org/10.1088/0264-9381/20/23/009}{\emph{Class. Quant. Grav.}
  {\bfseries 20} (2003) 5079}
  [\href{https://arxiv.org/abs/hep-th/0307268}{{\ttfamily hep-th/0307268}}].

\end{thebibliography}\endgroup

\end{document}